\newif\ifappendices\appendicestrue   % Uncomment to turn off appendices
\newif\ifdraft\drafttrue   % Uncomment to turn on comments
\documentclass[acmsmall,screen]{acmart}

\binoppenalty=\maxdimen
\relpenalty=\maxdimen
\newcommand{\toolname}{{\sc Burst}\xspace} % Bottom-Up Recusive SynThesis

\newcommand{\set}[1]{\{#1\}}

\usepackage{wrapfig}
\usepackage{stmaryrd}
\usepackage{tikzsymbols}
\usepackage{xspace}
\usepackage{mathpartir}
\usepackage{tabularx}
\usepackage{tabulary}
\usepackage{dsfont}
\usepackage{listings}
\usepackage[listings]{tcolorbox}
\usepackage{letltxmacro}
\usepackage[plain]{algorithm}
\usepackage[noend]{algpseudocode}
\usepackage{etoolbox,xspace}
\usepackage[normalem]{ulem}
\usepackage{amsmath,graphicx}
\usepackage{enumitem}
\usepackage{csvsimple}
\usepackage{pifont}
\usepackage{caption}
\usepackage{subcaption}

\usepackage{amssymb}
\definecolor{dkblue}{rgb}{0,0.1,0.5}
\definecolor{dkcyan}{rgb}{0.1, 0.3, 0.3}
\definecolor{dkgreen}{rgb}{0,0.3,0}
\definecolor{dkred}{rgb}{0.6,0,0}
\definecolor{dkpurple}{rgb}{0.7,0,0.4}
\definecolor{olive}{rgb}{0.4, 0.4, 0.0}
\definecolor{orange}{rgb}{0.9,0.6,0.2}

\definecolor{lightyellow}{RGB}{255, 255, 179}
\definecolor{lightgreen}{RGB}{170, 255, 220}
\definecolor{teal}{RGB}{141,211,199}
\definecolor{darkbrown}{RGB}{121,37,0}
\definecolor{princetonorange}{RGB}{255,143,0}
\definecolor{tmlblue}{RGB}{0,58,120}

\colorlet{listing-comment}{gray}
\colorlet{operator-symbol}{darkbrown}
\lstdefinelanguage{OCaml}{
    language=Caml,
    morekeywords={include, module, sig, struct, val, unit, switch, on},
    morekeywords={inl, inr, fix, fst, snd, unl, unr},
    morekeywords=[2]{false, true, f},
    keywordstyle=[2]\color{dkgreen},
    morekeywords=[3]{int, bool, list, nat},
    keywordstyle=[3]\color{dkcyan},
    literate=%
      {=}{{{\color{operator-symbol}=}}}1
      {<}{{{\color{operator-symbol}<}}}1
      {>}{{{\color{operator-symbol}>}}}1
      {:}{{{\color{operator-symbol}:}}}1
      {;}{{{\color{operator-symbol};}}}1
      {|}{{{\color{operator-symbol}|}}}1
      {[}{{{\color{operator-symbol}[}}}1
      {]}{{{\color{operator-symbol}]}}}1
      {\&}{{{\color{operator-symbol}\&}}}1
      {->}{{{\color{operator-symbol}->}}}1
}

\lstdefinestyle{default}{
    basicstyle=\linespread{0.9}\ttfamily,
    columns=fullflexible,
    commentstyle=\sffamily\color{black!50!white},
    escapechar=\#,
    framexleftmargin=1ex,
    framexrightmargin=1ex,
    keepspaces=true,
    keywordstyle=\color{dkblue},
    mathescape,
    showstringspaces=true,
    stepnumber=1,
    xleftmargin=2.5em,
}

\lstdefinestyle{smallstyle}{
    basicstyle=\footnotesize\ttfamily,
    columns=fullflexible,
    commentstyle=\sffamily\color{black!50!white},
    escapechar=\#,
    framexleftmargin=1ex,
    framexrightmargin=1ex,
    keepspaces=true,
    keywordstyle=\color{dkblue},
    mathescape,
    showstringspaces=true,
    stepnumber=1,
    xleftmargin=2.5em,
}

\newcommand\zlstinline{\let\par\endgraf\lstinline}

\lstnewenvironment{lstlistingsmall}
  {\small\lstset{}}
  {}
  {}

\newtcblisting
    {boxed-listing}
    [2][]
    {listing only,
     arc=2pt,
     boxrule=0.75pt,
     boxsep=0.5pt,
     colback=white,
     left=-40pt,
     right=4pt,
     top=1pt,
     bottom=0pt,
     listing options={
         xleftmargin=4.5em,
         #2,
     },
     #1,
    }
    
\algblock{ForEach}{EndForEach}
\algnewcommand\algorithmicforeach{\textbf{for\,each}}
\algnewcommand\algorithmicforeachdo{\textbf{do}}
\algnewcommand\algorithmicendforeach{\textbf{end\ for\,each}}
\algrenewtext{ForEach}[1]{\algorithmicforeach\ #1\ \algorithmicforeachdo}
\algrenewtext{EndForEach}{\algorithmicendforeach}
\algtext*{EndForEach}
    
\algblock{Match}{EndMatch}
\algnewcommand\algorithmicmatch{\textbf{match}}
\algnewcommand\algorithmicmatchwith{\textbf{with}}
\algnewcommand\algorithmicendmatch{\textbf{end\ match}}
\algrenewtext{Match}[1]{\algorithmicmatch\ #1\ \algorithmicmatchwith}
\algrenewtext{EndMatch}{\algorithmicendmatch}
\algtext*{EndMatch}

\algblock{Case}{EndCase}
\algnewcommand\algorithmiccase{|}
\algnewcommand\algorithmicendcase{}
\algrenewtext{Case}[2]{\algorithmiccase\ #1\ $\rightarrow$\ #2}
\algrenewtext{EndCase}{\algorithmicendcase{}}
\algtext*{EndCase}

\algnewcommand\Input{\textbf{input: }} 
\algnewcommand\Output{\textbf{output: }}

\makeatletter
\providecommand{\leftsquigarrow}{%
  \mathrel{\mathpalette\reflect@squig\relax}%
}
\newcommand{\reflect@squig}[2]{%
  \reflectbox{$\m@th#1\rightsquigarrow$}%
}
\makeatother

\makeatletter
\providecommand{\leftvdash}{%
  \mathrel{\mathpalette\reflect@tack\relax}%
}
\newcommand{\reflect@tack}[2]{%
  \reflectbox{$\m@th#1\vdash$}%
}
\makeatother

\newcommand{\mlstinline}[1]{\text{\lstinline{#1}}}

\newcommand{\COMMENT}[3]{\ifdraft\textcolor{#1}{[\textbf{#2}: {#3}]}}
\newcommand{\afm}[1]{\COMMENT{dkgreen}{AFM}{#1}}

\newcommand{\ToolText}[1]{\textsc{#1}\xspace}
\newcommand{\Burst}{\ToolText{Burst}}
\newcommand{\Myth}{\ToolText{Myth}}
\newcommand{\Smyth}{\ToolText{SMyth}}

\newcommand{\Synquid}{\ToolText{Synquid}}
\newcommand{\Leon}{\ToolText{Leon}}

\newcommand{\GEq}{\ensuremath{::=~}\xspace}

\newcommand{\Values}{\ensuremath{\mathcal{V}}\xspace}

\newcommand{\CreateVersionSpace}{\ensuremath{\textsc{CreateVersionSpace}}\xspace}
\newcommand{\GetWitness}{\ensuremath{\textsc{GetWitness}}\xspace}
\newcommand{\Intersect}{\ensuremath{\mathsf{Intersect}}\xspace}

\newcommand{\Label}{\ensuremath{\mathsf{Label}}\xspace}

\newcommand{\SuchThat}{\ensuremath{~|~}}
\newcommand{\SemanticsOf}[1]{\ensuremath{[ \! [#1] \! ]}}

\newcommand{\Angel}{\Innocey{}}
\newcommand{\angelicmodels}{\ensuremath{\models^{\text{\tiny \Angel}}}}
\newcommand{\Trace}{\ensuremath{\textsc{Trace}}}
\newcommand{\Synthesize}{\textsc{Synthesize}\xspace}

\newcommand{\List}[1]{\ensuremath{[#1]}}
\newcommand{\Append}{\ensuremath{+\!\!\!\!+\ }}
\newcommand{\GetAcceptingRun}{\ensuremath{\mathsf{GetAcceptingRun}}\xspace}

\newcommand{\GetAcceptingTraces}{\textsc{GetAcceptingTraces}}
\newcommand{\Conflict}{\textsc{Conflict}\xspace}
\newcommand{\Prohibit}{\textsc{Prohibit}\xspace}
\newcommand{\children}{\ensuremath{\mathsf{Children}}}

\newcommand{\ExtractAngelicTraces}{\ensuremath{\textsc{ExtractAngelicTraces}}}
\newcommand{\Spec}{\ensuremath{\varphi}}

\newcommand{\BuildAngelicFTA}{\textsc{BuildAngelicFTA}\xspace}

\newcommand{\groundSpec}{\ensuremath{\chi}\xspace}
\newcommand{\domain}{\ensuremath{\mathsf{dom}}}
\newcommand{\vs}{\ensuremath{\mathcal{V}}\xspace}

\newcommand{\fta}{\ensuremath{\mathcal{A}}}
\newcommand{\ftastate}{\ensuremath{q}}
\newcommand{\ftastates}{\ensuremath{Q}}
\newcommand{\alphabet}{\ensuremath{\Sigma}}
\newcommand{\wildcard}{\ensuremath{\bot}}
\newcommand{\nodes}{\ensuremath{N}}
\newcommand{\rootnode}{\ensuremath{\mathsf{root}}\xspace}
\newcommand{\finalstates}{\ensuremath{\ftastates_f}}
\newcommand{\transitions}{\ensuremath{\Delta}}

\newcommand{\LanguageOf}[1]{\ensuremath{\mathcal{L}(#1)}}
\newcommand{\true}{\emph{true}}
\newcommand{\false}{\emph{false}}

\newcommand{\cmark}{\ding{51}}%
\newcommand{\xmark}{\ding{55}}

\usepackage{xcolor}
\newcommand{\needsrev}[1]{{\color{blue}{#1}}}

\newcommand{\witness}{\omega}
\newcommand{\correct}{{\color{dkgreen} \cmark}}
\newcommand{\incorrect}{{\color{dkred} \xmark}}

\newcommand{\Respects}{\ensuremath{\textsc{Respects}}\xspace}

\DeclareMathOperator{\opr}{op}

\setcopyright{rightsretained}
\acmPrice{}
\acmDOI{10.1145/3498682}
\acmYear{2022}
\copyrightyear{2022}
\acmSubmissionID{popl22main-p124-p}
\acmJournal{PACMPL}
\acmVolume{6}
\acmNumber{POPL}
\acmArticle{21}
\acmMonth{1}

\keywords{Program Synthesis, Angelic Execution, Logical Specifications}

\ccsdesc[500]{Software and its engineering~Recursion}
\ccsdesc[300]{Software and its engineering~Functional languages}
\ccsdesc[500]{Theory of computation~Tree languages}

\begin{document}

%%% The following is specific to POPL '22 and the paper
%%% 'Bottom-Up Synthesis of Recursive Functional Programs using Angelic Execution'
%%% by Anders Miltner, Adrian Trejo Nuñez, Ana Brendel, Swarat Chaudhuri, and Isil Dillig.
%%%

\title{Bottom-Up Synthesis of Recursive Functional Programs using Angelic Execution} %title sucks, whatever

\begin{abstract}
%In this paper, 
We present a novel bottom-up method for the synthesis of functional recursive programs. 
%Bottom-up synthesis techniques possess some scalability advantages in that the intermediate programs that they produce during synthesis are all executable. 
%They are also especially compatible with version space methods, which advantages in incremental synthesis settings. 
While bottom-up synthesis techniques 
%are well-suited for settings in which the specification is a complex logical constraint or changes incrementally across synthesis tasks.
can work better than top-down methods in certain settings, 
there is no prior  technique for synthesizing recursive programs from logical specifications in a purely bottom-up fashion. The main challenge is that effective bottom-up methods need to execute sub-expressions of the code being synthesized, but it is impossible to execute a recursive subexpression of a program that has not been fully constructed yet.  In this paper, we address this challenge using the concept of \emph{angelic semantics}.
%and \emph{angelic synthesis}. 
Specifically, our method finds a program that satisfies the specification under angelic semantics (we refer to this as \emph{angelic synthesis}), analyzes the assumptions made during its angelic execution, uses this analysis to strengthen the specification, and finally reattempts synthesis with the strengthened specification. Our proposed angelic synthesis algorithm is based on version space learning and therefore deals effectively with many incremental synthesis calls made during the overall algorithm. 
We have implemented this approach in a prototype called \toolname and evaluate it on synthesis problems from prior work. 
Our experiments show that \toolname is able to synthesize a solution to 94\% of the benchmarks in our benchmark suite, outperforming prior work.
\end{abstract}

\author{Anders Miltner}
\affiliation{
  \institution{UT Austin}
  \city{Austin}
  \state{TX}
  \country{USA}                   %% \country is recommended
}
\email{amiltner@cs.utexas.edu}

\author{Adrian Trejo Nu\~nez}
\affiliation{
  \institution{UT Austin}
  \city{Austin}
  \state{TX}
  \country{USA}                   %% \country is recommended
}
\email{atrejo@cs.utexas.edu}

\author{Ana Brendel}
\affiliation{
  \institution{UT Austin}
  \city{Austin}
  \state{TX}
  \country{USA}                   %% \country is recommended
}
\email{anabrendel@utexas.edu}

\author{Swarat Chaudhuri}
\affiliation{
  \institution{UT Austin}
  \city{Austin}
  \state{TX}
  \country{USA}                   %% \country is recommended
}
\email{swarat@cs.utexas.edu}

\author{Isil Dillig}
\affiliation{
  \institution{UT Austin}
  \city{Austin}
  \state{TX}
  \country{USA}                   %% \country is recommended
}
\email{isil@cs.utexas.edu}

\maketitle

\section{Introduction}

Methods for program synthesis from formal specifications typically come in two flavors: \emph{top-down} and \emph{bottom-up}. Top-down methods~\citep{thesys,igor,flashfill,myth,myth2,lambda2,synquid} iterate through a sequence of \emph{partial programs}, starting with an ``empty'' program and progressively refining them through the addition of new code. In contrast, bottom-up methods~\citep{transit,escher,odena2020bustle,stun} maintain a pool of \emph{complete programs} and  progressively generate new programs by composing existing ones. 
%Commonly, this pool of programs is represented in the form of an efficient data structure --- a \emph{version space} --- such as an e-graph~\citep{} or a finite tree automaton~\citep{}. 

Top-down and bottom-up approaches  have complementary strengths. For example, top-down methods work well when the specification  can be naturally decomposed into subgoals through an analysis of partial programs.
%it is possible to 
%prove the infeasibility of partial programs using program analysis. 
However, they can run into imprecision or computational complexity issues when the specification or the language semantics are complicated. In contrast, a bottom-up approach only needs to evaluate complete sub-expressions of a program, which is generally a much easier task than that of reasoning about partial programs.

%has an easy way of checking that a candidate program satisfies a specification. All it needs to do is to execute the already-generated subprograms that constitute the candidate program, and compose the results of these executions. 

Unfortunately, it is difficult to apply bottom-up synthesis  to programming languages that permit {recursion}.
This is because effective bottom-up  approaches need to  execute all sub-expressions of the target program; however, for recursive programs, sub-expressions can call the function being synthesized, whose semantics are still unknown.
%In particular, because recursive functions call themselves (which are functions that are yet to be synthesized!), it is not possible to evaluate all sub-expressions of recursive procedures at synthesis time.
One way to overcome this issue is to assume that the specification is \emph{trace-complete}, i.e., that the result of each such evaluation is part of the specification. (Indeed, such a strategy is followed in the \ToolText{Escher}~\citep{escher} system for the bottom-up synthesis of recursive programs.) However, trace-completeness is a restrictive assumption, and writing trace-complete specifications can be cumbersome and unintuitive. 

In this paper, we propose a new approach to bottom-up program synthesis that addresses this difficulty.
The key insight behind our solution
% At a high level, our main insight 
 is to use \emph{angelic execution}~\cite{angelic-first} to evaluate recursive sub-expressions of the program being synthesized. Specifically, our method first performs \emph{angelic synthesis} to find a program $P$ that satisfies the specification under the assumption that recursive calls can return \emph{any} value that is consistent with the specification. For example, if the specification is $0 \leq f(x) \leq x$, the angelic synthesizer assumes that a recursive call $f(2)$ can return \emph{any} of the integers $0, 1, $ or $2$, although in reality it can only return \emph{one} of these. Thus, when performing angelic synthesis of a function $f$, we only need access to $f$'s \emph{specification} rather than its full \emph{implementation}. 
 
 One complication with this approach is that a program $P$ that angelically satisfies its specification $\varphi$ may not \emph{actually} satisfy $\varphi$. To deal with this  difficulty, our method combines angelic synthesis with \emph{specification strengthening} and back-tracking search. In more detail, given an angelic synthesis result $P$, our synthesis technique first checks if $P$ satisfies $\varphi$ under the standard semantics. If so, then $P$ is returned as a solution. Otherwise, our method analyzes the assumptions made in angelic executions of $P$, uses this information to strengthen the specification, and re-attempts synthesis with the strengthened specification. If synthesis is unsuccessful with the strengthened specification, it backtracks and tries a different strengthening, continuing this process until it either finds the right program or exhausts the search space. 
 
 %Overall, our method is both sound and complete: Beyond ensuring that the returned program satisfies the specification (soundness), our method is also guaranteed to find the correct program if one exists.
 
% In other words, we find a program that satisfies the specification under the assumption that the recursive call may angelically return \emph{any} value that is consistent with its specification. While this strategy allows us to perform bottom-up synthesis, it can lead to programs that are \emph{spurious}, meaning that they satisfy the specification under the \emph{angelic} but not the \emph{actual} semantics of recursive calls. Our method deals with this challenge by analyzing assumptions made in these angelic executions, uses them to strengthen the specification, and re-attempts synthesis with the strengthened specification.

{As illustrated by the above discussion, our end-to-end approach requires gradually strengthening the specification and making many calls to an \emph{angelic synthesizer}. Thus, for our approach to be practical, it is important to have an  angelic synthesis technique that can reuse partial synthesis results. Additionally, it must be possible to easily analyze assumptions made in angelic executions in order to determine how to  strengthen the specification. Motivated by these considerations, we propose an angelic synthesis technique based on  \emph{finite tree automata}~\cite{blaze,dace}. Our proposed angelic synthesizer handles incremental specifications by taking the intersection of previously constructed tree automata (for weaker specifications) with new automata constructed from the additional specifications. This incremental nature of the angelic synthesizer allows our approach to efficiently handle a series of increasingly more complex specifications. Furthermore, by inspecting runs of the tree automaton, we can easily and efficiently analyze the assumptions made by the angelic synthesizer.}

We have implemented our technique in a tool called \toolname\footnote{Bottom-Up Recursive SynThesizer} and evaluate it on 45 benchmarks from prior work~\cite{myth} using three types of specifications, namely (1) input-output examples, (2) reference implementations, and (3) logical formulas.
%on Smyth's example-based benchmark suite. Additionally, we wrote new logical specifications for the benchmark suite, and evaluated \toolname on this as well.
Our evaluation shows that \toolname can synthesize more functions than prior work on all three types of specifications. In particular, our tool is able to synthesize 96\% (43) of the functions from input/output examples, 96\% (43) of the functions from reference implementations, and 91\% (41) of the functions from logical specifications. We also compare \toolname against a simpler variant that does not perform specification strengthening, and we show that our proposed backtracking  search technique is useful in practice.
%Additionally, we show that \toolname permits efficient incremental synthesis.

In summary, this paper makes the following contributions:
\begin{itemize}[leftmargin=*]
\item We present the first bottom-up synthesis procedure that can handle general recursion and general logical specifications, and does not require the restrictive trace-completeness assumption.
%    \item We present a novel bottom-up technique for synthesizing functional programs with recursive procedures. 
    \item  We introduce a new form of \emph{angelic program synthesis} that combines the use of angelic program semantics and specification strengthening and can make use of efficient version space representations. Some of the insights in our algorithm may be applicable outside the immediate setting that we target. 
%    \afm{We introduce \emph{angelic recursive synthesis} and show how to combine it with specification refinement and backtracking search to solve the standard program synthesis problem.}
%\item We describe an angelic synthesis algorithm based on finite tree automata that (a) efficiently handles incremental specifications, and (b) makes it possible to analyze assumptions made in angelic executions.
    
    %uses angelic execution to generate an \emph{over-approximate} version space (represented as a tree automaton), then iteratively refines this version space. 
    %builds a tree automaton over-approximating a version space (space of all recursive programs consistent with a specification) and iteratively refines the version spaces by analyzing angelic traces. 
    %We believe that components of this algorithm can be useful in other synthesis settings. 
    
    %the space of all recursive programs consistent with a given ground formula. 
    %\item We show how to iteratively refine the version space by analyzing angelic traces and strengthening the specification.
    \item Our artifactual contribution is an implementation of our approach, called \Burst. We have conducted an extensive experimental evaluation on synthesis benchmarks from prior work.
    %and three different specifications, including input-output examples, reference implementations, and logical formulas. 
    Our experiments show that \Burst significantly outperforms the state-of-the-art in the synthesis of recursive programs on several counts.
\end{itemize}

%\section{Preliminaries}
%\label{sec:preliminaries}

\section{Overview}
\label{sec:extended-example}

\begin{figure}
\centering
\begin{minipage}{.5\textwidth}
  \centering
  \includegraphics[width=.8\linewidth]{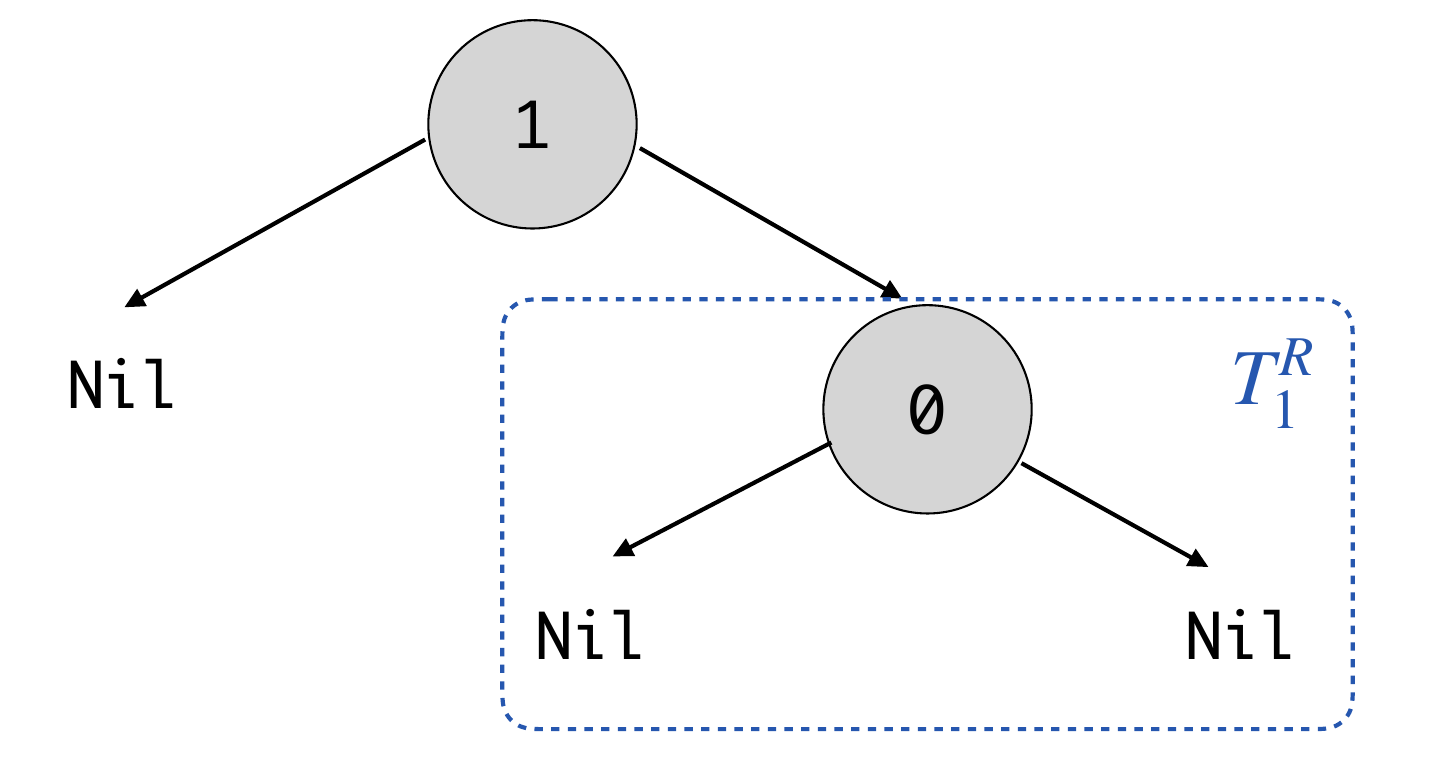}
  \captionof{figure}{First example $T_1$}
  \label{fig:tree1}
\end{minipage}%
\begin{minipage}{.5\textwidth}
  \centering
    \vspace{-0.65in}
  \includegraphics[width=.8\linewidth]{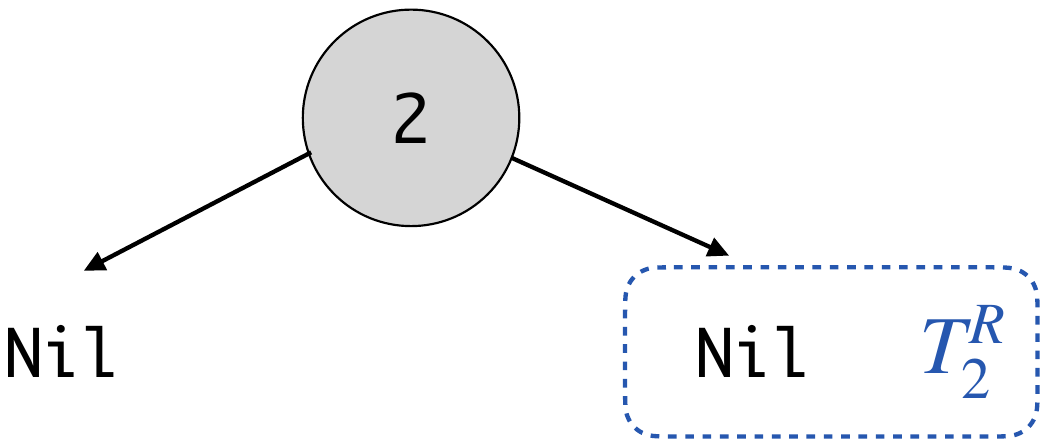}
  \vspace{-1.0in}
  \captionof{figure}{Second example $T_2$}
  \label{fig:tree2}
\end{minipage}
\end{figure}

\begin{comment}
\begin{figure}[t]
    \centering
    \includegraphics[scale=0.9]{tree1.pdf}
    \caption{Example trees $T_1$ and $T_1^R$. The right spine of $T_1$ is the list [1;0], and the right spine of $T_1^R$ is the list [0].}
    \label{fig:tree1}
\end{figure}

\begin{figure}[t]
    \centering
    \includegraphics[scale=0.6]{tree2.pdf}
    \caption{Example trees $T_2$ and $T_2^R$. The right spine of $T_1$ is the list [2], and the right spine of $T_2^R$ is the list [].}
    \label{fig:tree2}
\end{figure}
\end{comment}

In this section, we give an overview of our method with the aid of a motivating example. Our goal in this example is to synthesize a recursive implementation of the \lstinline{right_spine} procedure, which takes as input a tree and produces a list that is obtained by traversing the rightmost children of a node, starting from the root  and  continuing until a leaf node is reached.   As an example, Figures~\ref{fig:tree1} and~\ref{fig:tree2} show two trees $T_1$ and $T_2$, and a partial input-output specification for \lstinline{right_spine} is given as follows:
\begin{equation}
\label{eq:init-spec}
\begin{split}
\mlstinline{right_spine}(T_1) & = [1;0] \\ 
\mlstinline{right_spine}(T_2) & = [2]
\end{split}
\end{equation}
Note that our method can work with specifications that are not input-output examples (see Section~\ref{sec:logical-general}); here, we simply choose it for simplicity of presentation. We now explain how our technique synthesizes this \lstinline{right_spine} procedure in a bottom-up fashion.

\subsection{High Level Algorithm}\label{sec:overview-high}

Our algorithm works in a refinement loop that performs two major steps: (1) it synthesizes a program that \emph{angelically} satisfies the specification, and (2) strengthens the specification based on the assumptions made in the angelic execution. In this subsection, we illustrate the high-level approach on \lstinline{right_spine}, leaving the details of angelic synthesis to Section~\ref{sec:overview-angelic}.

\paragraph*{Iteration 1} The algorithm starts by invoking the angelic synthesizer to find a program that \emph{angelically} satisfies the specification shown in Equation~\ref{eq:init-spec}. As we will discuss later, the angelic synthesizer outputs the following program in this iteration:

\begin{lstlisting}
   let rec P1(x) = 
            match x with 
              | Nil -> []
              | Node(l,v,r) -> P1(r)
\end{lstlisting}
Clearly, this program does not actually satisfy the specification, but it does satisfy the specification under the angelic semantics: Since the specification from Eq.~\ref{eq:init-spec} does not constrain the output of the recursive call on $T_1^R$, the angelic synthesizer assumes that  the recursive call to \lstinline{P1} can return anything, including \lstinline{[1;0]}, for the right subtree of $T_1$. Thus, program \lstinline{P1} satisfies the specification under the angelic semantics of recursion.

Next, our algorithm checks whether the candidate program satisfies the specification under the actual semantics. Since 
\lstinline{P1}($T_1$) = \lstinline{P1}($T_2$) = \lstinline{[]},
it clearly does not, and our algorithm analyzes the assumptions made in the angelic execution to determine how to strengthen the specification. In this case, the angelic execution assumes that the recursive call on the right-subtrees $T_{1}^R$ and $T_2^R$ return $[1; 0]$ and $[2]$ respectively. Thus, our algorithm re-attempts synthesis using the following strengthened specification: 

\begin{equation}
\label{eq:spec2a}
\begin{array}{rlrl}
\mlstinline{right_spine}(T_1) & = [1;0] & \hspace*{2em} 
\mlstinline{right_spine}(T_2) & = [2] \\[4pt]
\mlstinline{right_spine}(T_1^R) & = [1;0] & \hspace*{2em}
\mlstinline{right_spine}(T_2^R) & = [2]
\end{array}
\end{equation}

\paragraph*{Iteration 2a} In the next recursive call, our algorithm invokes the angelic synthesizer to find a program consistent with the specification shown in  Eq.~\ref{eq:spec2a} but it fails.

%. However, the angelic synthesizer fails to find a program in its search space\footnote{Like most synthesizers, our angelic synthesizer looks for programs up to a certain size.} that is consistent with Eq.~\ref{eq:spec2a}. 

\paragraph*{Iteration 2b} Since synthesis was unsuccessful for Eq.~\ref{eq:spec2a}, our algorithm backtracks and tries a different strengthening. Specifically, since we could not find a program where the recursive calls on $T_1^R$ and $T_2^R$ return $[1;0]$ and $[2]$ respectively, we now strengthen the specification using the \emph{negation} of these assumptions. This yields the following specification for the next recursive call to the synthesizer:

\begin{equation}
\label{eq:spec2b}
\begin{array}{rcl}
\mlstinline{right_spine}(T_1) = [1;0] & & 
\mlstinline{right_spine}(T_2) = [2] \\[4pt]
\neg(\mlstinline{right_spine}(T_1^R) = [1;0] & \wedge & 
\mlstinline{right_spine}(T_2^R) = [2])
\end{array}
\end{equation}

In this case, the angelic synthesizer returns the following program:

\begin{lstlisting}
 let rec P2(x) = 
            match x with 
              | Nil -> [0]
              | Node(l,v,r) -> v::P2(r)
\end{lstlisting}

This program is again incorrect but it does satisfy Eq.~\ref{eq:spec2b} under the angelic semantics. Indeed, a ``witness" to angelic satisfaction is:
\begin{center}
\lstinline{P2}($T_1^R$) = \lstinline{[0]} $\land$ \lstinline{P2}($T_2^R$) = \lstinline{[]}
\end{center}
Note that this assumption is allowed under the angelic semantics since these return values on $T_1^R$ and $T_2^R$ are both consistent with Eq.~\ref{eq:spec2b}. Thus, using the witness to angelic satisfaction, we now strengthen the specification as follows:

\begin{equation}
\label{eq:spec3b}
\begin{array}{rcl}
\mlstinline{right_spine}(T_1) = [1;0] & & 
\mlstinline{right_spine}(T_2) = [2] \\[4pt]
\neg(\mlstinline{right_spine}(T_1^R) = [1;0] & \wedge & 
\mlstinline{right_spine}(T_2^R) = [2])\\[4pt]
\mlstinline{right_spine}(T_1^R) = [0] & &
\mlstinline{right_spine}(T_2^R) = [] \\
\end{array}
\end{equation}

\paragraph*{Iteration 3b} In the next (and last) iteration, when we invoke the angelic synthesizer on Eq.~\ref{eq:spec3b}, it outputs the following program:

\begin{lstlisting}
 let rec P3(x) = 
            match x with 
              | Nil -> []
              | Node(l,v,r) -> v::P3(r)
\end{lstlisting}

This program satisfies the specification under the actual semantics; thus, the algorithm terminates with \lstinline{P3} as the (correct) solution.

\subsection{Angelic Synthesis using FTAs}\label{sec:overview-angelic}

\begin{figure}[t]
    \centering
    \includegraphics[scale=0.23]{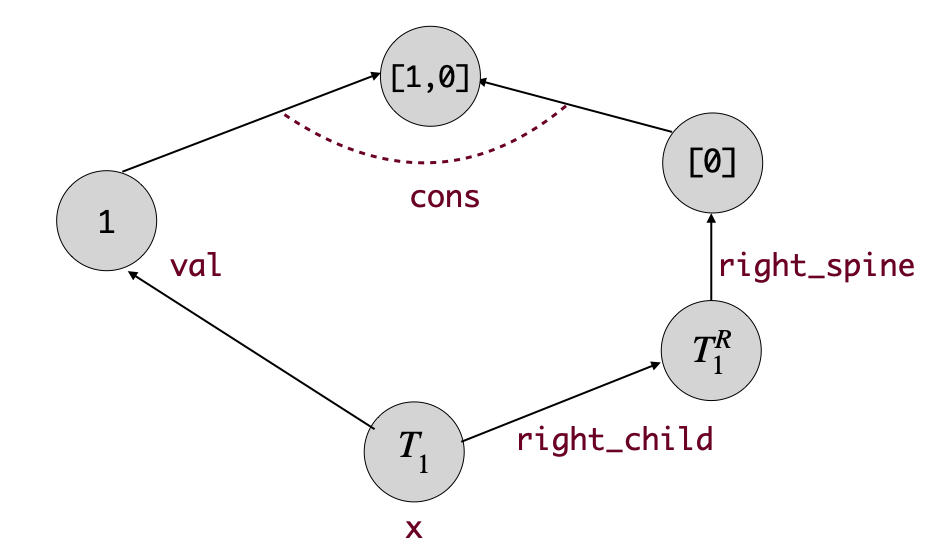}
    \caption{An example FTA that accepts a program that brings the input $T_1$ to the valid output of \lstinline{[1;0]}.}
    \label{fig:example-fta}
\end{figure}

As illustrated by the above discussion, a key piece of our technique is the \emph{angelic synthesizer} for finding a program that satisfies the specification under the angelic semantics. Inspired by prior work on bottom-up synthesis~\cite{blaze,dace}, our angelic synthesizer constructs a \emph{finite tree automaton (FTA)} that compactly represents a set of programs. In a nutshell, FTAs generalize standard automata by accepting trees instead of words. In our setting, the states in the automata correspond to concrete program values (e.g., lists like \lstinline{[1;0]} or \lstinline{[]}), and the trees accepted by the automaton correspond to programs (i.e., abstract syntax trees). 

In order to explain our angelic synthesis approach, we first briefly review the construction from prior work~\cite{dace}. The idea is to construct a separate automaton for each input (e.g., $T_1$ from Fig~\ref{fig:tree1}) and then take the intersection of all of these automata. To construct each automaton, we start with the given input and obtain new states by applying language constructs to the existing states. For example, given states $q_2$ and $q_3$ representing integers $2$ and $3$ and the operator $+$, we generate a new state $q_5$ (for integer 5) by applying the transition $+(q_2, q_3) \rightarrow q_5$. Since the accepting states of the FTA are those that satisfy the specification, the language of the constructed automaton includes exactly those programs that are consistent with the specification.

As illustrated by the above discussion, such an FTA-based synthesis method is \emph{bottom-up} in that it evaluates complete sub-expressions on the input and combines the values of these sub-expressions to generate new values. However, prior work cannot deal with recursive functions because it is not possible to evaluate a function that has not yet been synthesized. For example, consider the recursive call to \lstinline{right_spine(r)} where $r$ has value $T_1^R$ in our running example. Since \lstinline{right_spine} has not yet been synthesized, we simply do not know what \lstinline{right_spine} will return on $T_1^R$. 

To deal with this challenge, our angelic synthesizer assumes that the result of the recursive call could be any value that is consistent with the specification. In particular, given a specification $\varphi$ and  FTA states $q_1, \ldots, q_n$, we assume that a recursive  invocation expression \lstinline{f}($\overline{q}$) could evaluate to any $q_i$ as long as $q_i$ is consistent with $\varphi$. For instance, Figure~\ref{fig:example-fta} shows an FTA with states $\set{T_1,T_1^R,1,[0],[1;0]}$ for the angelic synthesis problem for Eq.~\ref{eq:spec2b}. Here, there is a transition \lstinline{right_spine}$(T_1^R) \rightarrow [0]$ since the call \lstinline{right_spine}$(T_1^R) = [0]$ is consistent with Eq.~\ref{eq:spec2b}. {Note that edges in Figure~\ref{fig:example-fta} correspond to program syntax and \lstinline{[1,0]} is an accepting state, so the program \lstinline{right_spine(x) = val(x)::right_spine(right_child(x)))} is accepted by this FTA.}

As illustrated by this discussion, the use of angelic semantics allows us to construct a bottom-up tree automaton despite not knowing what the recursive invocation will return on a given input. However, an obvious ramification of this is that programs accepted by the automaton may not satisfy the specification under the true semantics, which is why our method combines angelic synthesis with specification strengthening and backtracking search, as described in Section~\ref{sec:overview-high}.

\subsection{Incremental Synthesis} 

As we saw from Equations~\ref{eq:init-spec},~\ref{eq:spec2b},  and~\ref{eq:spec3b}  from Section~\ref{sec:overview-high}, successive calls to the synthesis algorithm involve increasingly strong specifications. In particular, if the synthesis algorithm is invoked on specification $\varphi$ in the $i$'th iteration, then the specification in the $i+1$'th iteration is of the form $\varphi \land \psi$. We exploit this incremental nature of the algorithm to make angelic synthesis more efficient.

In particular, recall that our angelic synthesizer based on FTAs constructs a different FTA for each input and then takes their intersection. Thus, given a specification $\varphi \land \psi$ where $\varphi$ is the old specification, we can simply construct a new FTA for $\psi$ and then take its intersection with the old FTA for $\varphi$. Hence, performing angelic synthesis using FTAs allows us to reuse all the work from prior iterations.

\begin{comment}
Note that the specification $\groundSpec_1$ was used multiple times in our extended example. When synthesizing a program that satisfied both $\groundSpec_{2a}$ and $\groundSpec_{2b}$, the programs had to also satisfy $\groundSpec_{1}$. When synthesizing a program that satisfied $\groundSpec_{3b}$, the program also had to satisfy $\groundSpec_{2b}$. By using FTAs, \Burst has a representation of \emph{all} programs up to a given size that satisfy $\groundSpec_1$. To refine this to those that satisfy $\groundSpec_{2a}$, \Burst merely has to intersect the programs that satisfy the \emph{new} constraints with those that satisfied the old ones.

While incremental synthesis is helpful \emph{within} a given synthesis, it can also be helpful when \Burst is used \emph{interactively}. Say, when using a user gave the specification $\groundSpec_1$, but they actually wanted to synthesize an \lstinline{inorder} traversal of the tree. After realizing they had underspecified \lstinline{inorder} they refine the specification to $\groundSpec_1' = \groundSpec_1 \wedge f(\mlstinline{N(N(L,3,L),1,L)}) = [3;1]$.

At it's simplest, we can simply reuse the specification $\groundSpec$ like we did within the program. But we can also memoize knowledge that $\groundSpec_{2a}$ was not satisfiable, and skip any specification that is stronger than it. In effect, we can restart the synthesis at the same place in the call stack as $\groundSpec_{3b}$, updated to be $\groundSpec_{3b}' = \groundSpec_{eb} \wedge f(\mlstinline{N(N(L,3,L),1,L)}) = [3;1]$.
\end{comment}

\subsection{Generalization to Arbitrary Logical Specs}
\label{sec:logical-general}

In our example so far, we illustrated the synthesis algorithm on the simple input-output examples from Eq.~\ref{eq:init-spec}. However, our method can be generalized to more complicated logical specifications using the standard counterexample-guided inductive synthesis (CEGIS) paradigm. In particular, since our core synthesis algorithm takes as input \emph{ground formulas} (defined in Section~\ref{sec:problem-statement}) as opposed to input-output examples, it can be easily incorporated within the CEGIS loop to handle more general logical specifications. For instance, our method can produce the correct implementation of \lstinline{right_spine} given the following logical specification:
\[
  \phi(in,out) := \mlstinline{no_left_subchildren}(in) \Rightarrow (\mlstinline{tree_size}(in) == \mlstinline{list_size}(out))
\]
where \lstinline{tree_size}  and \lstinline{list_size} return the number of elements in a tree and list respectively,  and \lstinline{no_left_subchildren} returns true if the left child of every node in the tree is a leaf.

\section{Problem Statement}
\label{sec:problem-statement}

\begin{figure}
  \begin{tabularx}{0.4\textwidth}{@{}r@{\ }c@{\ }c@{\ }l}
    $P$ &        &\GEq{} & {\lstinline|rec f($x$) = $\ e$|} \\
    \end{tabularx}\\
  \begin{tabularx}{0.4\textwidth}{@{}r@{\ }c@{\ }c@{\ }l@{\hspace*{2em}}c@{\ \ \ }l}
    $e$ &        &\GEq{} & $x$\\
        &        & | & {\lstinline|$e_1\ e_2$|} & | & {\lstinline|unit|}\\
        &        & | & {\lstinline|inl $\ e$|} & | & {\lstinline|inr $\ e$|} \\
        &        & | & {\lstinline|unl $\ e$|} & | & {\lstinline|unr $\ e$|}\\
        &        & | & {\lstinline|fst $\ e$|} & | & {\lstinline|snd $\ e$|}\\
        &        & | & {\lstinline|($e_1,\ e_2$)|} & | & {\lstinline|switch $\ e_3$ on inl _|\ $\rightarrow e_1$\ \lstinline{| inr _}\ $\rightarrow e_2$}\\
        % & & & & & \hspace*{1em} {\lstinline|inl $\ \rightarrow e_1$|} \\
        % & & & & & \hspace*{1em} {\lstinline|inr $\ \rightarrow e_2$|}
  \end{tabularx}\\
  \begin{tabularx}{0.4\textwidth}{@{}r@{\ }c@{\ }c@{\ }l@{\hspace*{2em}}c@{\ \ \ }l}
    $v$ &        &\GEq{} & {\lstinline|unit|} & | & {\lstinline|($v_1,\ v_2$)|}\\
        &        & | & {\lstinline|inl $\ v$|} & | & {\lstinline|inr $\ v$|}\\
  \end{tabularx}\\
\caption{A functional ML-like language with explicit recursion in which we synthesize programs. The nonterminal $P$ denotes programs in this language, and the nonterminal $v$ denotes values in this language.}
\label{fig:lang}
\end{figure}

In this section, we present our problem statement, which  is synthesizing recursive programs in a simple ML-like language  with products and sums (see Figure~\ref{fig:lang}). Without loss of generality, we assume that programs take a single input, as we can represent multiple inputs using tuples (i.e., pairs with nested pairs). Given a program $P$ and a concrete input $v$, we use the notation $\SemanticsOf{P}(v)$ to denote the result of executing $P$ on input $v$ according to the semantics presented in Figure~\ref{fig:standard-semantics}.

\begin{figure}
\begin{mathpar}
    \inferrule
    {
      e_2 \Downarrow v_2\\
      e_1[\hbox{\zlstinline!rec f!$(x)=e_1$}/\hbox{\zlstinline!f!},v_2/x] \Downarrow v_3\\
    }
    {
      (\hbox{\zlstinline!rec f!$(x)=e_1$})\ e_2 \Downarrow v_3
    }\\
    
    \inferrule
    {
    }
    {
      \hbox{\zlstinline!unit!} \Downarrow \hbox{\zlstinline!unit!}
    }
    
    \inferrule
    {
      e_1 \Downarrow v_1\\
      e_2 \Downarrow v_2\\
    }
    {
      (e_1,e_2) \Downarrow (v_1,v_2)
    }
    
    \inferrule
    {
      e \Downarrow (v_1,v_2)
    }
    {
      \hbox{\zlstinline!fst $\;e \Downarrow v_1$!}
    }
    
    \inferrule
    {
      e \Downarrow (v_1,v_2)
    }
    {
      \hbox{\zlstinline!snd $\;e \Downarrow v_2$!}
    }\\
    
    \inferrule
    {
      e \Downarrow v
    }
    {
      \hbox{\zlstinline!inl! $e \Downarrow$ \zlstinline!inl! $v$}
    }
    
    \inferrule
    {
      e \Downarrow v
    }
    {
      \hbox{\zlstinline!inr! $e \Downarrow$ \zlstinline!inr! $v$}
    }
    
    \inferrule
    {
      e \Downarrow \hbox{\zlstinline!inl!}\ v
    }
    {
      \hbox{\zlstinline!unl!}\ e \Downarrow v
    }
    
    \inferrule
    {
      e \Downarrow \hbox{\zlstinline!inr!}\ v
    }
    {
      \hbox{\zlstinline!unr!}\ e \Downarrow v
    }
    
    \inferrule
    {
      e_3 \Downarrow \hbox{\zlstinline!inl!}\ v_3\\
      e_1 \Downarrow v_1
    }
    {
      \hbox{\zlstinline!switch! $e_3$ \zlstinline!on inl! $\_ \to e_1$ \zlstinline!inr! $\_ \to e_2$}
        \Downarrow v_1
    }
    
    \inferrule
    {
      e_3 \Downarrow \hbox{\zlstinline!inr!}\ v_3\\
      e_2 \Downarrow v_2
    }
    {
      \hbox{\zlstinline!switch! $e_3$ \zlstinline!on inl! $\_ \to e_1$ \zlstinline!inr! $\_ \to e_2$}
        \Downarrow v_2
    }
\end{mathpar}
\caption{Program Semantics. The symbols $e$ range over expressions and $v$ range over values. Both $f$ and $x$ denote arbitrary free variables.
If $P=\mlstinline{rec f}(x)=e$
and $P\ v \Downarrow v'$
then $\SemanticsOf{P}(v) = v'$.}
\label{fig:standard-semantics}
\end{figure}

Our goal in this paper is to synthesize a \emph{single} recursive procedure {\tt f} from a given specification, which is represented as a \emph{ground formula} $\Spec$. We assume that $\Spec$ always contains a special uninterpreted function symbol $f$ which refers to the function to be synthesized. More formally, we define \emph{ground specifications} as follows:

\begin{definition}{\bf (Ground specification)}
A \emph{ground specification} is a boolean combination of atomic formulas of the form $f(i) \opr c$ where $f$ denotes the function to be synthesized, $i$ and $c$ are constants (with $i$ being the input), and $\opr$ is a binary relation.
\end{definition}

\begin{definition}{\bf (Satisfaction of ground spec)}
Given a program $P$ defining function $f$, we say that $P$ \emph{satisfies} a ground specification $\Spec$, denoted $P \models \Spec$, iff the following condition holds:
%\[
%P \models \Spec \ \Longleftrightarrow
%\mathsf{SAT}(\varphi \wedge \forall v. \SemanticsOf{P}(v) = v' \Rightarrow f(v)=v')
%\]
%\[
%P \models \Spec \ \Longleftrightarrow \ \forall v.  \ \ 
%\mathsf{SAT}(\varphi \wedge  f(v)=\SemanticsOf{P}(v) )
%\]
% \[
% P \models \Spec \ \Longleftrightarrow \ (\forall v.  \ \ 
% f(v)=\SemanticsOf{P}(v)) \Rightarrow \Spec
% \]
%
%
\[
P \models \Spec \ \ \Longleftrightarrow \ \  \models \Spec\big [\SemanticsOf{P}(x) / f(x) \big ]
\]
where the notation  $\Spec[\SemanticsOf{P}(x) / f(x)]$ denotes the formula \Spec{} with every ground term  $f(v_i)$ is replaced by $\SemanticsOf{P}(v_i)$.

%\[
%P \models \Spec \ \ \Longleftrightarrow \ \  (\bigwedge_{i \in Inputs(P)} f(i) = %\SemanticsOf{P}(i) ) \Rightarrow \Spec
%\]
%where $\mathsf{Inputs}(P)$ denotes the set of all possible inputs of program $P$. 

% \[
% P \models \Spec \ \Longleftrightarrow \ \forall v. \ \ \models \Spec[\SemanticsOf{P}(v) / f(v)]
% \]

% \[
% P \models \Spec \ \Longleftrightarrow \ \ \ \models \Spec[\SemanticsOf{P}(v_1) / f(v_1)]\ldots[\SemanticsOf{P}(v_n) / f(v_n)]
% \]
\end{definition}

\vspace{5pt}
\noindent
\fbox{\parbox{.95\linewidth}{
	{\bf Problem Statement}: Given a ground specification $\Spec$, find program $P$ such that $P \models \Spec$.
  }
}
\vspace{5pt}

Note that ground specifications are quite powerful: synthesizers that can generate programs from ground specifications can also perform synthesis from a number of specification classes.
For example we can always encode I/O examples as ground formulas. but not vice versa.
\begin{example}
Consider the set of input-output examples $\{1 \mapsto 2, 2 \mapsto 3 \}$. We can encode this specification in our format using the ground formula $f(1)= 2 \land f(2) = 3$. In general, I/O examples correspond to specifications of the form:
\[
\bigwedge_k f(i_k) = o_k
\]
where $(i_k, o_k)$ are input-output pairs.
\end{example}

Furthermore, we can also lift  synthesis from ground specifications to an even more general class of specifications  using the well-known CEGIS paradigm (see Figure~\ref{fig:cegis}). Given a fixed set of inputs $I$ and a general predicate $\varphi$ with variables (representing inputs), one can convert this into a ground formula of the form $\bigwedge_{i \in I}\varphi(i)$ where each $i$ a counterexample returned by the verifier. Since the CEGIS paradigm invokes the verifier to add new counterexamples if the synthesized program is not correct, 
%After synthesizing a candidate program $P$, the verifier will check whether this program generalizes to satisfy the general formula. If it does not, it will return a counterexample, which is added to the set of inputs. Thus, the ability to perform 
synthesis from ground formulas immediately provides a way to perform synthesis from more general logical specifications. 
\begin{example}
Consider the problem of synthesizing a function that returns a value greater than its input for all positive inputs. The specification for such a function is of the form $x>0 \Rightarrow f(x) > x$. While this specification is not a ground formula, we can embed our synthesis technique into the CEGIS paradigm and reduce it to inductively synthesizing programs from ground specifications of the form:
\[
\bigwedge_k f(i_k) > i_k
\]
where each $i_j$ is a positive integer returned as a counterexample by a verifier.
\end{example}

\begin{figure}[t]
    \centering
    \includegraphics[scale=0.3]{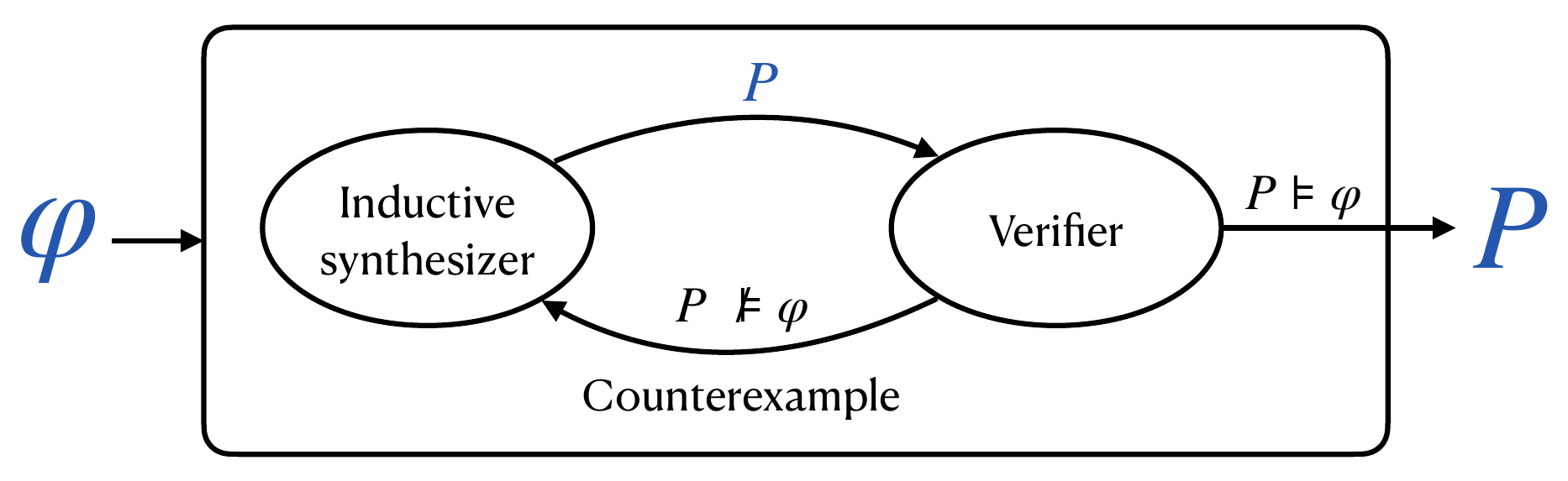}
    \caption{Counterexample-guided inductive synthesis. Since the input to the inductive synthesizer is a ground formula, our approach can be lifted to a general class of specifications using the CEGIS paradigm.}
    \label{fig:cegis}
\end{figure}

\section{Angelic Recursion}

As mentioned earlier, our method is based on bottom-up synthesis, which requires the ability to execute sub-expressions of the program being synthesized. Since this is not feasible for recursive procedures, we introduce the notion of \emph{angelic recursion} and \emph{angelic satisfaction}. 

\begin{definition}{\bf (Angelic recursion)} Given a recursive procedure $P$, the \emph{angelic semantics} of $P$ with respect to specification $\Spec$, denoted $\SemanticsOf{P}^\Spec$, is defined in Figure~\ref{fig:angelic-semantics}. These semantics are very similar to the semantics in Figure~\ref{fig:standard-semantics}; the key difference lies in how recursion is performed. When performing a recursive call \lstinline{f($v$)}, the result can be any $v'$ where $f(v) = v'$ is consistent with the specification. Thus, $\SemanticsOf{P}^\Spec(v)$ yields a \emph{set} of values $\Values$.
%where each $v_i \in \Values$ if it can be produced by the program when recursive calls can return any consistent value.
\end{definition}

\begin{figure}
\begin{mathpar}
    \inferrule
    {
      e \Downarrow^\varphi v\\
      \mathsf{SAT}(f(v) = v' \wedge \varphi)
    }
    {
      \hbox{\zlstinline!f!}\ e \Downarrow^\varphi v'
    }
    
    \inferrule
    {
    }
    {
      \hbox{\zlstinline!unit!} \Downarrow^\varphi \hbox{\zlstinline!unit!}
    }
    
    \inferrule
    {
      e_1 \Downarrow^\varphi v_1\\
      e_2 \Downarrow^\varphi v_2\\
    }
    {
      (e_1,e_2) \Downarrow^\varphi (v_1,v_2)
    }\\
    
    \inferrule
    {
      e \Downarrow^\varphi (v_1,v_2)
    }
    {
      \hbox{\zlstinline!fst $\;e \Downarrow^\varphi v_1$!}
    }
    
    \inferrule
    {
      e \Downarrow^\varphi (v_1,v_2)
    }
    {
      \hbox{\zlstinline!snd $\;e \Downarrow^\varphi v_2$!}
    }\\
    
    \inferrule
    {
      e \Downarrow^\varphi v
    }
    {
      \hbox{\zlstinline!inl! $e\Downarrow^\varphi$ \zlstinline!inl! $v$}
    }
    
    \inferrule
    {
      e \Downarrow^\varphi v
    }
    {
      \hbox{\zlstinline!inr! $e \Downarrow^\varphi$ \zlstinline!inr! $v$}
    }
    
    \inferrule
    {
      e \Downarrow^\varphi \hbox{\zlstinline!inl!}\ v
    }
    {
      \hbox{\zlstinline!unl!}\ e \Downarrow^\varphi v
    }
    
    \inferrule
    {
      e \Downarrow^\varphi \hbox{\zlstinline!inr!}\ v
    }
    {
      \hbox{\zlstinline!unr!}\ e \Downarrow^\varphi v
    }
    
    \inferrule
    {
      e_3 \Downarrow^\varphi \hbox{\zlstinline!inl!}\ v_3\\
      e_1 \Downarrow^\varphi v_1
    }
    {
      \hbox{\zlstinline!switch! $e_3$ \zlstinline!on inl! $x_1 \to e_1$ \zlstinline!inr! $x_2 \to e_2$}
        \Downarrow^\varphi v_1
    }\\
    
    \inferrule
    {
      e_3 \Downarrow^\varphi \hbox{\zlstinline!inr!}\ v_3\\
      e_2 \Downarrow^\varphi v_2
    }
    {
      \hbox{\zlstinline!switch! $e_3$ \zlstinline!on inl! $x_1 \to e_1$ \zlstinline!inr! $x_2 \to e_2$}
        \Downarrow^\varphi v_2
    }\\
    
    \inferrule
    {
      P = \hbox{\zlstinline|rec f|}(x) = e\\
      e[v/x] \Downarrow^\varphi v'
    }
    {
      v' \in \SemanticsOf{P}^\Spec(v)
    }
\end{mathpar}
\caption{Angelic Semantics. The key difference between angelic semantics and standard semantics lies in the first rule, for recursive calls.}
\label{fig:angelic-semantics}
\end{figure}

Intuitively, angelic recursion is useful in our setting because it allows us to ``execute" a recursive program without knowing the exact behavior of recursive calls.

\begin{example}
Let $P$ be the program \lstinline{rec f(x) = if x=0 then 1 else f(x-1)} and suppose that $\Spec = f(0) > 0$. Then, $\SemanticsOf{P}^\Spec(0) = \set{1}$, and $\SemanticsOf{P}^\Spec(1) = \set{y \mid y > 0}$. In particular, for input $1$, $f$ contains a recursive invocation on input $0$, and the angelic semantics allows the recursive call to return any value greater than $0$. Thus, $\SemanticsOf{P}^\Spec(1)$ is exactly the set of positive integers.

\end{example}

Next, we define a notion of angelic satisfaction:

\begin{definition}{\bf (Angelic satisfaction on input)}
Given a program $P$ defining function $f$, we say that $P$ \emph{angelically satisfies} specification $\Spec$ on input $v$, denoted $P \angelicmodels_v \Spec$, iff the following condition holds:
%Given formula $\Spec$ and $\psi$, we say that program $P$ angelically satisfies $\Spec$ with respect to $\psi$, denoted $P \models^\psi \Spec$, iff:
\[
P \angelicmodels_v \Spec \ \Longleftrightarrow \  \exists v'. \ v' \in \SemanticsOf{P}^\Spec(v) \land \mathsf{SAT}(f(v) = v' \wedge \varphi)
\]
\end{definition}

Next, we generalize this notion of angelic satisfaction from a single input to all inputs:

\begin{definition}{\bf (Angelic satisfaction)}
A program $P$ \emph{angelically satisfies} specification $\Spec$, denoted $P \angelicmodels \Spec$,  iff for all possible inputs $v$, we have  $P \angelicmodels_v \Spec$.
\end{definition}

Note that angelic satisfaction ($P \angelicmodels \Spec$) is a much weaker notion than standard satisfaction ($P \models \Spec$).
This is illustrated by the following example:

\begin{example}
Let $P$ be the program \lstinline{rec f(x) = if x=0 then 1 else f(x-1)} and suppose that $\Spec = f(0) > 0 \land f(1) > 1$. Then, $P \angelicmodels \Spec$, as $\forall x. \  x+1 \in \SemanticsOf{P}^\Spec(x)$. However, clearly, this program does not satisfy $\Spec$ with respect to the standard semantics because we have $\SemanticsOf{P}(1) = 1$. 
\end{example}

If a program $P$ angelically satisfies a specification $\Spec$, we can define a \emph{witness} to angelic satisfaction as follows:

\begin{definition}{\bf (Witness to angelic satisfaction)}\label{def:witness}
Let  $P$ be a program  such that $P \angelicmodels \Spec$. Then, a \emph{witness} $\witness$  to angelic satisfaction of $P$ is a formula $\bigwedge_i f(c_i) = c_i'$ such that, if $P \models \witness$, then $P \models \Spec$. 
\end{definition}

Intuitively, a \emph{witness} to angelic satisfaction specifies what the recursive calls in $P$ must return in order for $P$ to actually satisfy the specification. We discuss how to find these witnesses in Section~\ref{subsec:find-witness}.

%Note that, if the specification is of the form $\Spec = f(v_i) \ op\ v_o$, the angelic semantics in Figure~\ref{fig:angelic-semantics} also show how to compute a witness of angelic satisfaction. In particular, for      $ P = \hbox{\zlstinline|rec f|}(x) = e$, if we have $e[v_i/x] \Downarrow v \rightsquigarrow \witness$ according to Figure~\ref{fig:angelic-semantics} and $\mathsf{SAT}(\Spec \land \witness)$, then $\witness$ serves as a witness for $P \angelicmodels \Spec$.

\begin{example}
Consider the program \lstinline{rec f(x) = if x=0 then 1 else f(x-1)+1} and  the ground specification  $f(1) > 1 \land f(2) > 2$. This program angelically satisfies the specification in an execution where the recursive call returns $f(0) = 1$ and $f(1) = 2$. Thus, $f(0) = 1 \land f(1) = 2$ is a witness to angelic satisfaction. Of course, note that angelic witnesses are not unique. For example, $f(0) = 2 \land f(1) = 3$ is also a witness to angelic satisfaction.  
\end{example}

\section{ Synthesis Algorithm using Angelic Execution}\label{sec:top-level}

\begin{algorithm}[!t]
\begin{algorithmic}[1]
\Statex {\bf input:} Ground specification \groundSpec
\Statex {\bf output:} A program $P$ or $\bot$ 
\Statex {\bf global:} $\Omega$ is a learned \emph{anti-specification}, initially $\emptyset$
\Procedure{Synthesize}{\groundSpec}
\State $\text{result} \gets \Call{SynthesizeAngelic}{\groundSpec \land \bigwedge_{\phi_i \in \Omega} \neg\phi_i }$ \label{line:angelic}
\Match{$\text{result}$}
\Case{$\mathsf{Failure}(\kappa)$}{}
\State $\Omega \gets \Omega \cup \kappa$ \label{line:failure1}
\State \Return $\bot$
\EndCase
\Case{$\mathsf{Success}(P,\witness)$}{}
\If{$P \models \groundSpec$} \Return $P$ \label{line:success1}
\Else
\State $\mathit{P} \gets \Call{Synthesize}{\groundSpec \land \witness}$ \label{line:witness}
\If{$\mathit{P} = \bot$}
\Return \Call{Synthesize}{$\groundSpec \land \neg \witness$} \label{line:failure2}
\Else \
\Return $\mathit{P}$ \label{line:success2}
\EndIf
\EndIf
\EndCase
\EndMatch
\EndProcedure
\end{algorithmic}
\caption{Core Recursive Synthesis Algorithm}
\label{alg:core-recursive-synthesis}
\end{algorithm}

In this section, we describe our top-level synthesis algorithm based on angelic recursion. While this synthesis algorithm does not specify whether to construct programs in a top-down or bottom-up fashion, we emphasize that it is the use of \emph{angelic recursion} that makes it possible to implement its key components using a bottom-up approach (as we discuss in the next section).

\autoref{alg:core-recursive-synthesis} shows the high-level structure of our synthesis algorithm. The procedure {\sc Synthesize} takes as input a ground specification $\groundSpec$  and returns either a program $P$ or $\bot$ to indicate that synthesis is unsuccessful. Internally, the algorithm also maintains a global variable, namely set $\Omega$, that we refer to as an \emph{anti-specification} which is used for pruning the search space. In particular, $\Omega$ is constructed in such a way that any program that satisfies $\phi \in \Omega$ is guaranteed to \emph{not} satisfy the desired specification $\chi$, i.e.:
\begin{equation}\label{eq:anti-spec}
\forall P. \ \forall \phi \in \Omega. \ P \models \phi \Rightarrow P \not \models \chi
\end{equation}
Since the contrapositive of \autoref{eq:anti-spec} is
$$\forall P. \ \forall \phi \in \Omega. \ P \models \chi \Rightarrow P \not\models \phi$$
which implies that $P$ must satisfy 
$\bigwedge_{\phi_i \in \Omega} \neg \phi_i$ in order to also satisfy $\chi$, we can use $\Omega$ to construct a stronger specification and thereby reduce  the search space.

Our synthesis procedure starts by invoking a procedure call {\sc SynthesizeAngelic} (line \ref{line:angelic}) which takes as input a specification that the returned program must satisfy under the angelic semantics. In particular, given a (ground) specification $\varphi$, {\sc SynthesizeAngelic} either returns failure or a program $P$ that \emph{angelically} satisfies $\varphi$ (i.e., $P \angelicmodels \varphi$). If the output is failure (meaning that there is no program in the search space that satisfies $\varphi$), {\sc SynthesizeAngelic} also returns an anti-specification (i.e., set of formulas) $\kappa$ that serves as an ``explanation" of why angelic synthesis failed. In particular, $\kappa$ has the property, for every $\psi \in \kappa$, there is no program in the search space that satisfies $\psi$. Thus, if {\sc SynthesizeAngelic} returns  $\mathsf{Failure}(\kappa)$, we  add $\kappa$ to  $\Omega$ (line \ref{line:failure1}).

In the extended example shown in Section~\ref{sec:extended-example}, there was a failure in Iteration 2a. Our underlying synthesizer would identify that this failure was due to the constraints $(\mlstinline{right_spine}(T_2) = [2]$ and $\mlstinline{right_spine} (T_2^R) = [2])$. Including this anti-specification would yield the following stronger specification for Iteration 2b:

\begin{equation}
\label{eq:spec2bprime}
\begin{array}{rcl}
\mlstinline{right_spine}(T_1) = [1;0] & & 
\mlstinline{right_spine}(T_2) = [2] \\[4pt]
\neg(\mlstinline{right_spine}(T_1^R) = [1;0] & \wedge & 
\mlstinline{right_spine}(T_2^R) = [2]) \\[4pt]
\neg(\mlstinline{right_spine}(T_2) = [2] & \wedge & 
\mlstinline{right_spine}(T_2^R) = [2])
\end{array}
\end{equation}

If {\sc SynthesizeAngelic} returns a program $P$, our synthesis procedure checks whether $P$ satisfies the specification $\groundSpec$ under the true semantics (line \ref{line:success1}). If so, then it  returns $P$ as a valid solution to the synthesis problem. Otherwise, it uses the witness $\witness$ to angelic satisfaction returned by {\sc SynthesizeAngelic} to construct a stronger specification. In particular, recall that such a witness $\witness$ encodes assumptions that an angelic execution makes in order to satisfy the specification. Thus, we strengthen the specification as $\chi \land \witness$ and  re-attempt synthesis by recursively invoking {\sc Synthesize} on this stronger specification (line \ref{line:witness}). If synthesis is successful, we return the resulting program as a solution (line \ref{line:success2}); otherwise, we backtrack and recursively invoke {\sc Synthesize} with the alternative specification $\chi \land \neg \witness $ (line \ref{line:failure2}), which ends up ruling out $\witness$ from the search space.  Observe that  the anti-specification $\Omega$ also grows during the recursive calls; thus, the second recursive call at line \ref{line:failure2} actually prunes more programs than just those satisfying $\witness$.

%\needsrev{Effectively, when we identify the angelic witness, we are propogating what information must be satisfied for that set of recursive calls to correspond to a correct program, then see if we can update the rest of the program to also satisfy them. When we negate the witness we are essentially saying ``these recursive calls do not work, lets find different ones.'' Thus, instead of performing a na\"ive search through possible programs (which is a very large search space), we are performing a search through possible sets of recursive calls. Thus in this search, we either use witnesses to perform small alterations to the code that make the recursive calls viable, or to invalidate a set of recursive calls and identify new ones.}

The following theorems state the soundness and completeness of our synthesis algorithm.

\begin{theorem}{\bf (Soundness)}\label{thm:sound}
If {\sc Synthesize}($\chi$) returns a program $P$, then we have $P \models \chi$.
\end{theorem}
\begin{proof}
Follows directly from line \ref{line:success1} of Algorithm~\ref{alg:core-recursive-synthesis}.
\end{proof}
\def\SynthesizeAngelic{{\sc SynthesizeAngelic}}
\begin{theorem}{ \bf (Completeness)}\label{thm:completeness}
If $\Synthesize(\chi)$ returns $\bot$, then there is no program that satisfies $\chi$ under the assumption that (1) \SynthesizeAngelic is complete, and (2) if \ \SynthesizeAngelic \ returns $\mathsf{Failure}(\kappa)$, then $\kappa$ satisfies the assumption from \autoref{eq:anti-spec}.
\end{theorem}
\begin{proof}
  The proof is in
  \ifappendices
  Section~\ref{sec:proofs}.
  \else the full version of the paper~\cite{burst-full}.
  \fi
\end{proof}

\paragraph{Remark.} A simpler alternative to the specification strengthening approach in Algorithm~\ref{alg:core-recursive-synthesis} would be to perform enumerative search over the angelic synthesis results as opposed to strengthening the specification. However, as we show empirically in Section~\ref{sec:evaluation}, this simpler alternative is not  as effective. In particular,  many programs that angelically satisfy the specification are wrong due to \emph{shared} incorrect assumptions about recursive calls; thus, our proposed algorithm allows ruling out many incorrect programs at the same time.

\section{Bottom-up Angelic Synthesis using Tree Automata}
\label{sec:tree-aut}

Recall that our top-level synthesis procedure (Algorithm~\ref{alg:core-recursive-synthesis}) uses a key procedure called {\sc SynthesizeAngelic} to find a program that satisfies the specification under the angelic semantics. In this section, we describe a realization of the angelic synthesis algorithm using bottom-up finite tree automata. Towards this goal, we first review tree automata basics and then describe the angelic synthesis algorithm.

\subsection{Tree Automata Preliminaries}
\label{subsec:ta-prelim}

A \emph{finite tree automaton} is a state machine that describes sets of trees~\cite{tata}. More formally, a finite tree automaton is defined as follows:

%We will use these FTAs to describe abstract syntax trees that satisfy the specifications. These automata are over a ``ranked alphabet'' $\alphabet$, where each symbol $f$ in alphabet $\alphabet$ has an arity associated with it. We use the notation $\alphabet_k$ to denote the function symbols of arity $k$.

\begin{definition} {\bf (FTA)}
A bottom-up finite tree automaton (FTA) over alphabet $\alphabet$ is a tuple $\fta = (\ftastates, \finalstates, \transitions)$ where $\ftastates$ is the set of states, $\finalstates \subseteq \ftastates$ are the final states, and $\transitions$ is a set of transitions of the form
$\ell(\ftastate_1, \dots, \ftastate_n) \rightarrow \ftastate$
where $\ftastate, \ftastate_1, \dots, \ftastate_n \in \ftastates$ and $\ell \in \alphabet$.
%We use the notation $\fta.\ftastates$ to refer to the set of all states of $\fta$. Similarly, we write $\fta.\finalstates$ to denote the set of all final states and $\fta.\transitions$ to denote the transitions.
\end{definition}

Following prior work~\cite{dace,blaze}, the alphabet $\alphabet$ in our context corresponds to constructs in the underlying programming language; FTA states correspond to a finite set of values (i.e., constants);  final states indicate values that satisfy a given specification; and transitions encode the semantics of the programming language. For instance,  a transition $+(1, 2) \rightarrow 3$ indicates that adding the integers $1$ and $2$ yields $3$. 

Since tree automata accept trees, we view each term over alphabet $\alphabet$ as a tree $T = (n, V, E)$ where $n$ is the root node, $V$ is a set of labeled vertices, and $E$ is the set of edges.
We say that a term $T$ is accepted by an FTA if we can rewrite $T$ to some state $q \in Q_f$ using transitions $\transitions$.
Finally, the language of a tree automaton $\fta$ is  denoted as $\mathcal{L}(\fta)$ and consists of the set of all terms accepted by~$\fta$.

\begin{example}
\label{ex:introfta}
Consider a tree automaton $\fta$ with states $\ftastates = \set{q_0, q_1}$, final states $\finalstates = \set{q_0}$, and the following transitions:
$$\transitions = \set{x() \rightarrow q_1, \mlstinline{xor}(q_i, q_i) \rightarrow q_0, \mlstinline{xor}(q_i, q_j) \rightarrow q_1 \text{ if } i \neq j}
$$
where $x$ has arity zero and $\mlstinline{xor}$ is a binary function.
$\fta$ accepts boolean equations combining \zlstinline{xor} and $x$, where the resulting boolean equation evaluates to \false{} when $x$ is initially \true{}.
\end{example}

Next, we define the notion of an \emph{accepting run} of an FTA:

\begin{definition} {\bf (Accepting run)}
An \emph{accepting run} of an FTA $\fta = (Q, Q_f, \transitions)$ is a pair $(T, L)$ where $T = (n, V, E)$ is a term that is accepted by $\fta$ and $L$ is a mapping from each node in $V$ to an FTA state such that the following conditions are satisfied:
\begin{enumerate}
    \item $L(n) \in Q_f$
    \item If $n$ has children $n_1, \ldots, n_k$ such that $L(n) = q$ and $L(n_1) = q_1, \ldots, L(n_k) = q_k$, then $\mathsf{Label}(n)(q_1, \ldots, q_k) \rightarrow q$ is a transition in $\transitions$.
\end{enumerate}
\end{definition} 

\begin{example}
\begin{figure}[!t]
    \centering
    \includegraphics[scale=0.6]{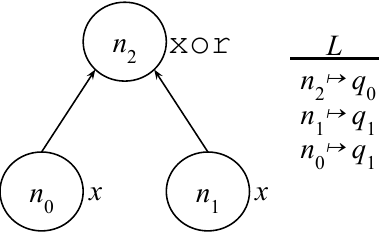}
    \caption{An example tree accepted by the automaton described in Example~\ref{ex:introfta}. The the tree and the associated mapping $L$ comprise an accepting run of the FTA.}
    \label{fig:example-norec}
\end{figure}
Figure~\ref{fig:example-norec} shows an accepting run $(T, L)$ over the FTA described in Example~\ref{ex:introfta}, where $T$ has nodes $\set{n_0,n_1,n_2}$ and edges from $n_0$ and $n_1$ to $n_2$. The labels of $n_0$ and $n_1$ are $x$ and the label of $n_2$ is \lstinline{xor}. This run is accepting since $L(n_2) = q_0 \in \finalstates$.
\iffalse
\needsrev{THIS NEEDS TO BE UPDATED: 
Figure~\ref{fig:example-norec} shows an accepting run $(T,L)$ over the FTA described in Example~\ref{ex:introfta}, where $T = (n_1,\set{n_0,n_1},\set{[] \to n_0,[n_0] \to n_1})$ and $L = \set{n_0 \mapsto q_0,n_1 \mapsto q_1}$. This run is accepting since $L(n_1) = q_1 \in \finalstates$ and $\mathsf{Label}(n_0) = x \wedge x() \to q_0 \in \transitions$ and $\mathsf{Label}(n_1) = \neg \wedge \neg(q_0) \to q_1 \in \transitions$.
}
\fi
\end{example}

\subsection{Angelic Synthesis Algorithm}

\begin{algorithm}[!t]
\begin{algorithmic}[1]
\Statex {\bf input:} A ground specification $\groundSpec$
\Statex {\bf output:} $\mathsf{Success}(P,\witness)$ for program $P$ and angelic witness $\witness$; $\mathsf{Failure}(\kappa)$ for anti-specification $\kappa$
\Procedure{SynthesizeAngelic}{$\groundSpec$}
\State $\kappa \gets \emptyset$
\ForEach{$\varphi \in \mathsf{DNFClauses}(\groundSpec)$}

\State $\text{first} \gets \mlstinline{true}$
\State $\psi \gets \mlstinline{true}$
\ForEach{$f(v) \  \emph{op} \ v' \in \varphi$}
\State $\psi \gets \psi \wedge f(v) \  \emph{op} \ v'$
\State $\fta \gets \BuildAngelicFTA(v,\varphi)$
\If{$\text{first}$} $\fta^* \gets \fta$
\State $\text{first} \gets \mlstinline{false}$
\Else{} $\fta^* \gets \Intersect(\fta^*, \fta)$
\EndIf
\If{$\LanguageOf{\fta^*} = \emptyset$} \textbf{break}
\EndIf

\EndForEach
\If{$\LanguageOf{\fta^*} \neq \emptyset$}
\State $(P,L) \gets \GetAcceptingRun(\fta^*)$
\State $\witness \gets \GetWitness(P,L)$
\State \Return $\mathsf{Success}(P,\witness)$
\Else{} $\kappa \gets \kappa \cup \set{\psi}$
\EndIf
\EndForEach
\State \Return $\mathsf{Failure}(\kappa)$
\EndProcedure
\end{algorithmic}
\caption{Angelic synthesis procecdure based on tree automata}
\label{alg:angelic-synthesis}
\end{algorithm}

In this section, we describe an FTA-based implementation of the {\sc SynthesizeAngelic} procedure that is invoked at line 2 of Algorithm~\ref{alg:core-recursive-synthesis}. This procedure, which is summarized  in Algorithm~\ref{alg:angelic-synthesis}, takes as input a ground specification $\chi$ and returns one of two things: If angelic synthesis is successful, the output is a program $P$ that angelically satisfies $\chi$, together with a witness $\witness$ to angelic satisfaction. On the other hand, if there is no program that angelically satisfies $\chi$, {\sc SynthesizeAngelic} returns a set of ground formulas $\kappa$ that serve as an anti-specification satisfying Equation~\ref{eq:anti-spec}.

The algorithm starts by converting the specification $\chi$ to disjunctive normal form (DNF) at line 3 and iterates over each of the DNF clauses (line 3--17). If there is a program $P$ that angelically satisfies  \emph{any} clause $\varphi$, then it returns $P$ (and its corresponding witness $\witness$) as a solution (line 16). On the other hand, if it exhausts all clauses without successfully finding a program, the algorithm returns $\mathsf{Failure}$ at line 18.

In more detail, each clause $\varphi$ is a conjunction of atomic predicates of the form $f(v) \opr v'$ where $v$ and $v'$ are constants (since $\chi$ is a ground specification). The nested loop at lines 6--12 iterates over all of these predicates, builds an FTA $\fta$ for each of them (line 8), and constructs a version space $\fta^*$ satisfying \emph{all} of them by taking the intersection of all FTAs (line 10).  If the language of the resulting automaton $\fta^*$ becomes empty (line 12), then this means there is no program satisfying the current clause so the algorithm moves on to the next clause.

On the other hand, if the final version space $\fta^*$ is non-empty after processing an entire clause (line 13), then we know that there exists a program that satisfies this clause under the angelic semantics. In this case, the algorithm finds an accepting run of this FTA, extracts a witness $\witness$ to angelic satisfaction by calling the {\sc GetWitness} procedure at line 15, and returns ``success" at line~16. 

Finally, if the algorithm exhausts all DNF clauses without finding a program, it returns $\mathsf{Failure}$ at line 18. In particular, the anti-specification $\kappa$ at line 18 consists of a set of \emph{unsynthesizable cores} (UC), where each UC $\psi$ is a conjunction of predicates such that there is no program $P$ that angelically satisfies $\psi$. Thus, it is always safe to strengthen the specification using the negation of an unsynthesizable core.

In the remainder of this subsection, we discuss the {\sc BuildAngelicFTA} and {\sc GetWitness } procedures in more detail.

\subsubsection{FTA Construction using Angelic Semantics}

\begin{figure}
\begin{mathpar}
\inferrule[Init]
{
}
{
q_{v_{in}} \in \ftastates\\
\hbox{\zlstinline|x|}() \to q_{v_{in}} \in \transitions
}

\inferrule[Final]
{
q_v \in \ftastates\\
\mathsf{SAT}(\varphi \land f(v_{in}) = v)
}
{
q_v \in \finalstates
}

\inferrule[Unit]
{
}
{
q_{\hbox{\zlstinline[basicstyle=\scriptsize\ttfamily]|unit|}} \in \ftastates\\
\hbox{\zlstinline|unit|}() \to q_{\hbox{\zlstinline[basicstyle=\scriptsize\ttfamily]|unit|}} \in \transitions
}

\inferrule[Pair]
{
q_{v_1} \in \ftastates\\
q_{v_2} \in \ftastates\\
}
{
q_{\texttt{(}v_1\texttt{,}v_2\texttt{)}} \in \ftastates\\
\texttt{(}\cdot{}\texttt{,}\cdot{}\texttt{)}(q_{v_1},q_{v_2}) \to q_{\texttt{(}v_1\texttt{,}v_2\texttt{)}} \in \transitions
}

\inferrule[Fst]
{
q_{(v_1,v_2)} \in \ftastates\\
}
{
q_{v_1} \in \ftastates\\
\hbox{\zlstinline|fst|}(q_{v_1,v_2}) \to q_{v_1} \in \transitions
}

\inferrule[Snd]
{
q_{(v_1,v_2)} \in \ftastates\\
}
{
q_{v_2} \in \ftastates\\
\hbox{\zlstinline|snd|}(q_{(v_1,v_2)}) \to q_{v_2} \in \transitions
}

\inferrule[Inl]
{
q_v \in \ftastates\\
}
{
q_{\hbox{\zlstinline[basicstyle=\scriptsize\ttfamily]|inl|}\,v} \in \ftastates\\
\hbox{\zlstinline|inl|}(v) \to q_{\hbox{\zlstinline[basicstyle=\scriptsize\ttfamily]|inl|}\,v} \in \transitions
}

\inferrule[Inr]
{
q_v \in \ftastates\\
}
{
q_{\hbox{\zlstinline[basicstyle=\scriptsize\ttfamily]|inr|}\,v} \in \ftastates\\
\hbox{\zlstinline|inr|}(v) \to q_{\hbox{\zlstinline[basicstyle=\scriptsize\ttfamily]|inr|}\,v} \in \transitions
}

\inferrule[Angelic Recursion]
{
\mathsf{SAT}(\varphi \land f(v_1)=v_2)\\
q_{v_1} \in \ftastates
}
{
q_{v_2} \in \ftastates\\
\hbox{\zlstinline|f|}(q_{v_1}) \to q_{v_2} \in \transitions
}

\inferrule[Uneval]
{
}
{
\bot \in \ftastates
}

\inferrule[Uneval Prod]
{
\ell \in \alphabet
}
{
\bot \in \ftastates\\
\ell(\bot,\ldots,\bot) \to \bot
}

\inferrule[Switch Left]
{
q_{\hbox{\zlstinline[basicstyle=\scriptsize\ttfamily]|inl|}\ v_3} \in \ftastates\\
q_{v_1} \in \ftastates\\
}
{
\hbox{\zlstinline|switch|}(q_{\hbox{\zlstinline[basicstyle=\scriptsize\ttfamily]|inl|}\ v_3},q_{v_1},\wildcard) \to q_{v_1} \in \transitions
}

\inferrule[Switch Right]
{
q_{\hbox{\zlstinline[basicstyle=\scriptsize\ttfamily]|inr|}\ v_3} \in \ftastates\\
q_{v_2} \in \ftastates\\
}
{
\hbox{\zlstinline|switch|}(q_{\hbox{\zlstinline[basicstyle=\scriptsize\ttfamily]|inr|}\ v_3},\wildcard,q_{v_2}) \to q_{v_2} \in \transitions
}
\end{mathpar}
\caption{Inference rules for $\BuildAngelicFTA(v_{in},\varphi)$.}
\label{fig:ftacreationrules}
\end{figure}

We now explain {\sc BuildAngelicFTA} procedure that takes as input a value $v_\emph{in}$ and a DNF clause $\varphi$ and returns an FTA $\fta$ whose language is the set of all programs that angelically satisfy $\varphi$ on input $v_\emph{in}$. That is:
\[
P \in \mathcal{L}(\fta) \Longleftrightarrow P \angelicmodels_{v_\emph{in}} \varphi
\]

This procedure  is summarized in Figure~\ref{fig:ftacreationrules} using  inference rules that stipulate which states and transitions should be part of the constructed FTA.  In particular, states in the FTA are of the form $q_v$ where $v$ is a value that arises when angelically executing some program $P$ on input $v_\emph{in}$.  In addition, there is a special state $\bot$ that denotes the  value of expressions that are never evaluated on input $v_\emph{in}$. Since our language contains conditionals in the form of switch statements, this special state $\bot$ is useful for representing the unknown value of branches that are never evaluated during an execution.

Next, we explain each of the rules from Figure~\ref{fig:ftacreationrules} in more detail:

\begin{itemize}
    \item The first rule, labeled {\sc Init}, adds the state $q_{v_{in}}$ to the FTA and adds a transition $x() \rightarrow q_{v_{in}}$. Since $x$ represents the program input, this rule essentially corresponds to binding $x$ to value~$v_{in}$. 
    \item The next rule, labeled Final, marks the final states of the FTA. In particular, since we want the language of the FTA to be those programs that angelically satisfy $\varphi$ on input $v_{in}$, we mark a state $q_v$ as accepting if $f(v_{in}) = v$ is consistent with the given specification $\varphi$. 
\item The next four rules ({\sc Unit}, ..., {\sc Inr}) add FTA states and transitions for the different contructors in our language. For example, according to the {\sc Pair} rule, if $q_{v_1}$ and $q_{v_2}$ are FTA states, then we also add $q_{(v_1, v_2)}$ as a state of the FTA as well as a corresponding transition for the pair constructor.
\item The rule labeled {\sc Angelic Recursion} encodes angelic execution semantics. In particular, if $q_{v_1}$ is an FTA state, then we add a transition $f(q_{v_1}) \rightarrow q_{v_2}$ as long as the formula $f(v_1) = v_2$ is consistent with $\varphi$. 
%Hence, this rule  encodes the semantics of angelic recursion.
\item The last two rules encode the semantics of switch statements. In particular, there is a transition $\mlstinline{switch}(q_{v_0}, q_{v_1}, \bot) \rightarrow q_{v_1}$ (resp. $\mlstinline{switch}(q_{v_0}, \bot, q_{v_2}) \rightarrow q_{v_2}$)  iff the \zlstinline{inl} (resp. \zlstinline{inr}) branch of the switch statement is executed on $v_0$ and produces value $v_1$ (resp. $v_2$). As mentioned earlier, the special state $\bot$ encodes the unknown value of expressions that are not evaluated, and the {\sc Uneval Prod} rule is used to propagate such ``unevaluated" values. 
\end{itemize}

\begin{comment}
\begin{example}
Consider the input of tree $t_1$, we can use the {\sc Init} rule to include a ``$t_1$'' tree in the FTA.
\[
\inferrule[Init]
{
}
{
q_{t_{1}} \in \ftastates\\
\hbox{\zlstinline|x|}() \to q_{t_{1}} \in \transitions
}
\]
With this rule, the FTA would look like this:

\begin{tikzpicture}[shorten >=1pt,node distance=2cm,on grid,auto]
\node[state,initial,initial text = {\zlstinline{x}}] (q_0) {$q_{t_1}$};
\end{tikzpicture}
\end{example}

\begin{example}
Consider the inputs of tree either $\mlstinline{inl}\ v_3$ or $\mlstinline{inr}\ v_3$, we can use the {\sc Switch} rules to include a ``$t_1$'' or ``$t_2$'' tree in the FTA.
With this rule, the FTA would look like this:

\begin{tikzpicture}[shorten >=1pt,node distance=3cm,on grid,auto,scale=1,state/.style={circle, draw, minimum size=1.5cm}]
\node[state] (q0) {$q_{\hbox{\zlstinline[basicstyle=\scriptsize\ttfamily]|inl|}\ v_3}$};
\node[state, right of=q0] (q1) {$q_{t_1}$};
\node[state, right of=q1] (q2) {$q_{\hbox{\zlstinline[basicstyle=\scriptsize\ttfamily]|inr|}\ v_3}$};
\node[state, right of=q2] (q3) {$q_{t_2}$};
\path[->]
    (q0) edge node {\zlstinline|switch|} (q1)
    (q2) edge node {\zlstinline|switch|} (q3);
\end{tikzpicture}
\end{example}
\end{comment}

The following theorem states the soundness of our angelic synthesis procedure for a specific input $v_\emph{in}$:

\begin{theorem}
\label{thm:65}
If {\sc BuildAngelicFTA}$(v_{in}, \varphi)$ returns $\fta$, then $P \angelicmodels_{v_{in}} \varphi$ for every $P \in \mathcal{L}(\fta) $.
\end{theorem}
\begin{proof}
  The proof is in
  \ifappendices
  Section~\ref{sec:proofs}.
  \else the full version of the paper.
  \fi
\end{proof}

The following theorem generalizes this from individual inputs to ground specifications:

\begin{theorem}
\label{thm:66}
Let $\varphi$ be a ground formula such that:
\[
V = \{ v_i \ | \ f(v_i) \in \mathsf{Terms}(\varphi) \}
\]
Then, if {\sc BuildAngelicFTA}$(v_{i}, \varphi)$ returns $\fta_i$ for inputs $V = \{ v_1, \ldots, v_n\}$, then, for every $P \in \mathcal{L}(\fta_1) \cap \ldots \cap \mathcal{L}(\fta_n) $, we have $P \angelicmodels \varphi$. 
\end{theorem}
\begin{proof}
  {As the definition of angelic satisfaction simply requires satisfaction on every individual input, this comes directly from Theorem~\ref{thm:65} and the definition of intersection.}
\end{proof}

%TODO(atn)
\subsubsection{Finding Witnesses to Angelic Satisfaction}
\label{subsec:find-witness}

\begin{figure}
\begin{mathpar}

\inferrule
{
P = (n, V, E) \\ 
\Label(n) = \hbox{\zlstinline|f|}\\
\children(n) = [n']\\
L \vdash (n', V, E) \rightsquigarrow \witness'\\
L(n) = (v_1,\ldots,v_k) \\
L(n') = (v_1',\ldots,v_k')}
{
L \vdash P \rightsquigarrow \witness'\ \wedge \bigwedge_{i \in [1\ldots{}k]} f(v_i') = v_i
}

\inferrule
{
P = (n, V, E)\\ \Label(n) \neq \hbox{\zlstinline|f|}\\
\children(n) = [n_1,\ldots,n_k]\\
\forall i\in[1\ldots{}k]. \ \ L\vdash (n_i, V, E) \rightsquigarrow \witness_i
}
{
L \vdash P \rightsquigarrow \wedge_{i \in [1\ldots{}k]} \witness_i
}

%\inferrule
%{
%L \vdash P \rightsquigarrow \witness
%}
%{
%\witness \in \GetWitness(P,L) 
%}
\end{mathpar}
\caption{Inference rules describing the {\sc GetWitness} procedure. }
\label{fig:get-assumptions}
\end{figure}

In this section, we describe our procedure for finding witnesses to angelic satisfaction. In particular, given an accepting run $(P, L)$ of the tree automaton where $P$ is a program (represented as an AST) and $L$ is a mapping from AST nodes to FTA states, {\sc GetWitness} returns a witness $\witness$ of the form $\bigwedge_i f(v_i) = v_i'$ that identifies all assumptions made during the angelic execution associated with labeling function $L$.

Before we explain the rules from Figure~\ref{fig:get-assumptions}, we note that $L$ maps each AST node to a tuple  $(v_1, \ldots, v_n)$ where each $v_i$ is a value. In particular, while the states for each individual FTA consist of individual values, recall that Algorithm~\ref{alg:angelic-synthesis} takes the intersection of several FTAs. Thus, after $n$ intersection operations, the states of the FTA correspond $n$-tuples of the form $(v_1, \ldots, v_n)$.

With this in mind, Figure~\ref{fig:get-assumptions} presents the {\sc GetWitness} procedure using  inference rules that derive judgments of the form $L \vdash P \rightsquigarrow \witness$. The meaning of this judgment is that $\witness$ is a witness to angelic satisfaction of $P$ in the angelic execution associated with labeling function $L$.  
The first rule in Figure~\ref{fig:get-assumptions} deals with recursive invocations of procedure $f$. In this case, the root node of the AST is a node $n$ labeled with $f$, and $n$  has a single child $n'$ (since $f$ takes a single argument). Now, suppose that $L$ maps $n$ to the tuple $(v_1, \ldots,  v_k)$ and $n'$ to $(v'_1, \ldots,  v'_k)$. Such a transition corresponds to the assumption that the recursive call to $f$ returns value $v_i$ on input $v_i'$. Thus, the resulting witness includes the conjunct $\bigwedge_i f(v_i') = v_i$. Furthermore, since the argument to $f$ can contain nested recursive calls, this rule also computes a witness $\witness'$ for the sub-AST rooted at $n'$ (i.e., $L \vdash (n', V, E) \rightsquigarrow \witness'$). The final witness is therefore the conjunction of $\witness'$ and $\bigwedge_i f(v_i') = v_i$.

The second rule in Figure~\ref{fig:get-assumptions} deals with the scenario where the top-level expression is \emph{not} a recursive call to $f$. However, since the sub-expressions may contain recursive calls, we recurse down to the children and obtain witnesses for the sub-expressions. The resulting witness is the conjunction of witnesses for all sub-expressions. 

\section{Implementation}
\label{sec:implementation}

We have implemented our proposed technique in a tool called \toolname that is implemented in OCaml. 
In this section, we discuss some important implementation details and optimizations omitted from the technical development.
%In addition to our core algorithm, we addressed the additional complexities of \emph{finitization} and \emph{minimality} in our implementation. Additionally, our implementation included a number of \emph{optimizations}.

\subsection{Termination of Synthesized Programs}
\label{subsec:termination}
{To ensure that our synthesized programs terminate, our implementation utilizes  a well-founded default ordering  on our values. In particular, \toolname ensures that it generates terminating programs by only permitting recursive calls on values that are \emph{strictly smaller} than the input. If $v_\emph{in}$ is provided as an input, recursive calls} to \lstinline{f} can only be applied to values $v$ when $v \prec v_\emph{in}$. This prevents generating infinite loops like \lstinline{let rec f(x) = f(x+1)}.

\subsection{Finitization of States}
Recall that our angelic synthesis technique uses finite tree automata to find a program that satisfies the specification under the angelic semantics. Further, recall that states in the tree automaton correspond to concrete values, of which there may be infinitely many. Similar to prior work~\cite{dace},  our implementation bounds the number of automaton states  using a parameter $k$ that controls the number of applications of the inference rules from Figure~\ref{fig:ftacreationrules}. By default, this parameter is set to 4. 

Another complication in our setting is due to the use of angelic recursion in the inference rules in Figure~\ref{fig:ftacreationrules}. In particular, under the angelic semantics, a recursive call can return any value that satisfies the specification. If the specification is \emph{true}, there are infinitely many concrete values that satisfy it. \Burst{} gets around this issue by iteratively constructing FTAs from smallest input to largest, and finitizing as it goes.

For example, consider trying to synthesize a program $f$ with the ground specification $(f(0) \geq 0 \wedge f(0) \text{ mod } 2 = 0) \wedge (f(1) \geq 1 \wedge f(1) \text{ mod } 2 = 0)$.
The smallest number involved in this ground specification is 0. As zero is the smallest element in the naturals, \Burst{} cannot make any recursive calls, so there is no need to worry about finitizing the outputs of recursive calls. Thus, \Burst{} can create $\fta_0$ (the automaton corresponding to input 0), which has final states 0, 2, and 4. Next, when creating $\fta_1$ (the automaton corresponding to input 1), \Burst{} only needs to consider the recursive call $f(0)$, which can only return 0, 2 and 4, as these are the only final states of $\fta_0$. In general, \Burst{} only constructs $\fta_n$ (the automaton on input n) after constructing $\fta_0, \ldots, \fta_{n-1}$; it then uses their final states to determine the possible values of the recursive calls.

\begin{comment}
Our approach relies on processing \emph{finite} tree automata to find programs. We perform a standard approach to finitization: we parameterize our core algorithm by an additional integer $i$ that corresponds to maximal depth of any deduction made by the inference rules in Figure~\ref{fig:ftacreationrules}. In other words, given $i=2$, one could not deduce $\mlstinline{inl (inl unit)} \in \ftastates$ with just the {\sc Inl} and {\sc Unit} rules, as it requires a minimum inference depth of 3.

Our technique has an additional source of trickiness -- how do we finitize our
possible recursive calls? If there is no constraint on the argument to the
recursive function, the recursive call could return anything! Whenever we reach
a possible recursive call \lstinline{f($v$)}, we simply lookup the FTA associated with input $v$. The possible results are the final values of that FTA.
\end{comment}

\subsection{Program Selection}

In general, there may be many programs that satisfy the given specification, and most synthesis algorithms use heuristics to choose which program to return to the user. One of these heuristics is to prefer smaller programs, and, inspired by the effectiveness of this heuristic in prior work~\cite{myth,smyth,lambda2}, \toolname  also returns the smallest program in terms of AST size. However, to provide such a minimality guarantee, our implementation slightly deviates from the core synthesis procedure shown in Algorithm~\ref{alg:core-recursive-synthesis}.

In particular, Algorithm~\ref{alg:core-recursive-synthesis} makes recursive calls to two distinct strengthened specifications -- one to $\groundSpec \wedge \witness$ (line 10) and one to $\groundSpec \wedge \neg\witness$ (line 11). In this algorithm, the call with input $\groundSpec \wedge \neg\witness$ is made only after the call with input $\groundSpec \wedge \witness$ fails. Our actual implementation maintains a priority queue over these specifications sorted according to the size of the minimal solution for the corresponding angelic synthesis problem. It then explores these programs from smallest to largest. Thus, our implementation guarantees that the program returned to the user is the smallest one among those that satisfy the specification.
\ifappendices{The interested reader can see the pseudo-code  in Section~\ref{sec:additional-algs} of the appendix.}\fi

\begin{comment}
The algorithm that ensures that $\mathsf{GetAcceptingRun}$ returns optimal results is the same as is used in \ToolText{Blaze}, it is an adopted version of the minimal weight B-path algorithm introduced in \citet{gallo}. Being expressed as a minimal weight B-path is also the primary limiter for what metrics can be used -- if we can express our metric as a minimal weight B-path, it is a valid metric.

For our overall algorithm to return the optimal program, we update our program to use \emph{Dijkstra's Algorithm} for search, where it is currently using a \emph{Depth-First Search}. In particular, given a specification $\groundSpec$ and a witness $\witness$, we first try to see if there are any programs that satisfy $\groundSpec \wedge \witness$. Only after this fails do we try to find a program that satisfies $\groundSpec \wedge \neg\witness$. To return an optimal program, we must search through these updated specifications concurrently. Each specification is given a priority according to the cost of the minimal program that angelically satisfies it. These specifications are put in a priority queue, and processed according to their priority. When a minimal priority specification is found where the associated program satisfies the specification with standard semantics, that program is returned.
\end{comment}

\subsection{Improving the CEGIS Loop}\label{sec:opt-cegis}

Recall that our technique can perform synthesis from logical specifications by integrating our proposed approach within a CEGIS loop. While the standard CEGIS paradigm only uses ground formulas for inductive synthesis, our approach can actually utilize the original logical specification when performing angelic synthesis. In particular, when deciding which values can be returned by a recursive call, our implementation utilizes the original logical specification as opposed to the weaker ground specifications. For example, suppose that the original specification is $f(x) < x$ and our current counterexamples include $3$ and $4$ (i.e., ground specification is $f(3) < 3 \land f(4) < 4$). While the ground specification does not constrain the output of recursive call $f(2)$, we can use the original specification to constrain the return value of $f(2)$ to be either $0$ or $1$ (assuming that $x$ is a natural number). 

\subsection{Optimizations}
Our implementation also utilizes a few standard optimizations described in prior synthesis literature. Since it is common to perform type-directed pruning in synthesis~\cite{myth,lambda2},  we construct  our FTAs to only accept well-typed programs. Inspired by prior work that utilizes eta-long beta-normal form~\cite{myth,myth2,smyth}, we also modify our FTA construction rules to only accept such normalized programs.

\section{Evaluation}
\label{sec:evaluation}

In this section, we describe a series of experimental evaluations that are designed to answer the following research questions:

\begin{enumerate}
\item[\textbf{\emph{RQ1.}}] Is \Burst able to effectively synthesize programs from a variety of different specifications?
\item[\textbf{\emph{RQ2.}}] How does \Burst compare against prior work in terms of synthesis efficiency and correctness of synthesized programs?
\item[\textbf{\emph{RQ3.}}] How important is it to combine angelic synthesis with specification strengthening?
\end{enumerate}

All experiments described in this section are performed on a 2.5 GHz Intel Core i7 processor with 16 GB of 1600 MHz DDR3 running macOS Big Sur with a time limit of 120 seconds.

\subsection{Benchmarks and Baselines}

To answer the  research questions listed above, we use \toolname to synthesize 45 recursive functional programs from prior work~\cite{myth,smyth} and compare it against the following baselines:

\begin{enumerate}[leftmargin=*]
    %\item {\bf Myth}~\cite{myth}, a top-down type-directed synthesizer that only works with \emph{trace complete} input-output examples.
    \item  {\bf SMyth}~\cite{smyth}, which is a top-down type-directed programming-by-example tool. In particular, \Smyth generalizes \Myth~\cite{myth} to handle input-output examples that are not trace-complete.
    \item {\bf Synquid}~\cite{synquid}, which performs synthesis from liquid types.
    \item {\bf Leon}~\cite{leon}, which is a synthesizer that performs synthesis from logical specifications.
\end{enumerate}

\subsection{Specifications}
To evaluate whether \toolname can handle a variety of different specifications and to compare it against different tools, we consider three classes of specifications for each of our 45 benchmarks:

\begin{enumerate}[leftmargin=*]
    \item {\bf IO:} These are input-output examples written by developers of  \Smyth~\cite{smyth}.
    \item {\bf Ref:} These are reference implementations written by us. 
\item {\bf Logical:} These are logical specifications that specify pre- and post-conditions (or, in the case of \Synquid, refinement types) on the function to be synthesized.
\end{enumerate}

While \toolname can perform synthesis from all three classes of specifications listed above, not all baselines can effectively handle these different specifcations. Thus, we only compare against \Smyth on the {\bf IO} and {\bf Ref} specifications and against \Leon and \Synquid for the {\bf Logical} specifications. Note that  \Smyth can be adapted to perform synthesis from a reference implementation by integrating it inside a CEGIS loop and obtaining input-output examples from the reference implementation. Furthermore, while \Leon and \Synquid can, \emph{in principle}, handle {\bf IO} specifications, prior work has shown that they are not effective when used for this purpose~\cite{smyth}. Thus, we only compare against \Leon and \Synquid on the {\bf Logical} specifications.

\subsection{Synthesis from Input/Output Specifications}
\label{subsec:synth-io}

\begin{figure}[!h]
    \centering
    \footnotesize
    \begin{tabular}{l|ccc|ccc}%
    \bfseries Test & \multicolumn{3}{c}{\bfseries \Burst} & \multicolumn{3}{c}{\bfseries \Smyth}\\
    & Time (s) & Correct? & Size & Time (s) & Correct? & Size
    \begin{filecontents*}{examples.csv}
,SimpleMinEltMax,SmythIsectMax,NonincrementalInitialCreationMax,SmythFullSynth,IsectMax,RandomMinEltTotal,SmythInitialCreationTotal,MinifyMax,NonincrementalInitialCreationTotal,Test,InitialCreationMax,NonincrementalLoopCount,RandomMinifyTotal,NonincrementalMinifyTotal,SmythComputationTime,IsectTotal,FullSynth,NonincrementalIsectMax,SmythIsectTotal,RandomInitialCreationMax,SmythMinifyMax,SimpleComputationTime,SimpleMinEltTotal,NonincrementalMinifyMax,SmythInitialCreationMax,SmythMinEltTotal,RandomIsectMax,SimpleFullSynthMax,SmythMinEltMax,NonincrementalAcceptsTermMax,SmythLoopCount,NonincrementalMinEltMax,RandomMinEltMax,SimpleLoopCount,AcceptsTermMax,NonincrementalIsectTotal,LoopCount,RandomComputationTime,AcceptsTermTotal,MinEltTotal,RandomAcceptsTermMax,NonincrementalMinEltTotal,FullSynthMax,SmythAcceptsTermTotal,SimpleFullSynth,RandomInitialCreationTotal,SimpleAcceptsTermMax,RandomFullSynthMax,SimpleInitialCreationMax,SimpleMinifyMax,MinEltMax,NonincrementalFullSynth,MinifyTotal,NonincrementalFullSynthMax,NonincrementalComputationTime,RandomIsectTotal,SmythAcceptsTermMax,NonincrementalAcceptsTermTotal,RandomLoopCount,ComputationTime,SimpleIsectTotal,SimpleMinifyTotal,RandomMinifyMax,SmythFullSynthMax,SimpleAcceptsTermTotal,InitialCreationTotal,SmythMinifyTotal,RandomAcceptsTermTotal,RandomFullSynth,SimpleIsectMax,SimpleInitialCreationTotal,Correct,SmythCorrect,Size,SmythSize
0,0,0,0,0,0,0,0,0,0,bool-always-false,0,0,0,0,0.04,0,0,0,0,0,0,0.04,0,0,0,0,0,0,0,0,0,0,0,0,0,0,0,0.04,0,0,0,0,0,0,0,0,0,0,0,0,0,0,0,0,0.04,0,0,0,0,0.02,0,0,0,0,0,0,0,0,0,0,0,\correct,\correct,4,4
1,0,0,0,0,0,0,0,0,0,bool-always-true,0,0,0,0,0.04,0,0,0,0,0,0,0.04,0,0,0,0,0,0,0,0,0,0,0,0,0,0,0,0.04,0,0,0,0,0,0,0,0,0,0,0,0,0,0,0,0,0.02,0,0,0,0,0.04,0,0,0,0,0,0,0,0,0,0,0,\correct,\correct,4,4
2,0,0,0,0,0,0,0,0,0,bool-band,0,0,0,0,0.04,0,0,0,0,0,0,0.04,0,0,0,0,0,0,0,0,0,0,0,0,0,0,0,0.04,0,0,0,0,0,0,0,0,0,0,0,0,0,0,0,0,0.04,0,0,0,0,0.04,0,0,0,0,0,0,0,0,0,0,0,\correct,\correct,14,14
3,0,0,0,0,0,0,0,0,0,bool-bor,0,0,0,0,0.04,0,0,0,0,0,0,0.04,0,0,0,0,0,0,0,0,0,0,0,0,0,0,0,0.04,0,0,0,0,0,0,0,0,0,0,0,0,0,0,0,0,0.04,0,0,0,0,0.04,0,0,0,0,0,0,0,0,0,0,0,\correct,\correct,14,14
4,0,0,0,0,0,0,0,0,0,bool-impl,0,0,0,0,0.03,0,0,0,0,0,0,0.02,0,0,0,0,0,0,0,0,0,0,0,0,0,0,0,0.02,0,0,0,0,0,0,0,0,0,0,0,0,0,0,0,0,0.03,0,0,0,0,0.04,0,0,0,0,0,0,0,0,0,0,0,\correct,\correct,13,14
5,0,0,0,0,0,0,0,0,0,bool-neg,0,0,0,0,0.02,0,0,0,0,0,0,0.02,0,0,0,0,0,0,0,0,0,0,0,0,0,0,0,0.02,0,0,0,0,0,0,0,0,0,0,0,0,0,0,0,0,0.02,0,0,0,0,0.03,0,0,0,0,0,0,0,0,0,0,0,\correct,\correct,10,12
6,0,0,0,0,0,0,0,0,0,bool-xor,0,0,0,0,0.04,0,0,0,0,0,0,0.04,0,0,0,0,0,0,0,0,0,0,0,0,0,0,0,0.04,0,0,0,0,0,0,0,0,0,0,0,0,0,0,0,0,0.04,0,0,0,0,0.04,0,0,0,0,0,0,0,0,0,0,0,\correct,\correct,22,22
7,0,0,0,0,0,0,0,0,0.08,list-append,0,0,0,0.01,0.04,0.01,0,0,0,0,0,0.09,0.03,0,0,0,0,0,0,0,0,0,0,0,0,0.01,0,0.08,0,0.01,0,0.02,0,0,0,0.01,0,0,0,0,0,0,0.01,0,0.19,0,0,0,0,0.19,0.01,0.01,0,0,0,0.10,0,0,0,0,0.02,\correct,\correct,31,24
8,0.04,\incorrect,0.03,\incorrect,0.13,1.09,\incorrect,0.04,0.12,list-compress,0.03,0,4.14,0.18,\incorrect,0.29,0,0.13,\incorrect,0.06,\incorrect,1.72,0.10,0.04,\incorrect,\incorrect,0.77,0,\incorrect,0,\incorrect,0.08,0.27,0,0,0.50,0,13.67,0,0.17,0,0.29,0,\incorrect,0,1.58,0,0,0.03,0.06,0.09,0,0.41,0,1.22,4.33,\incorrect,0,0,1.12,1,0.37,0.23,\incorrect,0,0.08,\incorrect,0,0,0.14,0.12,\correct,\na,63,\na
9,0,0,0.01,0,0,0.01,0,0,0.02,list-concat,0.01,0,0,0,0.04,0,0,0,0,0.01,0,0.08,0,0,0,0,0,0,0,0,0,0.01,0.01,0,0,0,0,0.08,0,0.01,0,0.01,0,0,0,0.01,0,0,0.01,0,0.01,0,0,0,0.07,0,0,0,0,0.07,0.01,0,0,0,0,0.02,0,0,0,0,0.01,\correct,\correct,21,21
10,0.09,0,0.01,0,0.07,0.01,0,0.06,0.36,list-drop,0.02,0,0.01,0.39,0.08,0.27,0,0.07,0,0.01,0,0.55,0.17,0.07,0,0,0,0,0,0,0,0.30,0.01,0,0,0.42,0,0.13,0,1.02,0,1.53,0,0,0,0.03,0,0,0.01,0.06,0.32,0,0.23,0,2.98,0.01,0,0,0,1.90,0.13,0.12,0,0,0,0.18,0,0,0,0.08,0.07,\correct,\correct,24,27
11,0,0,0,0,0,0,0,0,0.02,list-even-parity,0,0,0.01,0.01,0.07,0.01,0,0,0,0,0,0.08,0,0,0,0,0,0,0,0,0,0,0,0,0,0.01,0,0.08,0,0.01,0,0.01,0,0,0,0.01,0,0,0,0,0,0,0.01,0,0.12,0.01,0,0,0,0.07,0.01,0.01,0,0,0,0.01,0,0,0,0,0.02,\correct,\incorrect,32,36
12,\na,0,0.01,0,0.01,0.01,0,0,0.39,list-filter,0,0,0.01,0.09,0.12,0.13,0,0.01,0,0,0,\incorrect,\na,0,0,0,0.01,\na,0,0,0,0.02,0.01,\na,0,0.13,0,0.18,0,0.14,0,0.13,0,0,\na,0.04,\na,0,\na,\na,0.02,0,0.09,0,0.87,0.02,0,0,0,0.82,\na,\na,0,0,\na,0.37,0,0,0,\na,\na,\correct,\incorrect,56,41
13,0.01,0,0.02,0,0.02,0.02,0,0.01,0.05,list-fold,0.02,0,0.02,0.02,1.49,0.03,0,0.02,0,0.02,0,0.19,0.01,0.01,0,0,0.02,0,0,0,0,0.02,0.02,0,0,0.03,0,0.19,0,0.02,0,0.02,0,0,0,0.04,0,0,0.02,0.01,0.02,0,0.02,0,0.20,0.03,0,0,0,0.20,0.04,0.02,0.01,0,0,0.04,0,0,0,0.02,0.05,\incorrect,\incorrect,42,30
14,0,0,0,0,0,0,0,0,0.01,list-hd,0,0,0,0,0.04,0,0,0,0,0,0,0.04,0,0,0,0,0,0,0,0,0,0,0,0,0,0,0,0.04,0,0,0,0,0,0,0,0.01,0,0,0,0,0,0,0,0,0.04,0,0,0,0,0.04,0,0,0,0,0,0.01,0,0,0,0,0.01,\correct,\correct,12,13
15,\incorrect,0,0,0,0,0.94,0,0,0.03,list-inc,0.01,0,0.15,0,0.07,0,0,0,0,0.11,0,\incorrect,\incorrect,0,0,0,0.09,\incorrect,0,0,0,0,0.80,\incorrect,0,0,0,1.93,0,0,0,0,0,0,\incorrect,0.48,\incorrect,0,\incorrect,\incorrect,0,0,0,0,0.08,0.15,0,0,0,0.07,\incorrect,\incorrect,0.07,0,\incorrect,0.03,0,0,0,\incorrect,\incorrect,\correct,\correct,20,10
16,0,0,0.01,0,0,0,0,0,0.02,list-last,0,0,0.01,0,0.04,0,0,0,0,0,0,0.07,0,0,0,0,0.01,0,0,0,0,0,0,0,0,0,0,0.07,0,0,0,0,0,0,0,0.02,0,0,0,0,0,0,0,0,0.07,0.01,0,0,0,0.08,0,0,0,0,0,0.01,0,0,0,0,0.01,\correct,\correct,26,27
17,0,0,0,0,0,0,0,0,0,list-length,0,0,0,0,0.03,0,0,0,0,0,0,0.04,0,0,0,0,0,0,0,0,0,0,0,0,0,0,0,0.04,0,0,0,0,0,0,0,0,0,0,0,0,0,0,0,0,0.04,0,0,0,0,0.04,0,0,0,0,0,0,0,0,0,0,0,\correct,\correct,15,16
18,0,0,0,0,0,0,0,0,0.02,list-map,0,0,0.01,0,0.49,0,0,0,0,0,0,20.01,0,0,0,0,0,0,0,0,0,0,0,0,0,0,0,0.14,0,0,0,0,0,0,0,0.06,0,0,0,0,0,0,0,0,0.08,0.01,0,0,0,0.08,0,0,0,0,0,0.01,0,0,0,0,0.01,\correct,\incorrect,35,44
19,\incorrect,0,0.01,0,0.22,0.16,0,0.21,0.27,list-nth,0,0,0.32,1.58,0.07,0.60,0,0.22,0,0.01,0,\incorrect,\incorrect,0.23,0,0,0.22,\incorrect,0,0,0,0.12,0.14,\incorrect,0,1.63,0,0.85,0,0.33,0,0.93,0,0,\incorrect,0.02,\incorrect,0,\incorrect,\incorrect,0.12,0,0.62,0,4.99,0.25,0,0,0,1.97,\incorrect,\incorrect,0.21,0,\incorrect,0.08,0,0,0,\incorrect,\incorrect,\correct,\correct,35,29
20,\incorrect,0,0.10,0,0.06,0.41,0,0.03,0.39,list-pairwise-swap,0.10,0,0.07,0.05,0.31,0.08,0,0.05,0,0.12,0,\incorrect,\incorrect,0.02,0,0,0.05,\incorrect,0,0,0,0.66,0.39,\incorrect,0,0.09,0,1.59,0,0.74,0,0.76,0,0,\incorrect,0.92,\incorrect,0,\incorrect,\incorrect,0.66,0,0.05,0,1.37,0.08,0,0,0,13.10,\incorrect,\incorrect,0.02,0,\incorrect,0.28,0,0,0,\incorrect,\incorrect,\correct,\correct,53,47
21,\incorrect,0,0.04,0,0.02,0.36,0,0.01,0.52,list-rev-append,0.04,0,1.18,0.09,0.08,0.17,0,0.02,0,0.86,0,\incorrect,\incorrect,0.01,0,0,0.06,\incorrect,0,0,0,0.09,0.03,\incorrect,0,0.11,0,33.60,0,0.92,0,0.75,0,0,\incorrect,30.26,\incorrect,0,\incorrect,\incorrect,0.09,0,0.14,0,1.55,0.90,0,0,0,2.01,\incorrect,\incorrect,0.05,0,\incorrect,0.70,0,0,0,\incorrect,\incorrect,\correct,\correct,24,24
22,0,0,0.02,0,0,0.02,0,0,0.03,list-rev-fold,0.02,0,0,0,0.07,0,0,0,0,0.02,0,0.07,0,0,0,0,0,0,0,0,0,0,0.02,0,0,0,0,0.13,0,0,0,0,0,0,0,0.05,0,0,0.02,0,0,0,0,0,0.08,0,0,0,0,0.07,0,0,0,0,0,0.03,0,0,0,0,0.03,\correct,\correct,10,10
23,0.07,0,0.06,0,0.06,14.72,0,0.04,0.12,list-rev-snoc,0.05,0,0.39,0.05,0.04,0.06,0,0.06,0,0.07,0,0.90,0.07,0.04,0,0,0.60,0,0,0,0,1.95,14.71,0,0,0.06,0,16.12,0,1.92,0,1.95,0,0,0,0.11,0,0,0.07,0.11,1.92,0,0.05,0,2.26,0.61,0,0,0,2.22,0.40,0.26,0.36,0,0,0.11,0,0,0,0.17,0.14,\correct,\correct,20,20
24,\incorrect,0,0.01,0,0,0,0,0,0.18,list-rev-tailcall,0.01,0,0,0.01,0.05,0.01,0,0,0,0.01,0,\incorrect,\incorrect,0,0,0,0,\incorrect,0,0,0,0,0,\incorrect,0,0.01,0,0.09,0,0.01,0,0.01,0,0,\incorrect,0.02,\incorrect,0,\incorrect,\incorrect,0,0,0.01,0,0.25,0,0,0,0,0.14,\incorrect,\incorrect,0,0,\incorrect,0.08,0,0,0,\incorrect,\incorrect,\incorrect,\correct,45,24
25,\incorrect,0,0.46,0,0.01,0,0,0.07,2.10,list-snoc,0.39,0,0.10,0.22,0.05,0.02,0,0.01,0,0.37,0,\incorrect,\incorrect,0.08,0,0,0.01,\incorrect,0,0,0,0.28,0,\incorrect,0,0.02,0,0.86,0,0.42,0,0.47,0,0,\incorrect,0.74,\incorrect,0,\incorrect,\incorrect,0.26,0,0.21,0,2.80,0.01,0,0,0,2.61,\incorrect,\incorrect,0.05,0,\incorrect,1.98,0,0,0,\incorrect,\incorrect,\correct,\correct,36,28
26,0.01,0,0.01,0,0.01,0.05,0,0.01,0.05,list-sort-sorted-insert,0.01,0,0.01,0,0.05,0.01,0,0,0,0.01,0,20.46,0.01,0,0,0,0.01,0,0,0,0,0,0.03,0,0,0,0,0.14,0,0.02,0,0,0,0,0,0.04,0,0,0.01,0.01,0.02,0,0.01,0,0.14,0.01,0,0,0,0.14,0.02,0.01,0.01,0,0,0.04,0,0,0,0.01,0.05,\correct,\correct,20,20
27,\incorrect,0,\incorrect,0,\incorrect,\incorrect,0,\incorrect,\incorrect,list-sorted-insert,\incorrect,\incorrect,\incorrect,\incorrect,1.64,\incorrect,\incorrect,\incorrect,0,\incorrect,0,\incorrect,\incorrect,\incorrect,0,0,\incorrect,\incorrect,0,\incorrect,0,\incorrect,\incorrect,\incorrect,\incorrect,\incorrect,\incorrect,\incorrect,\incorrect,\incorrect,\incorrect,\incorrect,\incorrect,0,\incorrect,\incorrect,\incorrect,\incorrect,\incorrect,\incorrect,\incorrect,\incorrect,\incorrect,\incorrect,\incorrect,\incorrect,0,\incorrect,\incorrect,\incorrect,\incorrect,\incorrect,\incorrect,0,\incorrect,\incorrect,0,\incorrect,\incorrect,\incorrect,\incorrect,\na,\correct,\na,52
28,\incorrect,0,0.53,0,0,0,0,0,6.72,list-stutter,0.57,0,0,0.01,0.04,0.01,0,0,0,0.54,0,\incorrect,\incorrect,0,0,0,0,\incorrect,0,0,0,0.01,0,\incorrect,0,0.01,0,2.06,0,0.02,0,0.02,0,0,\incorrect,1.99,\incorrect,0,\incorrect,\incorrect,0,0,0.01,0,6.86,0,0,0,0,6.44,\incorrect,\incorrect,0,0,\incorrect,6.28,0,0,0,\incorrect,\incorrect,\correct,\correct,25,25
29,0,0,0.03,0,0,0,0,0,0.13,list-sum,0.02,0,0.01,0.01,0.08,0,0,0,0,0.03,0,0.13,0,0,0,0,0,0,0,0,0,0,0,0,0,0,0,0.17,0,0.01,0,0.01,0,0,0,0.08,0,0,0.02,0,0,0,0.01,0,0.23,0,0,0,0,0.19,0,0.01,0,0,0,0.10,0,0,0,0,0.06,\correct,\correct,10,10
30,4.18,0,\incorrect,0,\incorrect,0.12,0,\incorrect,\incorrect,list-take,\incorrect,\incorrect,0.15,\incorrect,0.07,\incorrect,\incorrect,\incorrect,0,0.01,0,58.98,18.02,\incorrect,0,0,0.13,0,0,\incorrect,0,\incorrect,0.11,0,\incorrect,\incorrect,\incorrect,0.49,\incorrect,\incorrect,0,\incorrect,\incorrect,0,0,0.05,0,0,0.04,10.47,\incorrect,\incorrect,\incorrect,\incorrect,\incorrect,0.13,0,\incorrect,0,\incorrect,19.22,20.90,0.12,0,0,\incorrect,0,0,0,9.57,0.22,\na,\correct,\na,33
31,0,0,0,0,0,0,0,0,0.01,list-tl,0,0,0,0,0.03,0,0,0,0,0,0,0.03,0,0,0,0,0,0,0,0,0,0,0,0,0,0,0,0.04,0,0,0,0,0,0,0,0,0,0,0,0,0,0,0,0,0.05,0,0,0,0,0.04,0,0,0,0,0,0,0,0,0,0,0.01,\correct,\correct,11,13
32,0,0,0,0,0.01,0,0,0.01,0.14,nat-add,0,0,0,0.03,0.04,0.03,0,0.01,0,0,0,0.08,0.01,0.01,0,0,0,0,0,0,0,0.01,0,0,0,0.04,0,0.08,0,0.03,0,0.04,0,0,0,0.03,0,0,0,0,0.01,0,0.03,0,0.29,0,0,0,0,0.26,0.01,0,0,0,0,0.10,0,0,0,0,0.03,\correct,\correct,22,18
33,0.01,0,0,0,0,0,0,0,0.01,nat-iseven,0,0,0,0,0.04,0,0,0,0,0,0,0.44,0.35,0,0,0,0,0,0,0,0,0,0,0,0,0,0,0.04,0,0,0,0,0,0,0,0.01,0,0,0,0,0,0,0,0,0.08,0,0,0,0,0.07,0,0,0,0,0,0.01,0,0,0,0,0.01,\incorrect,\correct,18,22
34,0.03,0,0,0,0.06,0.14,0,0.03,0.05,nat-max,0.01,0,0.25,0.07,0.13,0.10,0,0.07,0,0.01,0,0.38,0.04,0.03,0,0,0.22,0,0,0,0,0.03,0.05,0,0,0.15,0,1.23,0,0.04,0,0.06,0,0,0,0.06,0,0,0.01,0.03,0.03,0,0.09,0,0.38,0.69,0,0,0,0.33,0.18,0.08,0.06,0,0,0.04,0,0,0,0.08,0.06,\correct,\correct,24,33
35,0,0,0,0,0,0,0,0,0,nat-pred,0,0,0,0,0.03,0,0,0,0,0,0,0.03,0,0,0,0,0,0,0,0,0,0,0,0,0,0,0,0.02,0,0,0,0,0,0,0,0,0,0,0,0,0,0,0,0,0.03,0,0,0,0,0.04,0,0,0,0,0,0,0,0,0,0,0,\correct,\correct,9,11
36,0.09,\incorrect,0.15,\incorrect,0.29,\incorrect,\incorrect,0.14,1.33,tree-binsert,0.17,0,\incorrect,0.26,\incorrect,0.50,0,0.23,\incorrect,\incorrect,\incorrect,2.17,0.14,0.09,\incorrect,\incorrect,\incorrect,0,\incorrect,0,\incorrect,0.18,\incorrect,0,0,0.65,0,\incorrect,0,0.46,\incorrect,0.49,0,\incorrect,0,\incorrect,0,\incorrect,0.14,0.08,0.28,0,0.65,0,2.88,\incorrect,\incorrect,0,\incorrect,2.16,0.99,0.36,\incorrect,\incorrect,0,0.41,\incorrect,\incorrect,\incorrect,0.24,0.64,\correct,\na,90,\na
37,\incorrect,0,0.01,0,0.07,0.01,0,0.06,0.06,tree-collect-leaves,0.04,0,0,0.01,0.08,0.07,0,0,0,0.01,0,\incorrect,\incorrect,0,0,0,0,\incorrect,0,0,0,0.01,0.01,\incorrect,0,0.01,0,0.12,0,0.78,0,0.01,0,0,\incorrect,0.02,\incorrect,0,\incorrect,\incorrect,0.77,0,0.07,0,0.13,0,0,0,0,1.17,\incorrect,\incorrect,0,0,\incorrect,0.12,0,0,0,\incorrect,\incorrect,\correct,\correct,29,28
38,0.02,0,0.01,0,0.02,0.07,0,0.01,0.11,tree-count-leaves,0,0,0.04,0.05,0.80,0.05,0,0.02,0,0.01,0,0.24,0.11,0.01,0,0,0.02,0,0,0,0,0.05,0.02,0,0,0.09,0,0.33,0,0.12,0,0.19,0,0,0,0.07,0,0,0,0.01,0.05,0,0.03,0,0.54,0.05,0,0,0,0.33,0.03,0.02,0.01,0,0,0.05,0,0,0,0.02,0.02,\correct,\correct,24,29
39,0.02,0,0.01,0,0.03,0.18,0,0.02,0.03,tree-count-nodes,0.01,0,0.26,0.02,0.39,0.04,0,0.03,0,0.01,0,0.24,0.03,0.01,0,0,0.10,0,0,0,0,0.03,0.08,0,0,0.04,0,0.81,0,0.04,0,0.04,0,0,0,0.04,0,0,0.01,0.02,0.03,0,0.03,0,0.18,0.20,0,0,0,0.18,0.07,0.04,0.03,0,0,0.02,0,0,0,0.03,0.03,\correct,\correct,24,24
40,0.51,0,2.51,0,1.15,\incorrect,0,0.65,2.84,tree-inorder,2.77,0,\incorrect,0.73,0.08,1.16,0,0.97,0,\incorrect,0,8.89,0.52,0.72,0,0,\incorrect,0,0,0,0,5.45,\incorrect,0,0,0.99,0,\incorrect,0,5.94,\incorrect,5.47,0,0,0,\incorrect,0,\incorrect,2.68,0.77,5.94,0,0.67,0,10.09,\incorrect,0,0,\incorrect,10.74,3.46,1.86,\incorrect,0,0,2.93,0,\incorrect,\incorrect,1.88,3,\correct,\correct,29,28
41,\incorrect,0,0.04,0,0.01,0.01,0,0.01,0.46,tree-map,0.04,0,0.03,0.02,1.02,0.03,0,0.01,0,0.04,0,\incorrect,\incorrect,0,0,0,0.02,\incorrect,0,0,0,0.01,0.01,\incorrect,0,0.01,0,0.30,0,0.03,0,0.02,0,0,\incorrect,0.14,\incorrect,0,\incorrect,\incorrect,0.01,0,0.04,0,0.58,0.03,0,0,0,0.58,\incorrect,\incorrect,0.01,0,\incorrect,0.42,0,0,0,\incorrect,\incorrect,\correct,\correct,47,58
42,2.53,\incorrect,0.01,\incorrect,12.45,\incorrect,\incorrect,4.61,0.09,tree-nodes-at-level,0.01,0,\incorrect,8.19,\incorrect,15.86,0,12.54,\incorrect,\incorrect,\incorrect,54.45,3.79,4.66,\incorrect,\incorrect,\incorrect,0,\incorrect,0,\incorrect,4.37,\incorrect,0,0,20.14,0,\incorrect,0,7.20,\incorrect,10.23,0,\incorrect,0,\incorrect,0,\incorrect,0.01,4.08,4.57,0,16.18,0,38.79,\incorrect,\incorrect,0,\incorrect,39.69,34.20,16.28,\incorrect,\incorrect,0,0.04,\incorrect,\incorrect,\incorrect,8.51,0.07,\correct,\na,49,\na
43,0.44,\incorrect,0.13,\incorrect,1.96,0.40,\incorrect,0.47,0.77,tree-postorder,0.14,0,0.71,0.51,\incorrect,2.05,0,1.77,\incorrect,0.13,\incorrect,7.86,0.49,0.44,\incorrect,\incorrect,2.22,0,\incorrect,0,\incorrect,0.76,0.31,0,0,1.94,0,4.17,0,0.78,0,0.91,0,\incorrect,0,0.38,0,0,0.13,0.53,0.69,0,0.64,0,4.25,2.58,\incorrect,0,0,3.97,5.02,1.54,0.45,\incorrect,0,0.41,\incorrect,0,0,2.54,0.72,\correct,\na,34,\na
44,\incorrect,0,0.01,0,0.01,\incorrect,0,0,0.17,tree-preorder,0.01,0,\incorrect,0.02,0.13,0.01,0,0.01,0,\incorrect,0,\incorrect,\incorrect,0.01,0,0,\incorrect,\incorrect,0,0,0,0.01,\incorrect,\incorrect,0,0.02,0,\incorrect,0,0.02,\incorrect,0.03,0,0,\incorrect,\incorrect,\incorrect,\incorrect,\incorrect,\incorrect,0.01,0,0.01,0,0.34,\incorrect,0,0,\incorrect,0.29,\incorrect,\incorrect,\incorrect,0,\incorrect,0.08,0,\incorrect,\incorrect,\incorrect,\incorrect,\correct,\correct,29,28
\end{filecontents*}
\csvreader[head to column names]{examples.csv}{}% use head of csv as column names
    {\\\hline\Test & \ComputationTime & \Correct & \Size & \SmythComputationTime & \SmythCorrect & \SmythSize }% specify your coloumns here
    %\\\hline
    %\textbf{Summary} & \needsrev{2.17} & 88\% & \needsrev{0.35} & 0.90\%
    \end{tabular}
    \caption{The results of running \Burst and \Smyth on the IO benchmark suite. A cross mark under the Time column indicates failure (i.e., either timeout or terminating without finding a solution). Under the ``Correct?'' column  ``\correct'' indicates that the synthesized program is the desired one, and ``\incorrect" indicates that the synthesized program matches the IO examples but not the user intent.  The column labeled ``Size'' shows the size of the synthesized program. }
    \label{fig:io-table}
\end{figure}

Figure~\ref{fig:io-table} presents the results of our comparison against \Smyth on synthesis tasks from {\bf IO} specifications. Here, the column labeled ``Time'' shows the synthesis time in seconds, and a cross mark (\incorrect) indicates failure (e.g., time-out). The column labeled ``Correct?'' shows whether the synthesized program is the one intended by the user. The column labeled ``Size'' shows the size of the synthesized program. In particular, when synthesis is successful, the returned program always satisfies the provided IO examples, however, it may or may not be the program intended by the user. Thus, this additional column allows us to evaluate how generalizable the synthesis results of these tools are.

As we can see from Figure~\ref{fig:io-table}, \toolname is able to synthesize a program consistent with the IO examples in all but 2 cases, whereas \Smyth fails on 4 benchmarks. 
For the benchmarks that can be solved by both tools, the  running time of both tools is quite fast (a few seconds or less) with the exception of a few outliers. 

Finally, the programs synthesized by \toolname and \Smyth are roughly equal in terms of generalization power: for \Burst{} there are 3 cases where the synthesized program is not the intended one, for \Smyth{} there are 4 such cases.
%\footnote{\Smyth included some sketching in their specifications of the higher-order benchmarks, which we removed in this comparison.}

%The programs generated by \Burst tools do not generalize quite as well as \Smyth -- they both have a different language so the ``minimal program'' means different things to the two tools. However, \Smyth's benchmarks were expertly chosen for \Smyth specifically -- the following section shows that they perform comparably by automatically generated benchmarks. Additionally, the higher-order benchmarks included some \emph{sketching}. When no sketches were given, \Smyth returned the incorrect result on all of the higher-order benchmarks.

%As these results indicate, \toolname is competitive with \Smyth for synthesizing programs from IO specifications. Even though \toolname is slightly slower for benchmarks that can be solved by both tools, \toolname can both solve more benchmarks (\needsrev{X\%} for \toolname vs \needsrev{Y\%} for \Smyth) and, unlike \Smyth, is not specialized to IO specifications.

\paragraph*{Failure Analysis} Next, we analyze the two benchmarks that \toolname fails on and provide some intuition about why it is unable to solve them. For the benchmarks called list-take and list-sorted-insert, \toolname times out because the specification does not in any way constrain the outputs of the many possible recursive calls. This both makes FTA creation quite slow and also causes the algorithm to explore many different strengthenings of the specification.

%\Burst returns the incorrect result in list-rev-tailcall, because it currently cannot synthesize the correct result. When making a recursive call, we only make recursive calls on arguments that are strictly smaller than the input value, according to our default metric. As an example, our default metric does not recognize the input \lstinline{([1;2],[])} to be strictly larger than \lstinline{([2],[1])}, which is necessary for a tail call implementation of list reverse.

\begin{comment}
We find that \Burst has issues when the function is highly underspecified. Our algorithm has two primary components (1) FTA construction, and (2) the backtracking search. When general postconditions are given, both components are hurt. FTA construction is made slower because the fanout on recursive calls is made larger -- intersecting two highly nondeterministic FTAs is very slow. With more choices to make, the backtracking search is also made slower, as there are more possible routes to take. This is why \Burst performs badly on list-sorted-insert and list-drop. The specifications for both of these are highly unconstrained, and permit a number of possible recursive calls, resulting in an explosion in a state space explosion.

\end{comment} 

\vspace{5pt}
\noindent
\fbox{\parbox{.95\linewidth}{
	{\bf Result \#1}: When synthesizing programs from IO specifications, \toolname is competitive with \Smyth, a state-of-the-art  tool for synthesizing recursive programs from input-output examples. In particular, \toolname can solve two more benchmarks despite not specializing in IO specifications.
  }
}

\subsection{Synthesis from Reference Implementation}
\label{subsec:synth-re}

\begin{figure}[!h]
    \centering
    \footnotesize
    \begin{tabular}{l|ccc|ccc}%
    {\bfseries Test} & \multicolumn{3}{c}{\bfseries \Burst} & \multicolumn{3}{c}{\bfseries \Smyth}\\
    & Time (s) & \# Iters & Size & Time (s) & \# Iters & Size
    \begin{filecontents*}{equiv.csv}
,SimpleMinEltMax,SmythIsectMax,NonincrementalInitialCreationMax,SmythFullSynth,IsectMax,RandomMinEltTotal,NonincrementalComputationTime,SmythInitialCreationTotal,MinifyMax,NonincrementalAcceptsTermMax,SimpleInitialCreationTotal,Test,InitialCreationMax,SimpleAcceptsTermTotal,NonincrementalLoopCount,RandomMinifyTotal,SimpleAcceptsTermMax,NonincrementalMinifyTotal,SmythComputationTime,RandomFullSynthMax,FullSynth,NonincrementalIsectMax,SmythIsectTotal,RandomInitialCreationMax,SimpleInitialCreationMax,SimpleComputationTime,SimpleMinEltTotal,SmythFullSynthMax,SmythMinEltTotal,SmythMinEltMax,NonincrementalMinifyMax,SmythLoopCount,NonincrementalMinEltMax,RandomMinEltMax,SimpleLoopCount,NonincrementalIsectTotal,LoopCount,SimpleIsectTotal,RandomComputationTime,AcceptsTermTotal,MinEltTotal,RandomAcceptsTermMax,NonincrementalMinEltTotal,FullSynthMax,SmythInitialCreationMax,SimpleFullSynth,RandomInitialCreationTotal,AcceptsTermMax,IsectTotal,SmythMinifyMax,SimpleMinifyMax,MinEltMax,NonincrementalFullSynth,RandomFullSynth,NonincrementalFullSynthMax,MinifyTotal,RandomIsectTotal,SmythAcceptsTermMax,NonincrementalAcceptsTermTotal,RandomLoopCount,ComputationTime,NonincrementalInitialCreationTotal,SimpleMinifyTotal,RandomMinifyMax,SimpleFullSynthMax,RandomIsectMax,InitialCreationTotal,RandomAcceptsTermTotal,SmythMinifyTotal,SimpleIsectMax,SmythAcceptsTermTotal,Size,SmythSize
0,0,0,0,0,0,0,0.03,0,0,0,0,bool-always-false,0,0,0,0,0,0,0.02,0,0,0,0,0,0,0.03,0,0,0,0,0,0,0,0,0,0,0,0,0.03,0,0,0,0,0,0,0,0,0,0,0,0,0,0,0,0,0,0,0,0,0,0.02,0,0,0,0,0,0,0,0,0,0,4,4
1,0,0,0,0,0,0,0.02,0,0,0,0,bool-always-true,0,0,1,0,0,0,0.03,0,0,0,0,0,0,0.02,0,0,0,0,0,1,0,0,1,0,1,0,0.03,0,0,0,0,0,0,0,0,0,0,0,0,0,0,0,0,0,0,0,0,1,0.02,0,0,0,0,0,0,0,0,0,0,4,4
2,0,0,0,0.01,0,0,0.04,0,0,0,0,bool-band,0,0,3,0,0,0,0.04,0,0,0,0,0,0,0.04,0,0,0,0,0,3,0,0,3,0,3,0,0.03,0,0,0,0,0,0,0.01,0,0,0,0,0,0,0,0,0,0,0,0,0,3,13.39,0,0,0,0,0,0,0,0,0,0,14,14
3,0,0,0,0.01,0,0,0.05,0,0,0,0,bool-bor,0,0,4,0,0,0,0.04,0,0,0,0,0,0,0.05,0,0,0,0,0,4,0,0,4,0,4,0,0.05,0,0,0,0,0,0,0.01,0,0,0,0,0,0,0.01,0,0,0,0,0,0,3,0.04,0,0,0,0,0,0,0,0,0,0,14,14
4,0,0,0,0.01,0,0,0.03,0,0,0,0,bool-impl,0,0,3,0,0,0,0.05,0,0,0,0,0,0,0.02,0,0,0,0,0,3,0,0,3,0,3,0,0.04,0,0,0,0,0,0,0,0,0,0,0,0,0,0,0,0,0,0,0,0,3,0.03,0,0,0,0,0,0,0,0,0,0,13,14
5,0,0,0,0,0,0,0.03,0,0,0,0,bool-neg,0,0,2,0,0,0,0.02,0,0,0,0,0,0,0.04,0,0,0,0,0,2,0,0,2,0,2,0,0.03,0,0,0,0,0,0,0,0,0,0,0,0,0,0,0,0,0,0,0,0,2,0.03,0,0,0,0,0,0,0,0,0,0,10,12
6,0,0,0,0.01,0,0,0.03,0,0,0,0,bool-xor,0,0,2,0,0,0,0.05,0,0,0,0,0,0,0.03,0,0,0,0,0,3,0,0,2,0,2,0,0.02,0,0,0,0,0,0,0,0,0,0,0,0,0,0,0,0,0,0,0,0,3,0.02,0,0,0,0,0,0,0,0,0,0,22,22
7,0,0,0,0.04,0,0.37,0.63,0,0,0,0.02,list-append,0,0,7,2.89,0,0.01,0.58,1.26,0.03,0.01,0,0.02,0,0.54,0,0.01,0,0,0,8,0.01,0.07,6,0.02,6,0.02,5.83,0,0,0,0.02,0.01,0,0.05,0.19,0,0.01,0,0,0,0.09,5.25,0.04,0.01,1.13,0,0,21,0.54,0.03,0.01,0.08,0.02,0.20,0.01,0,0,0,0,31,24
8,\incorrect,\incorrect,0.01,\incorrect,0.05,\incorrect,2.24,\incorrect,0.02,0,\incorrect,list-compress,0.01,\incorrect,9,\incorrect,\incorrect,0.12,\incorrect,\incorrect,0.44,0.06,\incorrect,\incorrect,\incorrect,\incorrect,\incorrect,\incorrect,\incorrect,\incorrect,0.02,\na,0.08,\incorrect,\incorrect,0.36,9,\incorrect,\incorrect,0,0.10,\incorrect,0.26,0.19,\incorrect,\incorrect,\incorrect,0,0.12,\incorrect,\incorrect,0.06,1.06,\incorrect,0.76,0.12,\incorrect,\incorrect,0,\incorrect,1.55,0.22,\incorrect,\incorrect,\incorrect,\incorrect,0.07,\incorrect,\incorrect,\incorrect,\incorrect,63,\na
9,0,0,0.01,0.03,0,\incorrect,1.51,0,0,0,0.01,list-concat,0,0,5,\incorrect,0,0.01,1.60,\incorrect,0.02,0.02,0,\incorrect,0,1.34,0,0.01,0,0,0.01,4,0.03,\incorrect,3,0.03,4,0,\incorrect,0,0,\incorrect,0.04,0.01,0,0.01,\incorrect,0,0,0,0,0,0.11,\incorrect,0.08,0,\incorrect,0,0,\incorrect,1.38,0.02,0,\incorrect,0.01,\incorrect,0.01,\incorrect,0,0,0,21,21
10,0.01,\incorrect,0,\incorrect,0.05,2.26,1.33,\incorrect,0.04,0,0.05,list-drop,0,0,6,2.34,0,0.19,\incorrect,14.99,0.23,0.13,\incorrect,0.01,0,0.71,0.06,\incorrect,\incorrect,\incorrect,0.08,\na,0.12,0.21,6,0.28,5,0.06,16.30,0,0.05,0,0.23,0.16,\incorrect,0.22,3.68,0,0.06,\incorrect,0.01,0.05,0.78,15.75,0.39,0.06,3.11,\incorrect,0,9,0.69,0.07,0.05,0.15,0.10,0.20,0.04,0,\incorrect,0.01,\incorrect,24,27
11,0,0,0,0.07,0,0.22,0.18,0,0,0,0.01,list-even-parity,0,0,5,3.53,0,0.01,0.26,0.96,0.04,0,0,0.01,0,0.19,0,0.02,0,0,0,7,0,0.03,5,0.01,5,0.01,5.45,0,0.01,0,0.01,0.03,0,0.03,0.07,0,0.01,0,0,0,0.05,5.30,0.04,0.01,0.75,0,0,23,0.19,0.02,0,0.05,0.01,0.12,0.02,0,0,0,0,32,37
12,\incorrect,\incorrect,0,\incorrect,0.02,\incorrect,1.29,\incorrect,0.01,0,\incorrect,list-filter,0,\incorrect,8,\incorrect,\incorrect,0.04,\incorrect,\incorrect,0.19,0.02,\incorrect,\incorrect,\incorrect,\incorrect,\incorrect,\incorrect,\incorrect,\incorrect,0.01,\na,0.01,\incorrect,\incorrect,0.10,8,\incorrect,\incorrect,0,0.03,\incorrect,0.06,0.07,\incorrect,\incorrect,\incorrect,0,0.04,\incorrect,\incorrect,0.01,0.27,\incorrect,0.11,0.07,\incorrect,\incorrect,0,\incorrect,1.20,0.05,\incorrect,\incorrect,\incorrect,\incorrect,0.03,\incorrect,\incorrect,\incorrect,\incorrect,56,\na
13,0.84,\incorrect,0,\incorrect,0.06,\incorrect,33.51,\incorrect,0.05,0,0.06,list-fold,0.01,0,6,\incorrect,0,0.67,\incorrect,\incorrect,0.31,0.67,\incorrect,\incorrect,0.01,7.91,0.92,\incorrect,\incorrect,\incorrect,0.63,\na,9.58,\incorrect,6,0.72,6,3.71,\incorrect,0,0.07,\incorrect,9.87,0.18,\incorrect,7.03,\incorrect,0,0.09,\incorrect,0.74,0.04,11.33,\incorrect,11.22,0.10,\incorrect,\incorrect,0,\incorrect,1.21,0.06,2.34,\incorrect,6.27,\incorrect,0.04,\incorrect,\incorrect,1.42,\incorrect,37,\na
14,0,0,0,0,0,0,0.30,0,0,0,0,list-hd,0,0,2,0,0,0,0.31,0,0.01,0,0,0,0,0.25,0,0,0,0,0,2,0,0,2,0,2,0,0.25,0,0,0,0,0,0,0,0,0,0,0,0,0,0,0,0,0,0,0,0,3,0.29,0,0,0,0,0,0,0,0,0,0,12,13
15,0,0,0,0.02,0,0.01,0.56,0,0,0,0,list-inc,0,0,3,0.02,0,0,0.70,0.18,0.01,0,0,0.01,0,0.54,0,0.01,0,0,0,2,0,0,3,0,3,0,0.77,0,0,0,0,0,0,0.01,0.18,0,0,0,0,0,0.01,0.23,0,0,0.02,0,0,6,0.55,0,0,0,0,0,0,0,0,0,0,20,10
16,0,0,0,0.02,0,1.25,0.47,0,0,0,0.02,list-last,0,0,3,0.51,0,0,0.57,1.20,0.01,0,0,0.02,0,0.49,0,0.01,0,0,0,5,0,1.02,5,0,3,0.03,2.94,0,0,0,0,0.01,0,0.07,0.30,0,0,0,0,0,0.01,2.47,0.01,0,0.27,0,0,9,0.46,0.01,0.01,0.06,0.04,0.08,0.01,0,0,0.01,0,26,27
17,0,0,0,0.01,0,0,0.42,0,0,0,0.02,list-length,0,0,4,0,0,0,0.47,0.01,0.01,0,0,0,0,0.47,0,0,0,0,0,3,0,0,6,0,4,0.01,0.47,0,0,0,0,0.01,0,0.04,0.01,0,0,0,0,0,0.01,0.01,0.01,0,0,0,0,4,0.41,0.01,0.01,0,0.02,0,0.01,0,0,0,0,15,16
18,0.19,\incorrect,0,\incorrect,0,0.01,0.85,\incorrect,0,0,0.10,list-map,0,0,4,0.03,0,0.02,\incorrect,0.03,0.12,0.01,\incorrect,0,0.01,5.45,0.28,\incorrect,\incorrect,\incorrect,0,\na,0,0,8,0.02,4,2.75,0.81,0,0.01,0,0.02,0.04,\incorrect,4.69,0.02,0,0.02,\incorrect,0.43,0,0.12,0.09,0.04,0.02,0.03,\incorrect,0,6,0.82,0.06,1.56,0,3.56,0.01,0.06,0,\incorrect,0.82,\incorrect,35,\na
19,0,0,0,0.05,0.01,0.36,0.53,0,0,0,0.02,list-nth,0,0,7,0.98,0,0.01,0.59,0.69,0.08,0.01,0,0.01,0,0.50,0,0.02,0,0,0,5,0.01,0.09,6,0.02,7,0.02,2.86,0,0.01,0,0.02,0.04,0,0.05,0.15,0,0.01,0,0,0.01,0.11,2.40,0.06,0.02,0.62,0,0,18,0.51,0.04,0.01,0.06,0.03,0.17,0.02,0,0,0,0,35,29
20,\incorrect,\incorrect,0.01,\incorrect,0.01,0.01,0.53,\incorrect,0,0,\incorrect,list-pairwise-swap,0,\incorrect,6,0.04,\incorrect,0.01,\incorrect,0.27,0.08,0,\incorrect,0.01,\incorrect,\incorrect,\incorrect,\incorrect,\incorrect,\incorrect,0,\na,0.01,0,\incorrect,0.01,6,\incorrect,0.82,0,0.02,0,0.02,0.03,\incorrect,\incorrect,0.26,0,0.01,\incorrect,\incorrect,0.01,0.09,0.37,0.04,0.01,0.02,\incorrect,0,8,0.53,0.03,\incorrect,0,\incorrect,0.01,0.03,0,\incorrect,\incorrect,\incorrect,42,\na
21,\incorrect,\incorrect,0.63,\incorrect,0.10,\incorrect,14.45,\incorrect,0.08,0,\incorrect,list-rev-append,0.65,\incorrect,5,\incorrect,\incorrect,0.30,\incorrect,\incorrect,13.54,0.12,\incorrect,\incorrect,\incorrect,\incorrect,\incorrect,\incorrect,\incorrect,\incorrect,0.07,\na,2.82,\incorrect,\incorrect,0.49,5,\incorrect,\incorrect,0,11.36,\incorrect,11.22,11.39,\incorrect,\incorrect,\incorrect,0,0.48,\incorrect,\incorrect,2.85,13.37,\incorrect,11.26,0.35,\incorrect,\incorrect,0,\incorrect,14.55,0.95,\incorrect,\incorrect,\incorrect,\incorrect,0.99,\incorrect,\incorrect,\incorrect,\incorrect,24,\na
22,0,\incorrect,0.02,\incorrect,0,0.01,0.93,\incorrect,0,0,0.03,list-rev-fold,0.02,0,2,0,0,0,\incorrect,0.04,0.03,0,\incorrect,0.02,0.02,0.91,0,\incorrect,\incorrect,\incorrect,0,\na,0,0.01,2,0,2,0,21.28,0,0,0,0,0.03,\incorrect,0.03,0.02,0,0,\incorrect,0,0,0.04,0.04,0.03,0,0,\incorrect,0,2,1,0.03,0,0,0.03,0,0.02,0,\incorrect,0,\incorrect,10,\na
23,0.06,\incorrect,0.02,\incorrect,0.04,\incorrect,7.66,\incorrect,0.03,0,0.05,list-rev-snoc,0.02,0,4,\incorrect,0,0.04,\incorrect,\incorrect,5.92,0.04,\incorrect,\incorrect,0.02,1.86,0.07,\incorrect,\incorrect,\incorrect,0.03,\na,6.29,\incorrect,4,0.05,4,0.43,\incorrect,0,5.79,\incorrect,6.30,5.88,\incorrect,0.77,\incorrect,0,0.05,\incorrect,0.10,5.78,6.44,\incorrect,6.41,0.03,\incorrect,\incorrect,0,\incorrect,6.94,0.06,0.23,\incorrect,0.73,\incorrect,0.05,\incorrect,\incorrect,0.20,\incorrect,20,\na
24,\incorrect,0,\incorrect,0.05,\incorrect,\incorrect,\incorrect,0,\incorrect,\incorrect,\incorrect,list-rev-tailcall,\incorrect,\incorrect,\incorrect,\incorrect,\incorrect,\incorrect,0.68,\incorrect,\incorrect,\incorrect,0,\incorrect,\incorrect,\incorrect,\incorrect,0.01,0,0,\incorrect,9,\incorrect,\incorrect,\incorrect,\incorrect,\na,\incorrect,\incorrect,\incorrect,\incorrect,\incorrect,\incorrect,\incorrect,0,\incorrect,\incorrect,\incorrect,\incorrect,0,\incorrect,\incorrect,\incorrect,\incorrect,\incorrect,\incorrect,\incorrect,0,\incorrect,\incorrect,\incorrect,\incorrect,\incorrect,\incorrect,\incorrect,\incorrect,\incorrect,\incorrect,0,\incorrect,0,\na,24
25,\incorrect,0,0.01,0.04,0,0.99,0.63,0,0,0,\incorrect,list-snoc,0.01,\incorrect,3,1,\incorrect,0.01,0.68,1.85,0.08,0,0,0.05,\incorrect,\incorrect,\incorrect,0.01,0,0,0,7,0,0.36,\incorrect,0,3,\incorrect,3.93,0,0.01,0,0.01,0.07,0,\incorrect,0.19,0,0,0,\incorrect,0,0.13,3.23,0.09,0.01,0.89,0,0,13,0.62,0.11,\incorrect,0.29,\incorrect,0.48,0.07,0,0,\incorrect,0,36,27
26,0.07,0,0.01,0.05,0.11,\incorrect,2.16,0,0.07,0,0.07,list-sort-sorted-insert,0.01,0,5,\incorrect,0,0.10,1.98,\incorrect,0.46,0.12,0,\incorrect,0.01,2.48,0.08,0.01,0,0,0.07,6,0.13,\incorrect,5,0.16,5,0.38,\incorrect,0,0.15,\incorrect,0.18,0.36,0,0.75,\incorrect,0,0.14,0,0.07,0.13,0.51,\incorrect,0.41,0.13,\incorrect,0,0,\incorrect,2.12,0.06,0.22,\incorrect,0.60,\incorrect,0.04,\incorrect,0,0.14,0,20,20
27,\incorrect,0,0.01,13,0.09,\incorrect,1.48,0,0.04,0,\incorrect,list-sorted-insert,0.01,\incorrect,6,\incorrect,\incorrect,0.07,14.22,\incorrect,0.56,0.08,0,\incorrect,\incorrect,\incorrect,\incorrect,2.41,0,0,0.03,12,0.07,\incorrect,\incorrect,0.18,6,\incorrect,\incorrect,0,0.12,\incorrect,0.12,0.26,0,\incorrect,\incorrect,0,0.19,0,\incorrect,0.08,0.54,\incorrect,0.27,0.12,\incorrect,0,0,\incorrect,1.57,0.16,\incorrect,\incorrect,\incorrect,\incorrect,0.12,\incorrect,0,\incorrect,0,63,72
28,0,0,0.12,0.01,0,0,1.37,0,0,0,0.90,list-stutter,0.12,0,3,0,0,0,0.59,0.87,0.89,0,0,0.12,0.12,1.42,0,0,0,0,0,4,0,0,3,0,3,0,1.40,0,0,0,0,0.86,0,0.90,0.91,0,0,0,0,0,0.90,0.91,0.87,0,0,0,0,3,1.42,0.90,0,0,0.87,0,0.88,0,0,0,0,25,25
29,0,0,0.01,0.04,0,0,1.03,0,0,0,0.01,list-sum,0.01,0,2,0.01,0,0,1.02,0.02,0.01,0,0,0.01,0,1,0,0.02,0,0,0,2,0,0,2,0,2,0,1.07,0,0,0,0,0.01,0,0.01,0.01,0,0,0,0,0,0.02,0.03,0.01,0,0.01,0,0,3,1.03,0.01,0,0,0.01,0.01,0.01,0,0,0,0,10,10
30,0.04,0,0.01,0.07,0.07,\incorrect,9.81,0,0.05,0,0.09,list-take,0.01,0,9,\incorrect,0,1.28,0.63,\incorrect,16.66,0.09,0,\incorrect,0,1.40,0.11,0.02,0,0,0.07,7,0.20,\incorrect,10,1.73,9,0.49,\incorrect,0,5.40,\incorrect,3.28,12.67,0,0.96,\incorrect,0,3.51,0,0.03,0.14,9.33,\incorrect,8.09,2.47,\incorrect,0,0,\incorrect,17.12,1.74,0.25,\incorrect,0.48,\incorrect,3.27,\incorrect,0,0.07,0,40,33
31,0,0,0,0,0,0,0.28,0,0,0,0.01,list-tl,0,0,2,0,0,0,0.27,0,0.01,0,0,0,0,0.31,0,0,0,0,0,2,0,0,2,0,2,0,0.30,0,0,0,0,0,0,0.01,0.01,0,0,0,0,0,0.01,0.01,0,0,0,0,0,3,0.31,0,0,0,0,0,0,0,0,0,0,11,13
32,0,0,0,0.02,0,0,0.14,0,0,0,0.02,nat-add,0,0,6,0.01,0,0,0.13,0.01,0.03,0,0,0,0,0.13,0,0,0,0,0,6,0,0,6,0.01,6,0.01,0.14,0,0,0,0.01,0.01,0,0.04,0.01,0,0.01,0,0,0,0.04,0.03,0.01,0.01,0,0,0,6,0.12,0.02,0.01,0,0.02,0,0.01,0,0,0,0,22,18
33,0,0,0,0.01,0,0,0.05,0,0,0,0,nat-iseven,0,0,4,0,0,0,0.04,0,0.01,0,0,0,0,0.05,0,0,0,0,0,4,0,0,4,0,4,0,0.04,0,0,0,0,0,0,0.01,0,0,0,0,0,0,0.01,0,0,0,0,0,0,3,0.03,0,0,0,0,0,0,0,0,0,0,20,21
34,0.88,0,0.05,0.10,0.51,\incorrect,7.32,0,0.36,0,0.11,nat-max,0.01,0,5,\incorrect,0,1.01,0.24,\incorrect,2.21,0.50,0,\incorrect,0,24.60,3.23,0.04,0,0,0.35,7,1.08,\incorrect,8,1.64,5,13.67,\incorrect,0,1.19,\incorrect,3.12,2.11,0,24.48,\incorrect,0,0.55,0,1.18,1.18,7.23,\incorrect,7.11,0.44,\incorrect,0,0,\incorrect,2.28,0.24,7.47,\incorrect,10.91,\incorrect,0.03,\incorrect,0,2.27,0,24,42
35,0,0,0,0,0,0,0.05,0,0,0,0,nat-pred,0,0,2,0,0,0,0.03,0,0,0,0,0,0,0.03,0,0,0,0,0,2,0,0,2,0,2,0,0.05,0,0,0,0,0,0,0,0,0,0,0,0,0,0,0.01,0,0,0,0,0,2,0.03,0,0,0,0,0,0,0,0,0,0,9,11
36,\incorrect,\incorrect,0.24,\incorrect,0.78,\incorrect,18.95,\incorrect,0.42,0,\incorrect,tree-binsert,0.26,\incorrect,8,\incorrect,\incorrect,0.48,\incorrect,\incorrect,11.90,0.44,\incorrect,\incorrect,\incorrect,\incorrect,\incorrect,\incorrect,\incorrect,\incorrect,0.19,\na,2.64,\incorrect,\incorrect,1.06,8,\incorrect,\incorrect,0,7.72,\incorrect,5.42,6.96,\incorrect,\incorrect,\incorrect,0,1.33,\incorrect,\incorrect,4.26,7.95,\incorrect,3.91,2.21,\incorrect,\incorrect,0,\incorrect,22.10,0.92,\incorrect,\incorrect,\incorrect,\incorrect,0.51,\incorrect,\incorrect,\incorrect,\incorrect,90,\na
37,0.07,0,0.11,0.07,0.11,\incorrect,9.13,0,0.10,0,0.07,tree-collect-leaves,0.11,0,5,\incorrect,0,0.11,8.29,\incorrect,0.94,0.11,0,\incorrect,0.02,8.81,0.08,0.04,0,0,0.10,5,0.49,\incorrect,4,0.12,5,0.33,\incorrect,0,0.51,\incorrect,0.51,0.91,0,0.72,\incorrect,0,0.12,0,0.11,0.49,0.93,\incorrect,0.90,0.12,\incorrect,0,0,\incorrect,9.09,0.18,0.24,\incorrect,0.70,\incorrect,0.18,\incorrect,0,0.14,0,29,28
38,0.08,0,0,0.38,0.12,24.43,9.03,0,0.10,0,0.02,tree-count-leaves,0,0,4,9.97,0,0.10,8.63,49.13,0.57,0.12,0,0,0,8.76,0.08,0.36,0,0,0.10,4,0.29,23.98,4,0.13,4,0.33,58.95,0,0.32,0,0.32,0.54,0,0.64,0.02,0,0.13,0,0.10,0.30,0.58,49.84,0.54,0.11,15.38,0,0,5,23.57,0.02,0.21,9.48,0.61,15.24,0.01,0,0,0.17,0,24,29
39,\incorrect,0,0,0.07,0.01,\incorrect,8.18,0,0.01,0,\incorrect,tree-count-nodes,0,\incorrect,4,\incorrect,\incorrect,0.01,8.25,\incorrect,0.07,0.01,0,\incorrect,\incorrect,\incorrect,\incorrect,0.05,0,0,0.01,4,0.02,\incorrect,\incorrect,0.02,4,\incorrect,\incorrect,0,0.02,\incorrect,0.02,0.05,0,\incorrect,\incorrect,0,0.01,0,\incorrect,0.02,0.07,\incorrect,0.05,0.01,\incorrect,0,0,\incorrect,8.48,0.02,\incorrect,\incorrect,\incorrect,\incorrect,0.01,\incorrect,0,\incorrect,0,24,24
40,0.04,0,0.07,0.11,0.17,\incorrect,12.79,0,0.02,0,0.15,tree-inorder,0.07,0,5,\incorrect,0,0.04,11.72,\incorrect,0.46,0.16,0,\incorrect,0.07,12.81,0.04,0.04,0,0,0.03,6,0.04,\incorrect,5,0.18,5,0.22,\incorrect,0,0.08,\incorrect,0.07,0.44,0,0.49,\incorrect,0,0.20,0,0.04,0.05,0.44,\incorrect,0.42,0.04,\incorrect,0,0,\incorrect,13.98,0.14,0.08,\incorrect,0.46,\incorrect,0.14,\incorrect,0,0.08,0,29,28
41,0,\incorrect,0.01,\incorrect,0.01,0.19,8.50,\incorrect,0,0,0.04,tree-map,0.01,0,5,2.81,0,0.01,\incorrect,2.22,0.08,0,\incorrect,0.01,0.01,8.79,0.01,\incorrect,\incorrect,\incorrect,0,\na,0,0.06,5,0.01,5,0.03,13.04,0,0.01,0,0.01,0.03,\incorrect,0.10,0.05,0,0.02,\incorrect,0,0.01,0.10,4.34,0.04,0.02,1.11,\incorrect,0,11,8.98,0.06,0.01,0.19,0.05,0.54,0.02,0,\incorrect,0.01,\incorrect,47,\na
42,\incorrect,0,\incorrect,20.72,\incorrect,\incorrect,\incorrect,0,\incorrect,\incorrect,\incorrect,tree-nodes-at-level,\incorrect,\incorrect,\incorrect,\incorrect,\incorrect,\incorrect,26.75,\incorrect,\incorrect,\incorrect,0,\incorrect,\incorrect,\incorrect,\incorrect,13.27,0,0,\incorrect,5,\incorrect,\incorrect,\incorrect,\incorrect,\na,\incorrect,\incorrect,\incorrect,\incorrect,\incorrect,\incorrect,\incorrect,0,\incorrect,\incorrect,\incorrect,\incorrect,0,\incorrect,\incorrect,\incorrect,\incorrect,\incorrect,\incorrect,\incorrect,0,\incorrect,\incorrect,\incorrect,\incorrect,\incorrect,\incorrect,\incorrect,\incorrect,\incorrect,\incorrect,0,\incorrect,0,\na,38
43,0.55,\incorrect,0.12,\incorrect,3.94,3.47,15.46,\incorrect,0.94,0,0.33,tree-postorder,0.12,0,8,7.81,0,0.23,\incorrect,27.50,8.18,0.46,\incorrect,0.15,0.12,25.29,0.64,\incorrect,\incorrect,\incorrect,0.06,\na,0.11,2.82,7,1.37,8,9.40,50.75,0,1.82,0,0.43,6.92,\incorrect,12.17,0.46,0,4.56,\incorrect,0.39,1.62,2.75,34.69,1.13,1.39,22.67,\incorrect,0,11,21.94,0.66,1.80,3,11.42,18.37,0.36,0,\incorrect,3.54,\incorrect,34,\na
44,1.98,0,3.17,0.07,0,26.91,18.11,0,0,0,3.35,tree-preorder,0,0,5,2.59,0,0.03,11.89,29.53,0.02,0.02,0,0.60,3.25,80.01,1.99,0.04,0,0,0.01,5,0.01,26.18,5,0.03,4,41.36,52.44,0,0,0,0.02,0.01,0,65.28,1.16,0,0,0,6.40,0,3.39,37.04,3.36,0,6.26,0,0,9,11.72,3.29,18.59,1,65.25,2.92,0.01,0,0,19.15,0,29,28
\end{filecontents*}
\csvreader[head to column names]{equiv.csv}{}% use head of csv as column names
    {\\\hline\Test & \ComputationTime & \LoopCount & \Size & \SmythComputationTime & \SmythLoopCount & \SmythSize }
   % \textbf{Summary} & \needsrev{4.49} & 4.37 & \needsrev{3.07} & 4.52
    \end{tabular}
    \caption{The results of running \Burst and \Smyth on the \textbf{Ref} benchmark suite. A cross mark under the Time column indicates failure (i.e., either timeout or terminating without finding a solution). The column labeled ``\# Iters" shows the number of iterations within the CEGIS loop. The column labeled ``Size'' shows the size of the synthesized program.}
    \label{fig:equiv-table}
\end{figure}

In this section, we use \toolname to synthesize programs from reference implementations and compare against \Smyth. We incorporate both tools into a CEGIS loop and, for each candidate program, check whether it is equivalent to our reference implementation. If not, we ask the verifier for a concrete counterexample $I$ and obtain its corresponding output $O$ by running the reference implementation on $I$. We then add $(I, O)$ as a new input-output example and continue this process until the synthesizer program is indeed equivalent to the reference implementation.\footnote{We actually use bounded testing instead of verification; however, we manually confirmed that the generated programs are indeed equivalent to the reference implementation.}

The results of this evaluation are shown in Figure~\ref{fig:equiv-table}. In particular, the column labeled ``Time" indicates synthesis time in seconds, with \incorrect{} indicating failure as before. The second column labeled ``\# Iters" shows the number of iterations of the CEGIS loop for those benchmarks that can be synthesized. The third column ``Size'' shows the size of the synthesized program. As we can see from this figure, \toolname successfully solves all but two of the benchmarks with an average synthesis time of 4.49 seconds and 4.37 CEGIS iterations on average. In contrast, \Smyth fails to solve 12 of these benchmarks, and it takes an average of 3.07 seconds and 4.52 CEGIS iterations. Overall, we believe these results indicate that \toolname{} is able to deal better with the random examples generated by the verifier compared to \Smyth.\footnote{Recall that the IO examples used in the previous experiment are written by the \Smyth developers.}

\paragraph*{Failure Analysis} {The reason that \toolname fails on  tree-nodes-at-level  is similar to that for pure IO examples: the verifier returns IO pairs for which the results of the many possible recursive calls are highly under-constrained, resulting in slow FTA construction as well as many strengthening steps.}
{\Burst fails on list-rev-tailcall due to our requirement for ensuring termination of synthesized programs (see Section~\ref{subsec:termination}). Concretely, list-rev-tailcall needs to make a recursive call on \lstinline{([2],[1])} for input \lstinline{([1;2],[])}, but our default ordering does not consider \lstinline{([2],[1])} to be strictly less than \lstinline{([1;2],[])}.}
%While this issue could potentially be mitigated by modifying the verifier to return better counterexamples, we have not done so and leave this direction to future work.

%\Burst fails on 3 cases: list-rev-tailcall, tree-nodes-at-level, and list-sorted-insert. The failures on list-rev-tailcall and list-sorted-insert result from the same issues as in Section~\ref{subsec:synth-io}.

%The tree-nodes-at-level benchmark fails as an \textbf{Reimpl} benchmark, but not as an \textbf{IO} benchmark. Our verifier gave input/output examples that permitted more possible recursive calls than were provided by the \Smyth benchmark suite.

\vspace{5pt}
\noindent
\fbox{\parbox{.95\linewidth}{
	{\bf Result \#2}: \Burst is able to synthesize 96\% of the programs from a reference implementation. In contrast, \Smyth is only able to synthesize 73\%. 
  }
}

%\paragraph*{Comparison to Prior Work} When running \Smyth on these benchmarks, we find that sometimes \Smyth erroneously states that there are no satisfying solutions to a given set of input-output examples. We are not sure of the source of this issue. Besides these issues, performance is comparable to that in the \textbf{IO} benchmark suite.

%Furthermore, because the minimality guarantees of \Smyth are comparable to that of \Burst, the number of iterations required for convergence is similar.

\subsection{Synthesis from Logical Specifications}
\label{subsec:synth-logic}

Figure~\ref{fig:logical-table} shows the results of our evaluation on logical specifications for each  of \toolname, \Synquid, and \Leon. As before, the column labeled ``Time'' shows the synthesis time for each tool, and the column labeled ``Correct?'' shows whether the tools were able to generate the intended program from the given specification. The column ``Size'' shows the size of the synthesized program.\footnote{\Leon did not seem to provide an automated way to identify size, so we did not include a ``Size'' column for it.} As we can see in this table, \toolname solves more benchmarks than \Leon and significantly more compared to \Synquid. %Furthermore, all three tools tend to fail on the same benchmarks.
In what follows, we explain the failure cases of each tool and contrast them with each other.

\begin{figure}[!h]
    \centering
    \footnotesize
    \begin{tabular}{c|ccc|cc|ccc}%
    \bfseries Test & \multicolumn{3}{c}{\bfseries \Burst} & \multicolumn{2}{c}{\bfseries \Leon} & \multicolumn{2}{c}{\bfseries \Synquid}\\
    & Time (s) & Correct? & Size & Time (s) & Correct? 
    %& Size 
    & Time (s) & Correct? & Size
    \begin{filecontents*}{postconditional.csv}
,SimpleMinEltMax,NonincrementalInitialCreationMax,NonincrementalMinifyMax,IsectMax,RandomMinEltTotal,AcceptsTermMax,MinifyMax,NonincrementalInitialCreationTotal,Test,InitialCreationMax,NonincrementalLoopCount,RandomMinifyTotal,NonincrementalMinifyTotal,SimpleMinifyMax,IsectTotal,FullSynth,NonincrementalIsectMax,RandomInitialCreationMax,SimpleComputationTime,SimpleMinEltTotal,InitialCreationTotal,RandomIsectMax,NonincrementalAcceptsTermTotal,NonincrementalAcceptsTermMax,MinifyTotal,NonincrementalMinEltMax,RandomMinEltMax,SimpleLoopCount,NonincrementalIsectTotal,LoopCount,RandomComputationTime,AcceptsTermTotal,MinEltTotal,RandomAcceptsTermMax,NonincrementalMinEltTotal,FullSynthMax,SimpleFullSynth,RandomInitialCreationTotal,SimpleAcceptsTermMax,RandomFullSynthMax,SimpleInitialCreationMax,MinEltMax,NonincrementalFullSynth,NonincrementalFullSynthMax,NonincrementalComputationTime,RandomIsectTotal,RandomLoopCount,ComputationTime,SimpleIsectTotal,SimpleMinifyTotal,RandomMinifyMax,SimpleFullSynthMax,SimpleAcceptsTermTotal,RandomAcceptsTermTotal,RandomFullSynth,SimpleIsectMax,SimpleInitialCreationTotal,SynquidTime,LeonTime,Correct,Size,LeonCorrect,LeonSize,SynquidCorrect,SynquidSize
0,0,0,0,0,0,0,0,0,bool-always-false,0,0,0,0,0,0,0,0,0,0.02,0,0,0,0,0,0,0,0,0,0,0,0.02,0,0,0,0,0,0,0,0,0,0,0,0,0,0.02,0,0,0.08,0,0,0,0,0,0,0,0,0,0.30,0.04,\correct,4,\correct,4,\correct,3
1,0,0,0,0,0,0,0,0,bool-always-true,0,0,0,0,0,0,0,0,0,0.03,0,0,0,0,0,0,0,0,0,0,0,0.02,0,0,0,0,0,0,0,0,0,0,0,0,0,0.03,0,0,0.03,0,0,0,0,0,0,0,0,0,0.31,0.07,\correct,4,\correct,4,\correct,3
2,0,0,0,0,0,0,0,0,bool-band,0,0,0,0,0,0,0,0,0,0.05,0,0,0,0,0,0,0,0,0,0,0,0.05,0,0,0,0,0,0,0,0,0,0,0,0,0,0.03,0,0,0.03,0,0,0,0,0,0,0,0,0,0.32,6.57,\correct,14,\correct,8,\correct,6
3,0,0,0,0,0,0,0,0.01,bool-bor,0,0,0,0,0,0,0,0,0,0.04,0,0,0,0,0,0,0,0,0,0,0,0.05,0,0,0,0,0,0,0,0,0,0,0,0,0,0.04,0,0,0.05,0,0,0,0,0,0,0,0,0.01,0.30,6.62,\correct,14,\correct,8,\correct,6
4,0,0,0,0,0,0,0,0,bool-impl,0,0,0,0,0,0,0,0,0,0.04,0,0,0,0,0,0,0,0,0,0,0,0.04,0,0,0,0,0,0,0,0,0,0,0,0,0,0.03,0,0,0.04,0,0,0,0,0,0,0,0,0,0.31,6.61,\correct,13,\correct,8,\correct,6
5,0,0,0,0,0,0,0,0,bool-neg,0,0,0,0,0,0,0,0,0,0.02,0,0,0,0,0,0,0,0,0,0,0,0.03,0,0,0,0,0,0,0,0,0,0,0,0,0,0.03,0,0,0.03,0,0,0,0,0,0,0,0,0,0.29,1.63,\correct,10,\correct,7,\correct,5
6,0,0,0,0,0,0,0,0,bool-xor,0,0,0,0,0,0,0,0,0,0.04,0,0,0,0,0,0,0,0,0,0,0,0.04,0,0,0,0,0,0,0,0,0,0,0,0,0,0.04,0,0,0.04,0,0,0,0,0,0,0,0,0,0.31,3.60,\correct,22,\correct,8,\correct,9
7,0,0,0,0,0.03,0,0,0.04,list-append,0,0,0.12,0.01,0,0.01,0,0,0.01,0.85,0.01,0.02,0.02,0,0,0.01,0,0.01,0,0.01,0,1.32,0,0,0,0.01,0,0,0.11,0,0,0,0,0,0,0.89,0.07,0,0.88,0.01,0.01,0.01,0,0,0,0,0,0.04,0.51,44.28,\correct,31,\correct,18,\correct,39
8,\incorrect,\incorrect,\incorrect,1.12,\incorrect,0,0.50,\incorrect,list-compress,0.08,\incorrect,\incorrect,\incorrect,\incorrect,1.68,0,\incorrect,\incorrect,\incorrect,\incorrect,0.49,\incorrect,\incorrect,\incorrect,1.02,\incorrect,\incorrect,\incorrect,\incorrect,0,\incorrect,0,2.18,\incorrect,\incorrect,0,\incorrect,\incorrect,\incorrect,\incorrect,\incorrect,1.60,\incorrect,\incorrect,\incorrect,\incorrect,\incorrect,9.61,\incorrect,\incorrect,\incorrect,\incorrect,\incorrect,\incorrect,\incorrect,\incorrect,\incorrect,\incorrect,\incorrect,\correct,63,\na,\na,\na,\na
9,0,0,0,0,0.02,0,0,0.03,list-concat,0.01,0,0.05,0,0,0,0,0,0.01,2.25,0,0.02,0.03,0,0,0,0,0.01,0,0,0,2.59,0,0,0,0,0,0,0.06,0,0,0.01,0,0,0,2.39,0.05,0,2.28,0,0,0.01,0,0,0,0,0,0.02,0.39,11.57,\correct,21,\incorrect,15,\correct,71
10,0.02,0.01,0.01,0.08,0.11,0,0.03,0.08,list-drop,0.01,0,0.33,0.02,0.03,0.15,0,0.02,0.01,1.20,0.03,0.08,0.54,0,0,0.07,0.03,0.02,0,0.04,0,2.14,0,0.14,0,0.05,0,0,0.14,0,0,0.01,0.06,0,0,0.92,0.81,0,1.13,0.27,0.11,0.12,0,0,0,0,0.09,0.11,0.58,51.04,\correct,33,\correct,18,\correct,35
11,0,0,0,0,0,0,0,0.04,list-even-parity,0,0,0,0.01,0,0.01,0,0,0,0.34,0.01,0.04,0,0,0,0.01,0,0,0,0.01,0,0.29,0,0.01,0,0.01,0,0,0.01,0,0,0,0,0,0,0.34,0,0,0.70,0.02,0.01,0,0,0,0,0,0,0.03,\incorrect,42.82,\correct,32,\correct,23,\na,\na
12,0.01,0,0,0.01,0.07,0,0,0.09,list-filter,0.01,0,0.82,0.02,0.01,0.03,0,0.01,0.01,1.41,0.01,0.04,0.12,0,0,0.04,0.01,0.02,0,0.05,0,3.08,0,0.01,0,0.03,0,0,0.08,0,0,0.01,0.01,0,0,1.50,0.36,0,1.51,0.06,0.03,0.04,0,0,0,0,0.01,0.09,\incorrect,\incorrect,\incorrect,46,\na,\na,\na,\na
13,0.12,0.02,0.02,0.07,\incorrect,0,0.03,1.83,list-fold,0.02,0,\incorrect,0.26,0.05,0.33,0,0.06,\incorrect,4.35,2.56,1.33,\incorrect,0,0,0.21,0.06,\incorrect,0,0.43,0,\incorrect,0,0.29,\incorrect,0.43,0,0,\incorrect,0,\incorrect,0.01,0.07,0,0,4.32,\incorrect,\incorrect,3.44,0.29,0.14,\incorrect,0,0,\incorrect,\incorrect,0.12,0.20,\incorrect,\incorrect,\correct,37,\na,\na,\na,\na
14,0,0,0,0,0,0,0,0,list-hd,0,0,0,0,0,0,0,0,0,0.32,0,0,0,0,0,0,0,0,0,0,0,0.30,0,0,0,0,0,0,0,0,0,0,0,0,0,0.30,0,0,0.32,0,0,0,0,0,0,0,0,0,0.29,3.82,\correct,12,\correct,12,\correct,21
15,0,0.01,0,0,0,0,0,0.02,list-inc,0.01,0,0,0,0,0,0,0,0.01,1.73,0,0.01,0,0,0,0,0,0,0,0,0,1.67,0,0,0,0,0,0,0.01,0,0,0.01,0,0,0,1.68,0,0,1.68,0,0,0,0,0,0,0,0,0.01,0.48,5.30,\correct,20,\correct,8,\correct,64
16,0,0,0,0,0.03,0,0,0.01,list-last,0,0,0.11,0,0,0,0,0,0,0.40,0,0.01,0.02,0,0,0,0,0,0,0,0,0.57,0,0,0,0,0,0,0.03,0,0,0,0,0,0,0.42,0.07,0,0.43,0,0,0.01,0,0,0,0,0,0.01,\incorrect,12.30,\correct,26,\correct,20,\na,\na
17,0,0,0,0,0,0,0,0.01,list-length,0,0,0,0,0,0,0,0,0,0.37,0,0,0,0,0,0,0,0,0,0,0,0.37,0,0,0,0,0,0,0,0,0,0,0,0,0,0.36,0,0,0.40,0,0,0,0,0,0,0,0,0.01,\incorrect,8.45,\correct,15,\correct,13,\na,\na
18,0,0.02,0,0,0.01,0,0,0.17,list-map,0.02,0,0.04,0.01,0,0,0,0.01,0.02,2.72,0,0.10,0.02,0,0,0,0,0,0,0.01,0,2.90,0,0,0,0.01,0,0,0.13,0,0,0.02,0,0,0,2.74,0.04,0,2.87,0.01,0.01,0.01,0,0,0,0,0.01,0.13,\incorrect,\incorrect,\correct,35,\na,\na,\na,\na
19,0,0.01,0,0.01,0.09,0,0,0.10,list-nth,0.01,0,0.24,0.01,0,0.02,0,0.01,0.01,1.12,0,0.05,0.05,0,0,0.02,0.01,0.03,0,0.02,0,2.24,0,0.01,0,0.02,0,0,0.66,0,0,0.01,0.01,0,0,1.19,0.18,0,1.13,0.01,0.01,0.02,0,0,0,0,0,0.07,\incorrect,\incorrect,\correct,35,\na,\na,\na,\na
20,0,0.04,0,0.01,0.59,0,0,0.26,list-pairwise-swap,0.04,0,2.14,0.01,0,0.02,0,0.01,0.07,0.89,0.01,0.18,0.08,0,0,0.02,0,0.02,0,0.03,0,15.26,0,0.02,0,0.02,0,0,8.58,0,0,0.05,0,0,0,1.02,1.64,0,0.90,0.03,0.02,0.02,0,0,0,0,0,0.24,\incorrect,37.27,\correct,53,\incorrect,36,\na,\na
21,0,0.02,0,0,0,0,0,0.09,list-rev-append,0.02,0,0,0.01,0,0.01,0,0,0.02,1.45,0,0.08,0,0,0,0.01,0,0,0,0.01,0,1.42,0,0.01,0,0.01,0,0,0.07,0,0,0.02,0,0,0,1.45,0.01,0,1.53,0.01,0.01,0,0,0,0,0,0,0.08,\incorrect,12.87,\incorrect,22,\correct,16,\na,\na
22,0,0.01,0,0,0,0,0,0.03,list-rev-fold,0.02,0,0,0,0,0,0,0,0.01,1.98,0,0.03,0,0,0,0,0,0,0,0,0,1.94,0,0,0,0,0,0,0.03,0,0,0.01,0,0,0,1.91,0,0,1.97,0,0,0,0,0,0,0,0,0.03,\incorrect,13.87,\correct,10,\correct,15,\na,\na
23,0,0.01,0,0,\incorrect,0,0,0.04,list-rev-snoc,0.01,0,\incorrect,0,0,0,0,0,\incorrect,1.47,0,0.03,\incorrect,0,0,0,0,\incorrect,0,0,0,\incorrect,0,0,\incorrect,0,0,0,\incorrect,0,\incorrect,0.01,0,0,0,1.45,\incorrect,\incorrect,1.43,0.01,0,\incorrect,0,0,\incorrect,\incorrect,0,0.03,\incorrect,13.57,\incorrect,18,\incorrect,14,\na,\na
24,\incorrect,\incorrect,\incorrect,\incorrect,\incorrect,\incorrect,\incorrect,\incorrect,list-rev-tailcall,\incorrect,\incorrect,\incorrect,\incorrect,\incorrect,\incorrect,\incorrect,\incorrect,\incorrect,\incorrect,\incorrect,\incorrect,\incorrect,\incorrect,\incorrect,\incorrect,\incorrect,\incorrect,\incorrect,\incorrect,\incorrect,\incorrect,\incorrect,\incorrect,\incorrect,\incorrect,\incorrect,\incorrect,\incorrect,\incorrect,\incorrect,\incorrect,\incorrect,\incorrect,\incorrect,\incorrect,\incorrect,\incorrect,\incorrect,\incorrect,\incorrect,\incorrect,\incorrect,\incorrect,\incorrect,\incorrect,\incorrect,\incorrect,\incorrect,\incorrect,\na,\na,\na,\na,\na,\na
25,0,0.01,0,0,1.36,0,0,0.07,list-snoc,0.01,0,0.77,0,0,0,0,0,0.03,0.83,0,0.07,1.11,0,0,0,0,0.93,0,0,0,6.38,0,0,0,0,0,0,0.38,0,0,0.01,0,0,0,0.87,1.75,0,0.87,0,0,0.44,0,0,0,0,0,0.03,0.40,27.68,\incorrect,33,\correct,22,\correct,54
26,0.09,0.04,0.15,0.38,\incorrect,0,0.15,0.18,list-sort-sorted-insert,0.04,0,\incorrect,0.17,0.17,0.41,0,0.39,\incorrect,4.76,0.10,0.12,\incorrect,0,0,0.19,0.15,\incorrect,0,0.43,0,\incorrect,0,0.16,\incorrect,0.17,0,0,\incorrect,0,\incorrect,0.05,0.15,0,0,4.08,\incorrect,\incorrect,11.79,1.15,0.43,\incorrect,0,0,\incorrect,\incorrect,0.48,0.20,\incorrect,7.53,\correct,20,\correct,15,\na,
27,\incorrect,0.03,0.01,0.05,\incorrect,0,0.01,0.56,list-sorted-insert,0.03,0,\incorrect,0.06,\incorrect,0.17,0,0.04,\incorrect,\incorrect,\incorrect,0.43,\incorrect,0,0,0.09,0.02,\incorrect,\incorrect,0.14,0,\incorrect,0,0.08,\incorrect,0.08,0,\incorrect,\incorrect,\incorrect,\incorrect,\incorrect,0.02,0,0,4.38,\incorrect,\incorrect,4.78,\incorrect,\incorrect,\incorrect,\incorrect,\incorrect,\incorrect,\incorrect,\incorrect,\incorrect,\incorrect,\incorrect,\correct,63,\na,\na,\na,\na
28,0.07,0.36,0.06,0.19,\incorrect,0,0.06,1.22,list-stutter,0.35,0,\incorrect,0.06,0.06,0.20,0,0.19,\incorrect,3.03,0.07,1.16,\incorrect,0,0,0.06,0.22,\incorrect,0,0.19,0,\incorrect,0,0.22,\incorrect,0.22,0,0,\incorrect,0,\incorrect,0.37,0.22,0,0,3.17,\incorrect,\incorrect,3.11,0.29,0.10,\incorrect,0,0,\incorrect,\incorrect,0.17,1.22,0.44,8.36,\correct,25,\correct,18,\correct,42
29,0,0.01,0,0,0,0,0,0.02,list-sum,0.01,0,0,0,0,0,0,0,0.01,17.35,0,0.01,0,0,0,0,0,0,0,0,0,1.04,0,0,0,0,0,0,0.02,0,0,0.01,0,0,0,1.08,0,0,1.04,0,0,0,0,0,0,0,0,0.01,0.44,12.15,\correct,10,\correct,13,\correct,72
30,\incorrect,0.02,0.48,1.41,\incorrect,0,0.41,0.16,list-take,0.01,0,\incorrect,1.40,\incorrect,4.43,0,1.41,\incorrect,\incorrect,\incorrect,0.14,\incorrect,0,0,1.30,1.29,\incorrect,\incorrect,4.48,0,\incorrect,0,4.46,\incorrect,4.23,0,\incorrect,\incorrect,\incorrect,\incorrect,\incorrect,1.50,0,0,11,\incorrect,\incorrect,11.02,\incorrect,\incorrect,\incorrect,\incorrect,\incorrect,\incorrect,\incorrect,\incorrect,\incorrect,0.47,\incorrect,\incorrect,37,\incorrect,\na,\correct,39
31,0,0,0,0,0,0,0,0,list-tl,0,0,0,0,0,0,0,0,0,0.37,0,0,0,0,0,0,0,0,0,0,0,0.31,0,0,0,0,0,0,0,0,0,0,0,0,0,0.36,0,0,0.36,0,0,0,0,0,0,0,0,0,0.30,2.55,\correct,11,\correct,11,\correct,10
32,\incorrect,\incorrect,\incorrect,\incorrect,\incorrect,\incorrect,\incorrect,\incorrect,nat-add,\incorrect,\incorrect,\incorrect,\incorrect,\incorrect,\incorrect,\incorrect,\incorrect,\incorrect,\incorrect,\incorrect,\incorrect,\incorrect,\incorrect,\incorrect,\incorrect,\incorrect,\incorrect,\incorrect,\incorrect,\incorrect,\incorrect,\incorrect,\incorrect,\incorrect,\incorrect,\incorrect,\incorrect,\incorrect,\incorrect,\incorrect,\incorrect,\incorrect,\incorrect,\incorrect,\incorrect,\incorrect,\incorrect,\incorrect,\incorrect,\incorrect,\incorrect,\incorrect,\incorrect,\incorrect,\incorrect,\incorrect,\incorrect,\incorrect,19.38,\na,\na,\correct,9,\na,\na
33,0,0,0,0,0,0,0,0,nat-iseven,0,0,0,0,0,0,0,0,0,0.05,0,0,0,0,0,0,0,0,0,0,0,0.05,0,0,0,0,0,0,0,0,0,0,0,0,0,0.05,0,0,0.03,0,0,0,0,0,0,0,0,0,0.34,9.11,\correct,20,\correct,16,\correct,21
34,\incorrect,\incorrect,\incorrect,\incorrect,\incorrect,\incorrect,\incorrect,\incorrect,nat-max,\incorrect,\incorrect,\incorrect,\incorrect,\incorrect,\incorrect,\incorrect,\incorrect,\incorrect,\incorrect,\incorrect,\incorrect,\incorrect,\incorrect,\incorrect,\incorrect,\incorrect,\incorrect,\incorrect,\incorrect,\incorrect,\incorrect,\incorrect,\incorrect,\incorrect,\incorrect,\incorrect,\incorrect,\incorrect,\incorrect,\incorrect,\incorrect,\incorrect,\incorrect,\incorrect,\incorrect,\incorrect,\incorrect,\incorrect,\incorrect,\incorrect,\incorrect,\incorrect,\incorrect,\incorrect,\incorrect,\incorrect,\incorrect,0.56,25.48,\na,\na,\incorrect,15,\correct,26
35,0,0,0,0,0,0,0,0,nat-pred,0,0,0,0,0,0,0,0,0,0.03,0,0,0,0,0,0,0,0,0,0,0,0.03,0,0,0,0,0,0,0,0,0,0,0,0,0,0.02,0,0,0.02,0,0,0,0,0,0,0,0,0,0.31,2.77,\correct,9,\correct,10,\correct,16
36,\incorrect,\incorrect,\incorrect,\incorrect,\incorrect,\incorrect,\incorrect,\incorrect,tree-binsert,\incorrect,\incorrect,\incorrect,\incorrect,\incorrect,\incorrect,\incorrect,\incorrect,\incorrect,\incorrect,\incorrect,\incorrect,\incorrect,\incorrect,\incorrect,\incorrect,\incorrect,\incorrect,\incorrect,\incorrect,\incorrect,\incorrect,\incorrect,\incorrect,\incorrect,\incorrect,\incorrect,\incorrect,\incorrect,\incorrect,\incorrect,\incorrect,\incorrect,\incorrect,\incorrect,\incorrect,\incorrect,\incorrect,\incorrect,\incorrect,\incorrect,\incorrect,\incorrect,\incorrect,\incorrect,\incorrect,\incorrect,\incorrect,\incorrect,\incorrect,\na,\na,\na,\na,\na,\na
37,0.01,0.02,0,0.02,0.02,0,0,0.07,tree-collect-leaves,0.03,0,0.02,0.01,0,0.02,0,0.02,0.03,8.26,0.01,0.05,0.04,0,0,0.01,0.02,0.01,0,0.02,0,8.53,0,0.02,0,0.02,0,0,0.05,0,0,0.02,0.02,0,0,8.46,0.05,0,8.29,0.03,0.01,0.01,0,0,0,0,0.01,0.06,\incorrect,9.01,\correct,29,\correct,20,\na,\na
38,0.01,0,0.01,0.01,0.02,0,0.01,0.02,tree-count-leaves,0,0,0.01,0.01,0.01,0.02,0,0.01,0,8.13,0.01,0.01,0.01,0,0,0.01,0.02,0.02,0,0.02,0,8.16,0,0.02,0,0.03,0,0,0.01,0,0,0,0.02,0,0,8.18,0.02,0,8.16,0.03,0.02,0.01,0,0,0,0,0.02,0.02,9.77,5.71,\correct,24,\correct,18,\incorrect,75
39,0.39,0,0.01,0.56,0.02,0,0.45,0.04,tree-count-nodes,0.01,0,0.02,0.01,0.41,0.61,0,0.01,0.01,14.33,0.41,0.02,0.01,0,0,0.59,0.02,0.01,0,0.03,0,11.58,0,0.54,0,0.03,0,0,0.02,0,0,0,0.51,0,0,11.26,0.02,0,13.49,1.50,0.95,0.01,0,0,0,0,0.69,0.04,16.26,5.31,\correct,24,\correct,18,\incorrect,71
40,0.78,0.38,1.83,35.67,0.02,0,1.56,1.30,tree-inorder,0.32,0,0.01,3.32,1.84,37.71,0,36.66,0.33,64.37,1.32,0.58,0.05,0,0,5.35,3.17,0.01,0,41.03,0,15,0,4.79,0,7.37,0,0,0.39,0,0,0.30,2.73,0,0,68.08,0.05,0,64.82,41.54,6.33,0.01,0,0,0,0,17.83,1.06,\incorrect,5.86,\correct,29,\correct,20,\na,\na
41,0.01,0.06,0.01,0.07,0.30,0,0.01,0.85,tree-map,0.06,0,4.06,0.03,0.01,0.09,0,0.05,0.07,12.09,0.02,0.29,1.81,0,0,0.04,0.01,0.11,0,0.09,0,21.22,0,0.02,0,0.03,0,0,0.55,0,0,0.06,0.01,0,0,12.32,3.37,0,11.79,0.26,0.07,0.40,0,0,0,0,0.09,0.62,\incorrect,\incorrect,\correct,47,\na,\na,\na,\na
42,2.50,0.01,4,10.80,\incorrect,0,4.28,0.14,tree-nodes-at-level,0.01,0,\incorrect,7.91,3.73,14.11,0,11.03,\incorrect,51.54,3.55,0.06,\incorrect,0,0,15.39,4.27,\incorrect,0,18.75,0,\incorrect,0,6.90,\incorrect,10.05,0,0,\incorrect,0,\incorrect,0.01,4.35,0,0,43.10,\incorrect,\incorrect,42.73,26.76,15.37,\incorrect,0,0,\incorrect,\incorrect,6.16,0.13,\incorrect,\incorrect,\correct,49,\na,\na,\na,\na
43,0.99,0.19,1.04,3.96,2.27,0,0.97,0.99,tree-postorder,0.17,0,8.42,1.27,1.22,4.65,0,4.02,0.25,45.60,1.08,0.46,14.67,0,0,1.36,1.55,0.96,0,5.21,0,49.24,0,1.76,0,1.91,0,0,1.04,0,0,0.18,1.57,0,0,26.62,23.35,0,23.21,25.66,4.21,0.94,0,0,0,0,10.59,0.93,\incorrect,10.42,\correct,34,\correct,23,\na,\na
44,0,0.18,0.38,1.64,\incorrect,0,0.41,0.80,tree-preorder,0.19,0,\incorrect,0.44,0,1.85,0,1.81,\incorrect,13.30,0,0.43,\incorrect,0,0,0.48,0.62,\incorrect,0,2.03,0,\incorrect,0,0.61,\incorrect,0.71,0,0,\incorrect,0,\incorrect,0.01,0.53,0,0,19.11,\incorrect,\incorrect,18.99,0.01,0.01,\incorrect,0,0,\incorrect,\incorrect,0,0.03,\incorrect,8.12,\correct,29,\correct,20,\na,\na
\end{filecontents*}
\csvreader[head to column names]{postconditional.csv}{}% use head of csv as column names
    {\\\hline\Test & \ComputationTime & \Correct & \Size & \LeonTime & \LeonCorrect 
    %& \LeonSize
    & \SynquidTime & \SynquidCorrect & \SynquidSize }
%    \textbf{Summary} & \needsrev{6.26} & 93\% & 12.10 & 85\% & 1.53 & 87\%
    \end{tabular}
    \caption{ The results of running \Burst, \Synquid, and \Leon on the \textbf{Logical} benchmark suite. The result ``\incorrect'' under ``Time'' indicates failure (in this case, timeout). Under the ``Correct?'' column,``\correct'' (resp. ``\incorrect'') indicates that the synthesized program was (resp. not) the intended one. The column labeled ``Size'' shows the size of the synthesized program. }
    \label{fig:logical-table}
\end{figure}

\paragraph*{Failure Analysis for \toolname.} The two unique failure cases for \toolname are nat-add and nat-max. For the first one, given the unary encoding of two natural numbers, the goal is to add them, and for the second one, the goal is to return the maximum. While these benchmarks look very easy at first glance, \toolname fails on them because the specification is very under-constrained. For example, the specification of nat-max only states that the output should be greater than or equal to both inputs, but under the angelic semantics of recursion, this results in many possible outputs of the recursive calls and causes a blow-up. This is a common theme for \toolname across all types of specifications: Since it constructs a version space based on angelic semantics, synthesis becomes more difficult when the specification is ``loose" for arguments used in recursive calls.

\paragraph{Behavior of \Synquid.} At first glance, \Synquid seems to perform surprisingly poorly on these benchmarks.  Upon further investigation, we found that \Synquid is only able to successfully synthesize programs from highly stylized specifications. For example, consider our specification for the list-compress benchmark shown in Figure~\ref{fig:synquid-lc1}. While \Synquid is unable to synthesize a program from this specification within the given time limit, it \emph{can} synthesize the desired program from the specification shown in Figure~\ref{fig:synquid-lc2}. These specifications are semantically equivalent; however, \Synquid's behavior on them is very different. Thus, while it may be possible to re-engineer our specifications so that \Synquid performs successful synthesis, coming up with specifications that are \Synquid-friendly is a highly non-trivial task.

%Fundamentally, Synquid's specification is no longer merely a function from Lists to Lists. Their ``PList''s have both elements as well as predicates embedded within, and their specification cannot be written as a simple postcondition over a function from lists to lists.

\begin{figure}
    \begin{lstlisting}[style=smallstyle]
data List where
  Nil :: List
  Cons :: Nat -> List -> List

measure heads :: List -> Set Nat where
  Nil -> []
  Cons x xs -> [x]
  
measure no_adjacent_dupes :: List -> Bool where
  Nil -> True 
  Cons x xs -> !(x in heads xs) && no_adjacent_dupes xs

compress :: xs: List a -> {List | elems xs == elems _v && no_adjacent_dupes _v}
compress = ??\end{lstlisting}
    \caption{Our list-compress benchmark specification for \Synquid.}
    \label{fig:synquid-lc1}

\end{figure}

\begin{figure}
\begin{lstlisting}[style=smallstyle]
data PList a <p :: a -> PList a -> Bool> where
	Nil :: PList a <p>
	Cons :: x: a -> xs: {PList a <p> | p x _v} -> PList a <p>
  
measure heads :: PList a -> Set a where
  Nil -> []
  Cons x xs -> [x]

type List a = PList a <{True}>  
type CList a = PList a <{!(_0 in heads _1)}>

compress :: xs:List a -> {CList a | elems xs == elems _v}
compress = ??\end{lstlisting}
\caption{Alternative list-compress benchmark specification for \Synquid.}
\label{fig:synquid-lc2}
\end{figure}

\paragraph*{Comparisons to Leon} \Leon performs better than \Synquid for our specifications; however, it solves 34 benchmarks compared to the 41 of \toolname. Overall, \Leon tends to perform relatively poorly on benchmarks with nontrivial branching (e.g., list-compress and tree-binsert). This behavior is likely due to their condition abduction procedure  failing to infer the correct branch conditions when they are deeply nested. On the higher-order benchmarks (list-filter, list-map), \Leon either reports an error or returns a wrong solution that does not satisfy the provided logical specification. 
%(besides list-filter, which \Leon reports an error on) \Leon immediately terminates with a simple incorrect answer. For example, on list-map it simply returns \lstinline{let rec P(x) = x}. Lastly, due to it's baked-in recursion schemes, \Leon is unable to synthesize the functions that require complex recursion, like list-rev-append and list-pairwise-swap.

%\Leon performs particularly well in situations where the branch conditions are simple, but the branches themselves are complex. Leon is able to synthesize the correct solutions to the tree-ordering benchmarks quite quickly, and does not seem to  have a substantially harder time with tree-postorder than tree-preorder. Unlike \Smyth, which simply enumerates candidate expressions top-down refinement fails, \Leon uses a combination approach where they first find a representation for their search space, then intelligently prune groups of candidates with concrete counterexamples. This lets them avoid hitting the enumeration wall as quickly as \Smyth or \Synquid does when their top-down refinements fail.

\vspace{5pt}
\noindent
\fbox{\parbox{.95\linewidth}{
	{\bf Result \#3}: \Burst is able to synthesize  91\% of the benchmarks from logical specifications and solves more benchmarks than both \Leon (76\%) and \Synquid (49\%).
  }
}

\subsection{Ablation Study for Specification Strengthening}

Our proposed synthesis algorithm combines angelic synthesis with specification strengthening. However, an alternative approach is to perform enumerative search over all programs that angelically satisfy the specification. That is, one could repeatedly sample solutions to the angelic synthesis problem and test whether they satisfy the specification until we exhaust the search space or find the correct program. In this section, we perform an ablation study to evaluate the benefit of specification  strengthening compared to a simpler enumerative search baseline.

The results of this ablation study are presented in Figure~\ref{fig:ablation}, where {\sc Burst}$^\dagger$  is a variant of \toolname that performs basic enumerative search instead of backtracking search with specification strengthening. As we can see from this table, \toolname with specification strengthening solves more benchmarks within the given time limit, and this difference is particularly pronounced for the {\bf IO} and {\bf Ref} specifications. For {\bf Logical} specifications, the difference between \toolname  and {\sc Burst}$^\dagger$ is less stark due to the optimization described in Section~\ref{sec:opt-cegis}.

\begin{figure}[t]
    \centering
    \begin{tabular}{|r|c|c|c|}%
    \hline
    & \textbf{IO} & \textbf{Ref} & \textbf{Logical} \\
    \hline
    \Burst 
    & 96\% %43/45
    & 96\% %43/45
    & 91\% %41/45
    \\\hline
    {\sc Burst}$^\dagger$
    & 73\% %33/45
    & 78\% %35/45
    & 84\% %38/45
    \\\hline
    \end{tabular}
    \caption{Number of benchmarks that can be solved within the time limit for each of the three specifications.}
    \label{fig:ablation}
\end{figure}

\vspace{5pt}
\noindent
\fbox{\parbox{.95\linewidth}{
	{\bf Result \#4}: The variant of \toolname that performs enumerative search over angelic synthesis results solves fewer benchmarks than \toolname (with specification strengthening) for all three specification types. 
  }
}

\section{Related Work}
\label{sec:related-work}

%In this work, we have developed a bottom-up algorithm for synthesis of recursive functions from ground specifications.  

The prior work most related to this paper can be divided into four overlapping categories: (i) 
bottom-up synthesis, 
%The primary pieces of related work lie in
(ii) version-space-based synthesis, (iii) synthesis of functional recursive programs, and (iv) synthesis based on angelic semantics. 
%and (iv) deductive synthesis techniques that operate through classifications of a value space. 
Now we elaborate on these categories of work. 
For a broader survey of program synthesis, see \citet{ps-now}. 

\paragraph*{Bottom-up Synthesis} Bottom-up enumeration is a classic approach to program synthesis. A canonical example is \ToolText{Transit} \citep{transit}. \ToolText{Transit} grows a pool of programs of increasing complexity, ensuring that no program in the pool is \emph{observationally equivalent} to another program in the pool. 
The \ToolText{Stun}~\citep{stun} approach generalizes this method by decomposing the input-output specification into multiple parts, synthesizing programs that work for these sub-specifications in a bottom-up way, then combining these programs using a form of anti-unification. \ToolText{Bustle}~\citep{odena2020bustle} offers another generalization, using a learning algorithm to guide bottom-up exploration. However, none of these methods handle programs with general recursion. 

\ToolText{Escher}~\citep{escher} is a bottom-up inductive synthesis approach that handles recursion. The algorithm here combines a forward search, in which terms are generated bottom-up, with a procedure for inferring conditional statements. However, a key limitation of this approach is that it requires a trace-complete specification to handle recursive calls. 

%Our angelic synthesis approach offers a path to extending these ideas to (functional) recursive programs without th. 

\paragraph*{Version-Space-Based Synthesis} Version space approaches to synthesis use an efficient data structure to represent the set of all programs that satisfy a specification. 
Early techniques~\citep{lau2003learning} proved hard to scale. 
\ToolText{FlashFill}~\cite{flashfill}, which represented version spaces using a form of e-graphs~\citep{Downey:1980}, was a major leap forward.
%permit efficient representations of sets of possible programs, and allow for intersections of such representations. 
%, which automatically identifies text transformations. 
\ToolText{FlashFill}'s success led to followup methods, for example, \ToolText{FlashExtract}~\citep{flashextract}, \ToolText{FlashRelate} \citep{flashrelate}, and \ToolText{Refazer}~\citep{refazer}. 
This line of work culminated in \ToolText{FlashMeta} \citep{flashmeta}, a framework for version-space-based synthesis that supports the above methods as instantiations.
Unlike \Burst, these methods all construct versions spaces top-down. They work backward from the desired output for a specific input, iteratively producing subgoals describing the unknown parts of the target program, and then construct version spaces for these subgoals.

%techniques were then used similarly in future work on text extraction~\cite{flashextract}, pattern recognition~\cite{flashprofile}, program refactoring~\cite{refazer}, and more. 
%abstracts this approach into a general-purpose framework, where these previously mentioned version-space techniques are instantiations of this framework.

The use of tree automata (FTAs) as a version space representation was first explored by \citet{madhusudan2011synthesizing} in the setting of reactive synthesis. Subsequently, the \ToolText{Dace} \citep{dace}, \ToolText{Relish} \citep{relish}, and \ToolText{Blaze} \citep{blaze} systems used tree automata to represent version spaces in the synthesis of functional programs.
Like \Burst, these three approaches proceed bottom-up:
%destruct the desired outputs into sub-components. 
%However, there is also a category of methods that combine version-space-based and bottom-up synthesis.
%learning performs a bottom-up search using FTAs. 
rather than starting from the desired outputs and producing subgoals for incomplete programs, they start from the input values and propagate these inputs through a space of programs. 
However, for reasons explained earlier, these methods cannot handle general recursion. 
%Originally applied to synthesizing data completion scripts~\cite{dace}, this approach retains many of the benefits of the FlashMeta family of tools while not requiring inverse semantics. FTA-based synthesis has since been used as a framework for relational program synthesis~\cite{relish} and abstraction refinement in the context of program synthesis~\cite{blaze}.

% Recent work has shown that enumerative techniques and deductive techniques are complementary to each other. Indeed, many deductive synthesizers~\cite{synquid, myth, l2} have enumerative subcomponents for when deduction fails. Other work~\cite{neo, morpheus} has found that tools can deduce generalizeable failures from enumerated candidates. This relationship has been made more explicit in \ToolText{DryadSynth}, which generalizes the interactions between enumerative and deductive techniques to work with arbitrary grammars, and \ToolText{Duet} which demonstrates that tighter integration the two types of synthesizers helps improve overall performance.

\paragraph*{Synthesis of Recursive Programs} The synthesis of (functional) recursive programs has a long history. 
Most methods in this area consider synthesis from examples, but synthesis from richer specifications, such as refinement types, has also been considered. 
The \ToolText{Thesys}~\citep{thesys} and  \ToolText{Igor2}~\citep{igor} systems are two early examples of work of this sort. 
Given a set of examples, these methods first synthesize straight-line programs in a top-down manner, then identify patterns within a given program, then generalize these patterns into a recursive program. 

More recently, the \ToolText{Myth}~\citep{myth} and $\lambda^2$~\cite{lambda2} systems introduced types as a means of directing an inductive synthesis process. %\ToolText{Myth} relied on trace-complete examples to evaluate recursive calls.
\ToolText{Myth2}~\citep{myth2} extended \Myth with more complex types of refinement types, including negative examples, intersection, and union types. 
%$\lambda^2$~\cite{lambda2} addressed a similar problem of synthesizing functional programs using built-in higher-order functions.
All these approaches are top-down; also, all of them, except for $\lambda^2$, rely on trace-completeness to handle recursive calls. While $\lambda^2$ does not assume trace-completeness, it only applies deductive reasoning to limited forms of recursion assuming trace-completeness, and defaults to brute-force enumeration when handling general recursion or on non-trace-complete examples.

\Smyth~\citep{smyth} is a generalization of \Myth that also performs top-down deduction-aided search, but does not have the trace-completeness requirement. 
%Like \Myth, it is an example-based synthesizer, and performs top-down deduction of programs, and like $\lambda^2$, falls back to raw enumeration when synthesis fails. 
To handle non-trace-complete specifications, \Smyth generates partial programs, then propagates constraints from partial programs to the remaining holes. This propagation is both complex and domain-specific, and incomplete. By contrast, \Burst relies on a single generalizable principle of angelic recursion, and is complete.

\ToolText{Leon}~\cite{leon} performs synthesis modulo recursive functions. Leon takes a pre-/post-condition specification and searches for a recursive function that can satisfy it. \ToolText{Leon} solves this task using a term generation engine that produces candidate programs, a condition abduction engine that synthesizes branches, and a verification engine that evaluates candidate solutions. 
\Leon's term generation runs similarly (though not exactly, due to the lack of angelic execution) to our \Burst{}$^\dagger$ ablation; it does not search through programs based on recursive results but simply based on increasing cost. Its condition abduction engine is top-down and quite distinct from ours.

\ToolText{Synquid}~\cite{synquid} synthesizes programs from polymorphic refinement types. Such types are expressive and allow the specification of desired functions in a way that is both compositional and tight.
However, while \Synquid can synthesize nearly any function from a carefully crafted refinement type, there are many kinds of  %\emph{undesirable 
realizable specifications on which it simply gives up. 
%Specifically, when \Synquid fails at verifying a function, it gives up on that function.
%, whereas \Leon falls back to testing, and can potentially return incorrect functions. 
In particular, \Synquid cannot synthesize from specifications that are non-inductive, including many of the specifications in our benchmark suite. 
%Furthermore, our specifications do not make full use of refinement types. 
More generally, \Burst and \Synquid address different problems: \Synquid focuses on always being able to synthesize from a well-written refinement type, while \Burst focuses on best-effort synthesis from arbitrary 
%(and potentially poorly written) 
logical specifications.

The recent \ToolText{Cypress}~\citep{cypress} system targets synthesis of recursive programs from separation logic specifications. 
\ToolText{Cypress} generates a satisfying straight-line program, then ``folds'' that program into a generalized recursive procedure. 
We attempted a similar strategy in our setting but found the space of possible foldings to be prohibitively large.
%, when trying to eagerly fold part of the FTA onto itself, the search space blew up. 
In contrast, our synthesis algorithm follows a lazier strategy: instead of synthesizing the full search space, then finding ways to fold it together, we overapproximate the search space and then discover ways to refine it.

\paragraph*{Angelic Synthesis}
While angelic non-determinism has been used to expose synthesizers to programmers, there is almost no work on the use of angelic semantics as a core part of a synthesis algorithm. The one approach that we know of is \ToolText{FrAngel}~\cite{shi2019frangel}, which adds control structures to the well-studied problem of component-based synthesis. \ToolText{FrAngel} first identifies candidate partial programs by synthesizing programs with \emph{no} control structures, but instead with \emph{angelic placeholders}, then attempts to place appropriate control structures in place of the angelic placeholders. 
%This strategy has some similarities to that in \Burst, which identifies candidate recursive calls with angelic synthesis. However, 
In contrast, \Burst does not have angelic placeholders but instead  updates the generated code itself to fulfill the requirements put in place by  angelic recursion.

Angelic execution has used in the related field of program repair. SPR~\cite{spr} uses angelic executions to identify candidate locations for condition repair, then instantiates those conditions in a second phase. SPR is similar to FrAngel, as it stages the synthesis into a sketch identification stage (using Angelic Semantics to identify promising sketches), and a sketch completion stage. Angelix~\cite{angelix} generalizes this approach to perform multi-location repairs, through their novel ``angelic forest'' data structure.

% \paragraph*{Value Space Directed Classification Techniques}

% Recent work in invariant synthesis has shown that techniques that are
% \emph{value space directed} (also known as ``data-driven'' techniques) 
% %\afm{Is this a valid way to avoid that term. I hate it.} 
%   can help derive loop
% invariants~\ref{loopinvgen}, solutions to Constrained Horn Clause (CHC) problems~\ref{lineararbitrary},
% and data structure invariants~\ref{hanoi}.

% Broadly, these techniques classify candidate inputs as either \true{} or \false.
% These classifications correspond to whether a given value should be accepted by
% the loop invariant, should be permitted in a given CHC clause, or should be
% accepted by the data structure invariant. With a backtracking search, these
% techniques strengthen the requirements for the underlying synthesizer by
% performing classifications until an invariant is found.

% Our algorithm is similar to, and inspired by, this work. Instead of classifying
% our inputs into \true{} or \false, we classify them to a possible output. We
% refine these classifications until one is found that is compatible with the
% overall specification.

\section{Future Work}

As found in our failure analysis in Section~\ref{sec:evaluation}, \Burst{} has issues with severely underconstrained specifications. This is due to two primary reasons: (1) extensive backtracking and (2) output blowup. To address (1) we believe that additional work in anti-specifications could be helpful. By identifying a more general anti-specification, more parts of the search space are eliminated, necessitating less backtracking. To address (2) we believe that integrating an abstraction refinement algorithm like that used by \textsc{Blaze}~\cite{blaze} could tame blowups in candidate outputs.

\textsc{DryadSynth}~\cite{dryadsynth} and \textsc{Duet}~\cite{duet} have shown that integrating bottom-up and top-down approaches results in synthesizers greater than the sum of their parts, and we believe these findings would generalize to problems involving recursion. In particular, we think there is promising future work in integrating \Burst{} with a top-down recursive synthesizer like \textsc{SMyth}. 

Lastly, there is a large class of important specifications that \Burst{} cannot currently address -- relational specifications. Relational specifications describe the interactions between multiple different program runs. For example, this means that \Burst cannot synthesize programs that are idempotent, as we cannot reduce the postcondition $f(f(x)) = f(x)$ to a ground specification. We think there is promising future work in integrating \Burst{} with existing techniques for synthesizing programs from relational specifications~\cite{relish}.

\section{Conclusion}
\label{sec:conclusion}

In this paper, we presented a new technique for synthesizing recursive functional programs. Our approach differs from prior work in this space as it performs synthesis in a \emph{bottom-up} fashion. Our algorithm first performs \emph{angelic synthesis} wherein recursive calls may return any value consistent with the specification. This result may be spurious, so our method analyzes the assumptions made in angelic executions and gradually strengthens the specification to find the correct program.

We have implemented the proposed algorithm in a tool called \Burst and showed that it can synthesize programs from a variety of specifications, including examples, reference implementations, and logical formulas. Our comparison against three synthesizers (\Smyth, \Leon, and \Synquid) shows that \Burst advances the state-of-the-art in synthesizing recursive functional programs. 

\section*{Acknowledgements}

We thank our anonymous reviewers, our anonymous shepherd, Ben Mariano, and Todd Millstein for their helpful feedback. We thank Michael James, Tristan Knoth, and Nadia Polikarpova for their help with Synquid and the tooling surrounding it. This work is supported in part by NSF Award 1762299, NSF Award 1811865, NSF Award 1918651, DARPA Contract FA8750-20-C-0208, and US Air Force and DARPA Contract FA8750-20-C-0002.

%Our algorithm works by first identifying programs that satisfies the specification under angelic semantics, and then refining the specification until the program also satisfies the specification under standard semantics. We have implemented our algorithm as the tool, \Burst. We used \Burst to synthesize a number of programs from a variety of specifications, typically in under 5 seconds. We are excited about applying existing techniques that improve bottom-up synthesis, like abstraction refinement~\ref{blaze}, and bottom-up/top-down cooperation~\ref{?,?}.

%\printbibliography
\bibliography{ref}

\ifappendices

\appendix

\section{Additional Algorithms}
\label{sec:additional-algs}

\begin{algorithm}
\begin{algorithmic}[ht]
\Statex \Input{Ground specification \groundSpec} 
\Statex \Output{A queue element $(pri,P,\witness,\groundSpec)$, or $\bot$.} 
\Procedure{MakeQE}{\groundSpec}
\State $res \gets \Call{SynthesizeAngelic}{\groundSpec}$
\Match{$res$}
\Case{$\mathsf{Failure}(\_)$}{\Return $\bot$}
\EndCase
\Case{$\mathsf{Success}(P,\witness)$}{}
\State $pri \gets |P|$
\State \Return $(pri,P,\witness,\groundSpec)$
\EndCase
\EndMatch
\EndProcedure
\end{algorithmic}
\begin{algorithmic}
\Statex \Input{Ground specification \groundSpec}
\Statex \Output{A program $P$, or $\bot$.}
\Procedure{Synthesize}{\groundSpec}
\State $PQ \gets$ Empty
\State $PQ.\textsc{Push}(\Call{MakeQE}{\groundSpec})$
\While{PQ $\neq$ Empty}
\State $(\_,P,\witness,\groundSpec) \gets PQ.\textsc{Pop}()$
\If{$P \models \groundSpec$} \Return $P$
\Else
\State $PQ.\textsc{Push}(\Call{MakeQE}{\groundSpec \wedge \witness})$
\State $PQ.\textsc{Push}(\Call{MakeQE}{\groundSpec \wedge \neg\witness})$
\EndIf
\EndWhile
\EndProcedure
\end{algorithmic}
\caption{Recursive Synthesis Algorithm Guaranteeing Optimality}
\label{alg:opt-recursive-synthesis}
\end{algorithm}

\section{Proofs}
\label{sec:proofs}
\subsection{Completeness of \Synthesize}

We prove the completeness of \Synthesize by introducing the following definition and using it to prove the claim.

\begin{definition}
We define $\Respects(e \Downarrow v,\varphi)$ as the following.
\begin{mathpar}
    \inferrule
    {
      %\Respects(e_1 \Downarrow v_1, \varphi)\\
      \Respects(e_2 \Downarrow v_2, \varphi)\\
      \Respects(e_1 [v_2/x] \Downarrow v_3, \varphi)\\
      {\sf SAT}(\varphi\wedge f(v_2) = v_3)
      %{\sf SAT}(f(v_1) = v_2 \wedge \varphi)
    }
    {
      \Respects((\texttt{rec}\ f(x)=e_1)\ e_2 \Downarrow v_3, \varphi)
      %\Respects(f(e_1) \Downarrow v_2, \varphi)
    }\\
    
    \inferrule
    {
    }
    {
      \Respects(\texttt{unit} \Downarrow \texttt{unit}, \varphi)
    }
    
    \inferrule
    {
      \Respects(e_1 \Downarrow v_1,\varphi)\\
      \Respects(e_2 \Downarrow v_2,\varphi)\\
    }
    {
      \Respects((e_1,e_2) \Downarrow (v_1,v_2),\varphi)
    }\\
    
    \inferrule
    {
      \Respects(e \Downarrow (v_1,v_2), \varphi)
    }
    {
      \Respects(\texttt{fst } e \Downarrow v_1, \varphi)
    }
    
    \inferrule
    {
      \Respects(e \Downarrow (v_1,v_2), \varphi)
    }
    {
      \Respects(\texttt{snd } e \Downarrow v_2, \varphi)
    }\\
    
    \inferrule
    {
      \Respects(e \Downarrow v, \varphi)
    }
    {
      \Respects(\texttt{inl } e \Downarrow \texttt{inl } v, \varphi)
    }
    
    \inferrule
    {
      \Respects(e \Downarrow v, \varphi)
    }
    {
      \Respects(\texttt{inr } e \Downarrow \texttt{inr } v, \varphi)
    }\\
    
    \inferrule
    {
      \Respects(e \Downarrow \texttt{inl } v, \varphi)
    }
    {
      \Respects(\texttt{unl } e \Downarrow v, \varphi)
    }
    
    \inferrule
    {
      \Respects(e \Downarrow \texttt{inr } v, \varphi)
    }
    {
      \Respects(\texttt{unr } e \Downarrow v, \varphi)
    }\\
    
    \inferrule
    {
      \Respects(e_3 \Downarrow \texttt{inl } v_3, \varphi)\\
      \Respects(e_1 \Downarrow v_1, \varphi)
    }
    {
      \Respects(\texttt{switch } e_3 \texttt{ on inl } x_1 \to e_1 \texttt{ inr } x_2 \to e_2
        \Downarrow v_1, \varphi)
    }
    
    \inferrule
    {
      \Respects(e_3 \Downarrow \texttt{inr } v_3, \varphi)\\
      \Respects(e_2 \Downarrow v_2, \varphi)
    }
    {
      \Respects(\texttt{switch } e_3 \texttt{ on inl } x_1 \to e_1 \texttt{ inr } x_2 \to e_2
        \Downarrow v_2, \varphi)
    }
\end{mathpar}
\end{definition}

The following lemma relates the semantics in \autoref{fig:angelic-semantics} and the \Respects relation.
\begin{lemma}\label{lemma:respects}
$\Respects(e \Downarrow v,\varphi) \implies e \Downarrow^\varphi v$
\end{lemma}
\begin{proof}
We proceed by structural induction on the evaluation of \Respects.

Case {\sc Unit}:
%Suppose that the last rule applies was $\Respects(\mlstinline{unit} \Downarrow \mlstinline{unit}, \varphi)$.
By the angelic semantics it always holds that $\texttt{unit} \Downarrow^\varphi \texttt{unit}$.

Case {\sc Tuple Construction}:
%$e = (e_1, e_2)$.
Suppose the last rule applied was
\[
    \inferrule
    {
      \Respects(e_1 \Downarrow v_1,\varphi)\\
      \Respects(e_2 \Downarrow v_2,\varphi)
    }
    {
      \Respects((e_1,e_2) \Downarrow (v_1,v_2),\varphi)
    }
\]

By IH, $e_1 \Downarrow^\varphi v_1$ and $e_2 \Downarrow^\varphi v_2$.

By the angelic semantics, we have the following derivation:
\[
    \inferrule
    {
      e_1 \Downarrow^\varphi v_1\\
      e_2 \Downarrow^\varphi v_2
    }
    {
      (e_1,e_2) \Downarrow^\varphi (v_1,v_2)
    }
\]

Therefore, $(e_1,e_2) \Downarrow^\varphi (v_1,v_2)$.

Case {\sc Function Application}:
%$e = (\mlstinline{rec f}(x)=e_1)\ e_2$.
Suppose the last rule applied was:
\[
    \inferrule
    {
      \Respects(e_2 \Downarrow v_2,\varphi)\\
      \Respects(e_1[\texttt{rec}\ f(x)=e_1/f,v_2/x] \Downarrow v_3,\varphi)\\
      \mathsf{SAT}(\varphi\wedge f(v_2) = v_3)
    }
    {
      \Respects((\texttt{rec}\ f(x)=e_1)\ e_2 \Downarrow v_3,\varphi)
    }
\]
By IH, $e_2 \Downarrow^\varphi v_2$.
It also holds that $\mathsf{SAT}(\varphi \wedge f(v_2) = v_3)$.

By the angelic semantics, we have the following derivation:
\[
    \inferrule
    {
      e_2 \Downarrow^\varphi v_2\\
      \mathsf{SAT}(\varphi\wedge f(v_2) = v_3)
    }
    {
      f(e_2) \Downarrow^\varphi v_3
    }
\]
Therefore, $f(e_2) \Downarrow^\varphi v_3$.

The rest of the rules follow similarly.

\end{proof}

\begin{lemma}\label{lemma:rec}
If $P \models \varphi$ and $P \gets \texttt{rec }f(x) = e$, then for all $v, v'$ where $v' \in \SemanticsOf{P}(v)$, it holds that $\Respects(P(v) \Downarrow v',\varphi)$.
\end{lemma}
\begin{proof}
%By induction on the derivation of $v' \in $ \SemanticsOf{P}($v$).
By induction on the derivation of $P(v) \Downarrow v'$.

Case {\sc Unit}: If $P(v) = \mlstinline{unit}$, then $\mlstinline{unit} \Downarrow \mlstinline{unit}$. It follows that $\Respects(\mlstinline{unit} \Downarrow \mlstinline{unit}, \varphi)$.

Case {\sc Tuple Construction}:
Suppose the last rule applied was
\[
    \inferrule
    {
      e_1 \Downarrow v_1\\
      e_2 \Downarrow v_2\\
    }
    {
      (e_1,e_2) \Downarrow (v_1,v_2)
    }
\]

By IH, $\Respects(e_1 \Downarrow v_1, \varphi)$ and $\Respects(e_2 \Downarrow v_2, \varphi)$.

By the semantics of \Respects, we have the following derivation:
\[
    \inferrule
    {
      \Respects(e_1 \Downarrow v_1, \varphi)\\
      \Respects(e_2 \Downarrow v_2, \varphi)
    }
    {
      \Respects((e_1,e_2) \Downarrow (v_1,v_2), \varphi)
    }
\]
Therefore, $\Respects((e_1,e_2) \Downarrow (v_1,v_2), \varphi)$.

Case {\sc Function Application}:
Suppose the last rule applied was:
\[
    \inferrule
    {
      e_2 \Downarrow v_2\\
      e_1[\texttt{rec f}(x)=e_1/f,v_2/x],v_2/x] \Downarrow v_3\\
    }
    {
      (\texttt{rec f}(x)=e_1)\ e_2 \Downarrow v_3
    }
\]
By IH, $\Respects(e_2 \Downarrow v_2, \varphi)$ and $\Respects(e_1[\texttt{rec f}(x)=e_1/f,v_2/x] \Downarrow v_3, \varphi)$.
It also holds that $\mathsf{SAT}(\varphi \wedge f(v_2) = v_3)$ which follows from the fact that $P \models \varphi$.

By the semantics of \Respects, we have the following derivation:
\[
    \inferrule
    {
      \Respects(e_2 \Downarrow v_2,\varphi)\\
      \Respects(e_1[\texttt{rec}\ f(x)=e_1/f,v_2/x] \Downarrow v_3,\varphi)\\
      \mathsf{SAT}(\varphi\wedge f(v_2) = v_3)
    }
    {
      \Respects((\texttt{rec}\ f(x)=e_1)\ e_2 \Downarrow v_3,\varphi)
    }
\]
Therefore, $\Respects(f(e_2) \Downarrow v_3, \varphi)$.

The rest of the rules follow similarly.

\iffalse
Suppose that $P \models \varphi$ and $P \gets \texttt{rec }f(x) = e$. Let $v, v'$ be values such that $v' \in \SemanticsOf{P}(v)$.

By the angelic semantics, $v' \in \SemanticsOf{P(v)}$ and $P \gets \texttt{rec } f(x) = e$ implies that $e[v/x] \Downarrow^\varphi v'$ by inversion of the last rule in \autoref{fig:angelic-semantics}.
Therefore, it holds that $\Respects(e[v/x] \Downarrow v',\varphi)$.
\fi
\end{proof}

\begin{lemma}\label{lemma:semantics}
If $P \models \varphi$ then then $P \angelicmodels \varphi$
\end{lemma}
\begin{proof}
By the previous two lemmas, the conclusion follows.

Let $v$ be a value. We want to show there exists a value $v'$ such that $v' \in [[P]](v)$.

By Lemma \ref{lemma:rec}, $\Respects(e[v/x] \Downarrow v',\varphi)$.

By Lemma \ref{lemma:respects}, $e[v/x] \Downarrow^\varphi v'$, thus $v' \in [[P]](v)$.
\end{proof}

% \autoref{alg:core-recursive-synthesis}
% \autoref{alg:angelic-synthesis}
\begin{theorem}{\bf (Completeness)}
Suppose that (1) \SynthesizeAngelic \ is complete,
and (2) if \SynthesizeAngelic \ returns $\mathsf{Failure}(\kappa)$, then $\kappa$ satisfies the assumption from \autoref{eq:anti-spec}.
If $\Synthesize(\chi)$ returns $\bot$, then there is no program that satisfies $\chi$.
\end{theorem}
\begin{proof}[Proof of \autoref{thm:completeness}]
The statement of the completeness of \SynthesizeAngelic is as follows: If \SynthesizeAngelic$(\chi)$ returns $\mathsf{Failure}(\kappa)$, then there is no program $P$ such that $P \angelicmodels \chi$.

Assume that there is a program $S$ such that $S \models \chi$.
Suppose that \SynthesizeAngelic$(\groundSpec \land \bigwedge_{\phi_i \in \Omega} \neg\phi_i)$ returns $\mathsf{Failure}(\kappa)$.
By assumption (1), there is no program $P$ such that $P \angelicmodels \chi$. By \autoref{lemma:semantics}, there is no program $P$ such that $P \models \chi$.
Thus, $S \not\models \groundSpec$ which contradicts our assumption.
\iffalse
By assumption (2),
$$
\forall P. \ \forall \phi \in \kappa. \ P \models \phi \Rightarrow P \not \models \groundSpec
$$
By assumption (1), $\kappa$ is non-empty, which means that $S \not\models \groundSpec$ which contradicts our assumption.
\fi

Thus, \SynthesizeAngelic$(\groundSpec \land \bigwedge_{\phi_i \in \Omega} \neg\phi_i)$ always returns $\mathsf{Success}(P,\witness)$.
If $P \models \groundSpec$ then the algorithm returns $P \neq \bot$ and the conclusion holds.

Suppose that $P \not\models \groundSpec$.
Then it must be the case that either $\Synthesize(\chi \land \witness) \neq \bot$ or $\Synthesize(\chi \land \neg\witness) \neq \bot$, otherwise it would contradict assumption (1) and the assumption that there is some program $S$ such that $S \models \groundSpec$.
Furthermore, recursive calls do not destroy the invariant that $\Omega$ is an anti-specification due to assumption (2).
Thus, the algorithm returns some program $P \neq \bot$ and the conclusion holds.

Therefore, the theorem holds in all cases.

\iffalse
Suppose that \SynthesizeAngelic$(\groundSpec \land \bigwedge_{\phi_i \in \Omega} \neg\phi_i)$ returns $\mathsf{Failure}(\kappa)$.
Assuming $P$ 

Suppose $\Synthesize(\chi)$ returns $\bot$.
If \SynthesizeAngelic$(\groundSpec \land \bigwedge_{\phi_i \in \Omega} \neg\phi_i)$ returns $\mathsf{Failure}(\kappa)$, then by assumption (2),
$$
\forall P. \ \forall \phi \in \kappa. \ P \models \phi \Rightarrow P \not \models \groundSpec
$$
By assumption (1), $\kappa$ is non-empty in this case, therefore the conclusion holds.

If \SynthesizeAngelic$(\groundSpec \land \bigwedge_{\phi_i \in \Omega} \neg\phi_i)$ returns $\mathsf{Success}(P,\witness)$.
Suppose that $P \models \groundSpec$, then the algorithm would return $P$ which would contradict our conclusion.
However, by assumption (1), since the algorithm returns $P$, it must hold that $P \angelicmodels \groundSpec$, 
$$
\forall P. \ \forall \phi \in \kappa. \ P \models \groundSpec \land \bigwedge_{\phi_i \in \Omega} \neg\phi_i \Rightarrow P \not\models \phi
$$
\fi
\end{proof}
\iffalse
TODO
If \SynthesizeAngelic is complete

Show angelic satisfaction is weaker than the larger one

Suppose $\Synthesize(\chi)$ returns $\bot$ but there is a program $P$ that satisfies $\chi$ under the assumptions above.

Suppose $P$ is $\mlstinline{unit}$.
\fi

\subsection{Completeness of {\sc BuildAngelicFTA}}
\begin{lemma}
If {\sc BuildAngelicFTA}$(v_{in},\varphi)$ returns $\fta$ and $(T,L)$ is an run on $\fta$ where $L(\mathsf{root}(T)) \neq \bot$, then $T \Downarrow^\varphi v$ where
$L(\mathsf{root}(T)) = q_v$
\end{lemma}
\begin{proof}
By induction on the tree $T$.
Assume the root is a node $n$ with children $\set{n_1,\dots,n_k}$.
Proceed by cases on $\textsf{Label}(T)$.

Case {$\textsf{Label}(n) = \mlstinline{x}$}:
By inversion, then $L(n) = q_{v_{in}}$. $x[v_{in}] \Downarrow^\varphi v_{in}$, as desired.

Case {$\textsf{Label}(n) = \mlstinline{unit}$}:
By inversion, then $L(n) = q_{\texttt{unit}}$. $unit[v_{in}] \Downarrow^\varphi \texttt{unit}$, as desired.

Case {$\textsf{Label}(n) = \mlstinline{switch on inl/inr}$}:
By inversion, then either it comes from a left or a right $\bot$ rule. Without loss of generality, assume it's a right $\bot$ rule, since the left side follows a symmetric argument.
Then, $L(n) = q_{v_1}$ and $\mathsf{children}(n) = [n_1,n_2,n_3]$ where $L(n_1) = q_{\mlstinline{inl}\ v_3}$, $L(n_2) = q_{v_1}$, and $L(n_3) = \bot$.

By IH, $n_1[v/x] \Downarrow^\varphi \mlstinline{inl}\ v_3$ and $n_2[v/x] \Downarrow^\varphi v_1$.
Thus, we have the following deduction:
\[
\inferrule*
{
n_1[v/x] \Downarrow^\varphi \mlstinline{inl}\ v_3\\
n_2[v/x] \Downarrow^\varphi v_1
}
{
\mlstinline{switch}\ n_1[v/x]\ \mlstinline{on inl _}\ \rightarrow n_2[v/x] \mid \mlstinline{inr _}\ \rightarrow n_3[v/x]  \Downarrow^\varphi (v_1,v_2)
}
\]

Since we can rewrite $\mlstinline{switch}\ n_1[v/x]\ \mlstinline{on inl _}\ \rightarrow n_2[v/x]\mid \mlstinline{inr _}\ \rightarrow n_3[v/x]$ as $(\mlstinline{switch}\ n_1\ \mlstinline{on inl _} \rightarrow n_2 \mid \mlstinline{inr _} \rightarrow n_3)[v/x]$,
then we have the conclusion as desired.

Case {$\textsf{label}(n) = f$}:
By inversion, then $L(n) = q_{v'}$, and $\mathsf{children}(n) = n'$, where $L(n') = q_{v}$ and $\mathsf{SAT}(\varphi \wedge f(v) = v')$.
By IH, $n'[v/x] \Downarrow^\varphi v$.
So we have the following deduction:
\[
\inferrule*
{
n'[v/x] \Downarrow^\varphi v\\
\mathsf{SAT}(\varphi \wedge f(v) = v')
}
{
f\ n'[v/x] \Downarrow^\varphi v'
}
\]
As $f\ n'[v/x] = (f\ n')[v/x]$, we have $(f\ n')[v/x] \Downarrow^\varphi v'$, as desired.

Case {$\textsf{Label}(n) = (\ast,\ast)$}:
By inversion, then $L(n) = q_{(v_1,v_2)}$, and $\mathsf{children}(n) = [n_1,n_2]$, where $L(n_1) = q_{v_1}$ and $L(n_2) = q_{v_2}$.
By IH, $n_1[v/x] \Downarrow^\varphi v_1$ and $n_2[v/x] \Downarrow^\varphi v_2$.
So we have the following deduction:
\[
\inferrule*
{
n_1[v/x] \Downarrow^\varphi v_1\\
n_2[v/x] \Downarrow^\varphi v_2
}
{
(n_1[v/x],n_2[v/x]) \Downarrow^\varphi (v_1,v_2)
}
\]

$(n_1[v/x],n_2[v/x]) = (n_1,n_2)[v/x]$, so $(n_1,n_2)[v/x] \Downarrow^\varphi (v_1,v_2)$, as desired.

All remaining cases can proceed similarly to $(\ast,\ast)$.
\end{proof}

\begin{theorem}
If {\sc BuildAngelicFTA}$(v_{in}, \varphi)$ returns $\fta$, then $P \angelicmodels_{v_{in}} \varphi$ for every $P \in \mathcal{L}(\fta)$.
\end{theorem}
\begin{proof}[Proof of \autoref{thm:65}]
Assuming that $P$ terminates on every input, this follows from the previous lemma.

\autoref{fig:angelic-semantics}, \autoref{fig:ftacreationrules}
Assume that {\sc BuildAngelicFTA}$(v_{in}, \varphi)$ returns $\fta$ and let $P \in \mathcal{L}(\fta)$.
\[
P \angelicmodels_{v_{in}} \Spec \ \Longleftrightarrow \  \exists v. \ v \in \SemanticsOf{P}^\Spec(v_{in}) \land \mathsf{SAT}(f(v_{in}) = v \wedge \varphi)
\]

Let $P \in \LanguageOf{A}$, then there is a mapping $L$ such that $(P,L)$ is a run and $L(\mathsf{root}(P)) \in Q_f$.
Then by inversion, it follows that $L(\mathsf{root}(P)) = q_v$ for some $v$ where $\mathsf{SAT}(\varphi \wedge f(v_{in}) = v)$.
\end{proof}

\begin{theorem}
The following theorem generalizes this from individual inputs to ground specifications:
Let $\varphi$ be a ground formula such that:
\[
V = \set{ v_i \ | \ f(v_i) \in \mathsf{Terms}(\varphi)}
\]
Then, if {\sc BuildAngelicFTA}$(v_{i}, \varphi)$ returns $\fta_i$ for inputs $V = \{ v_1, \ldots, v_n\}$, then, for every $P \in \mathcal{L}(\fta_1) \cap \ldots \cap \mathcal{L}(\fta_n) $, we have $P \angelicmodels \varphi$.
\end{theorem}
\begin{proof}[Proof of \autoref{thm:66}]
Follows directly from \autoref{thm:65}.
\end{proof}

\fi

\end{document}